\documentclass[acmsmall, screen, authorversion]{acmart}

\usepackage[T1]{fontenc}
\usepackage{graphicx}
\usepackage{wrapfig2}
\usepackage{mathrsfs}
\usepackage{stackrel}
\usepackage{mathtools}
\usepackage{nicefrac}
\usepackage{thmtools}
\usepackage{stmaryrd}
\usepackage{wasysym}
\usepackage{pifont}
\usepackage{array}
\usepackage{cases}
\usepackage{adjustbox}
\usepackage[inline]{enumitem}
\usepackage{url}
\usepackage{doi}
\usepackage{acronym}

\usepackage{framed}
\usepackage{listings}
\usepackage{tabularx}
\usepackage{booktabs}
\usepackage{multirow}

\usepackage{tikz}
\usetikzlibrary{positioning}
\usetikzlibrary{arrows,automata}
\usetikzlibrary{shapes.geometric} 
\usetikzlibrary{shapes.misc} 
\usetikzlibrary{decorations.pathmorphing} 
\usetikzlibrary{matrix}
\usetikzlibrary{calc}

\usepackage{comment}

\usepackage{float}

\newcommand{\slimparagraph}[1]{\paragraph{#1}}

\AtEndPreamble{
\theoremstyle{acmdefinition}
\newtheorem{remark}[theorem]{Remark}}

\let\doteqrel\doteq
\renewcommand*{\doteq}{\mathbin{\doteqrel}}

\newcommand*{\lnextbase}{\ocircle}
\newcommand*{\lbackbase}{\circleddash}
\newcommand*{\lglobbase}{\square}

\newcommand*{\lnextsup}[1]{\mathop{\lnextbase^{#1}}}
\newcommand*{\lbacksup}[1]{\mathop{\lbackbase^{#1}}}

\newcommand*{\lnextgen}[2]{\mathop{\lnextbase^{#1}_{#2}}}

\newcommand*{\lguntil}[4]{#3 \mathbin{\mathcal{U}^{#1}_{#2}} #4}
\newcommand*{\luntil}[3]{#2 \mathbin{\mathcal{U}^{#1}} #3}

\newcommand*{\lglob}[1]{\mathop{\lglobbase^{#1}}}

\newcommand*{\lnext}{\mathop{\lnextbase}}

\newcommand*{\lluntil}[2]{#1 \mathbin{\mathcal{U}} #2}

\newcommand*{\llglob}{\mathop{\lglobbase}}

\newcommand*{\ldnext}{\lnextsup{d}}
\newcommand*{\lunext}{\lnextsup{u}}
\newcommand*{\ldback}{\lbacksup{d}}

\newcommand*{\lcnext}[1]{\mathop{\chi_F^{#1}}}
\newcommand*{\lcdnext}{\lcnext{d}}
\newcommand*{\lcunext}{\lcnext{u}}
\newcommand*{\lcback}[1]{\mathop{\chi_P^{#1}}}

\newcommand*{\lcduntil}[2]{{#1} \mathbin{\mathcal{U}_\chi^d} {#2}}

\newcommand*{\lcuuntil}[2]{{#1} \mathbin{\mathcal{U}_\chi^u} {#2}}

\newcommand*{\lhnext}[1]{\lnextgen{#1}{H}}

\newcommand*{\lhunext}{\lhnext{u}}

\newcommand*{\chain}{\chi}

\newcommand*{\powset}[1]{{\mathcal{P}(#1)}}

\newcommand*{\prf}{\pi}
\newcommand*{\pr}{\mathrel{\prf}}

\newcommand{\cmark}{\ding{51}}
\newcommand{\xmark}{\ding{55}}

\DeclareMathOperator{\pfx}{pf}
\newcommand*{\pvar}[3]{\llbracket {#1}, {#2} \, | \, {#3} \rrbracket}
\newcommand*{\nex}[1]{\llbracket {#1} \, \mathord{\uparrow} \rrbracket}

\newcommand*{\clos}[1]{\operatorname{Cl}({#1})}
\newcommand*{\clospn}[1]{\operatorname{Cl}_\mathit{pend}({#1})}
\newcommand*{\closst}[1]{\operatorname{Cl}_\mathit{st}({#1})}
\newcommand*{\atoms}[1]{\operatorname{Atoms}({#1})}

\DeclareMathOperator{\aconstr}{\mathscr{A}\mathscr{C}}
\DeclareMathOperator{\dconstr}{\mathscr{D}\mathscr{R}}
\DeclareMathOperator{\subf}{subf}
\DeclareMathOperator{\ssubf}{ssubf}

\newcommand*{\lcall}{\mathbf{call}}
\newcommand*{\lret}{\mathbf{ret}}

\newcommand*{\lstm}{\mathbf{stm}}
\newcommand*{\lobs}{\mathbf{obs}}
\newcommand*{\lqry}{\mathbf{qry}}

\newcommand{\shortstackrel}[3][.3ex]{%
  \mathrel{\vbox{\offinterlineskip\ialign{%
    \hfil##\hfil\cr
    $\scriptscriptstyle{#2}$\cr
    \noalign{\kern#1}
    $#3$\cr
}}}}
\newcommand{\apush}[1]{\shortstackrel{#1}{\rightarrow}}
\newcommand{\ashift}[1]{\shortstackrel{#1}{\dashrightarrow}}
\newcommand{\apop}[1]{\shortstackrel{#1}{\Rightarrow}}
\newcommand{\asupp}[1]{\stackrel{#1}{\leadsto}}

\newcommand{\suppedge}{\rightarrow}
\newcommand{\spush}{\shortstackrel[-.1ex]{\mathit{push}}{\longrightarrow}}
\newcommand{\sshift}{\shortstackrel[-.1ex]{\mathit{shift}}{\longrightarrow}}
\newcommand{\ssupp}{\shortstackrel[-.1ex]{\mathit{supp}}{\longrightarrow}}
\newcommand{\gedge}{\rightarrow}

\newcommand{\ochain}[3]{{}^{#1}[ #2 ]{}^{#3}}
\newcommand{\opchain}[2]{{}^{#1}[ #2 }

\newcommand{\config}[3]{\langle #1, \allowbreak #2, \allowbreak #3 \rangle}
\newcommand{\symb}[1]{\mathop{smb}(#1)}
\newcommand{\State}[1]{\mathop{st}(#1)}
\newcommand{\tp}{\mathop{top}}

\acrodef{ACP}{Algebra of Communicating Processes}
\acrodef{AP}{Atomic Proposition}
\acrodef{BDD}{Boolean Decision Diagram}
\acrodef{BPA}{Basic Process Algebra}
\acrodef{BSCC}{Bottom Strongly Connected Component}
\acrodef{CEGAR}{Counterexample-Guided Abstraction Refinement}
\acrodef{CFG}{Context-Free Grammar}
\acrodef{CFL}{Context-Free Language}
\acrodef{CPU}{Central Processing Unit}
\acrodef{CTL}{Computation Tree Logic}
\acrodef{DCFL}{De\-ter\-min\-is\-tic Context-Free Language}
\acrodef{DFS}{Depth-First Search}
\acrodef{DHP}{Downward Hierarchical Path}
\acrodef{DSP}{Downward Summary Path}
\acrodef{DS}{Downward Summary}
\acrodef{DRA}[DRA]{Deterministic Rabin Automaton}
\acrodefplural{DRA}[DRA]{Deterministic Rabin Automata}
\acrodef{EF}{Ehrenfeucht-Fra\"iss\'e}
\acrodef{ERSM}{(Extended) Recursive State Machine}
\acrodef{ETR}{Existential first-order Theory of Real numbers}
\acrodef{FOL}{First-Order Logic}
\acrodef{FO}{First-Order}
\acrodef{FSA}{Finite-State Automaton}
\acrodefplural{FSA}[FSA]{Finite-State Automata}
\acrodef{FSM}{Finite-State Machine}
\acrodef{JDK}{Java Development Kit}
\acrodef{lhs}{left-hand side}
\acrodef{LIFO}{Last In First Out}
\acrodef{LR}{Left-Recursive}
\acrodef{LTL}{Linear Temporal Logic}
\acrodef{MC}{Model Checking}
\acrodef{MSOL}{Monadic Second-Order Logic}
\acrodef{MSO}{Monadic Second-Order}
\acrodef{NBA}{Nondeterministic B\"uchi Automaton}
\acrodefplural{NBA}[NBA]{Nondeterministic B\"uchi Automata}
\acrodef{NWA}{Nested Words Automaton}
\acrodefplural{NWA}[NWA]{Nested Words Automata}
\acrodef{NWTL}{Nested Words Temporal Logic}
\acrodef{OPA}{Operator Precedence Automaton}
\acrodefplural{OPA}[OPA]{Operator Precedence Automata}
\acrodef{OPG}{Operator Precedence Grammar}
\acrodef{OPL}{Operator Precedence Language}
\acrodef{OPM}{Operator Precedence Matrix}
\acrodef{OPTL}{Operator Precedence Temporal Logic}
\acrodef{OP}{Operator Precedence}
\acrodef{OVI}{Optimistic Value Iteration}
\acrodef{PDA}{Pushdown Automaton}
\acrodefplural{PDA}[PDA]{Pushdown Automata}
\acrodef{PDS}{Pushdown System}
\acrodefplural{PDS}[PDS's]{Pushdown Systems}
\acrodef{pOPA}{Probabilistic Operator Precedence Automaton}
\acrodefplural{pOPA}[pOPA]{Probabilistic Operator Precedence Automata}
\acrodef{POMC}{Precedence Oriented Model Checker}
\acrodef{POTL}{Precedence Oriented Temporal Logic}
\acrodef{POTLF}[POTLf$\chain$]{Precedence Oriented Temporal Logic}
\acrodef{pPDA}{Probabilistic Pushdown Automaton}
\acrodefplural{pPDA}[pPDA]{Probabilistic Pushdown Automata}
\acrodef{PPL}{Probabilistic Programming Language}
\acrodef{PR}{Precedence Relation}
\acrodef{pVPA}{Probabilistic Visibly Pushdown Automaton}
\acrodefplural{pVPA}[pVPA]{Probabilistic Visibly Pushdown Automata}
\acrodef{RAM}{Random-Access Memory}
\acrodef{rhs}{right-hand side}
\acrodefplural{rhs}[rhs's]{right-hand sides}
\acrodef{RMC}{Recursive Markov Chain}
\acrodef{RR}{Right-Recursive}
\acrodef{RSM}{Recursive State Machine}
\acrodef{SAT}{Satisfiability}
\acrodef{SCC}{Strongly Connected Component}
\acrodef{SMT}{Satisfiability Modulo Theories}
\acrodef{ST}{Syntax Tree}
\acrodef{TS}{Transition System}
\acrodef{UHP}{Upward Hierarchical Path}
\acrodef{UML}{Unified Modeling Language}
\acrodef{UOT}{Unranked Ordered Tree}
\acrodef{USP}{Upward Summary Path}
\acrodef{US}{Upward Summary}
\acrodef{VLTL}{Visibly Linear Temporal Logic}
\acrodef{VPA}{Visibly Pushdown Automaton}
\acrodefplural{VPA}[VPA]{Visibly Pushdown Automata}
\acrodef{VPL}{Visibly Pushdown Language}
\acrodef{OOPL}[$\omega$OPL]{Operator Precedence $\omega$-Language}
\acrodef{OPBA}[$\omega$OPBA]{Operator Precedence B\"uchi Automaton}
\acrodefplural{OPBA}[$\omega$OPBA]{Operator Precedence B\"uchi Automata}
\acrodef{OVPL}[$\omega$VPL]{Visibly Pushdown $\omega$-Language}

\setcopyright{cc}
\copyrightyear{2026}
\acmVolume{27}
\acmNumber{4}
\acmArticle{24}
\acmYear{2026}
\acmDOI{10.1145/3822593}
\acmArticleSeq{0}
\acmJournal{TOCL}

\begin{document}

\title{Model Checking Probabilistic Operator Precedence Automata}

\author{Francesco Pontiggia}
\email{francesco.pontiggia@tuwien.ac.at}
\orcid{0000-0003-2569-6238}

\author{Ezio Bartocci}
\email{ezio.bartocci@tuwien.ac.at}
\orcid{0000-0002-8004-6601}
\affiliation{%
  \institution{TU Wien}
  \city{Vienna}
  \country{Austria}
}

\author{Michele Chiari}
\email{michele.chiari@tuwien.ac.at}
\orcid{0000-0001-7742-9233}
\affiliation{%
  \institution{TU Wien}
  \city{Vienna}
  \country{Austria}
}
\affiliation{%
  \institution{AIT Austrian Institute of Technology}
  \city{Vienna}
  \country{Austria}
}

\begin{abstract}
We address the problem of model checking context-free specifications for probabilistic pushdown automata,
with applications in the verification of recursive probabilistic programs. 
Operator Precedence Languages (OPLs) are an expressive subclass of context-free languages
suitable for model checking recursive programs. The derived
Precedence Oriented Temporal Logic (POTL) can express fundamental OPL specifications
such as pre/post-conditions and exception safety.

We introduce \emph{probabilistic Operator Precedence Automata} (pOPA),
a class of probabilistic pushdown automata whose traces are OPLs,
and study their model checking problem against POTL specifications.
We identify a fragment of POTL, called POTLf$\chain$,
for which we develop an \textsc{exptime} algorithm for qualitative probabilistic model checking, 
and an \textsc{expspace} algorithm for the quantitative variant.
The algorithms rely on the property of \emph{separation}
of automata generated from POTLf$\chain$ formulas.
By the same property, the algorithms allow for model checking pOPA
against Linear Temporal Logic (LTL) specifications.
POTLf$\chain$ is then the first context-free logic for which
an optimal probabilistic model checking algorithm has been developed,
matching its \textsc{exptime} lower bound in complexity.
In comparison, the best known algorithm for probabilistic model checking of CaRet,
a prominent temporal logic based on Visibly Pushdown Languages (VPL),
is doubly exponential.

\end{abstract}

\acused{POTLF}




\maketitle

\section{Introduction}
\label{intro}
\acp{PDA} are a well-established formalism for model checking \emph{recursive} programs~\cite{AlurBE18,AlurBEGRY05}.
Their stack can model a program's stack of activation frames, 
and they represent the infinite state-space of its contents with a finite structure.
\acp{PDA} are usually checked against regular properties, including \ac{LTL}~\cite{Pnueli77} formulas,
but the recursive nature of programs they model often requires more expressive, non-regular specifications.
Examples of such requirements are function-local properties,
Hoare-style pre/post-conditions~\cite{Hoare69},
total and partial correctness~\cite{AlurEM04},
and stack inspection~\cite{JensenLT99}
(i.e., constraints on which functions can be active at a given point of the execution).
The temporal logics CaRet~\cite{AlurEM04}, based on \acp{VPL}~\cite{AluMad04},
 and \ac{POTL}~\cite{ChiariMP21a}, based on \acp{OPL}~\cite{Floyd1963}, express these properties.
Both \acp{OPL} and \acp{VPL} are strict subclasses of deterministic context-free languages
that retain closure w.r.t.\ Boolean operations,
but \acp{OPL} are strictly more expressive than \acp{VPL} \cite{CrespiMandrioli12}.
Hence, \ac{POTL} is also able to express properties about exceptions~\cite{PontiggiaCP21}
such as \emph{exception safety}~\cite{Abrahams00} and the \emph{no-throw} guarantee.

\slimparagraph{Probabilistic Programs.} 
In addition to procedural constructs, \acp{PPL} provide primitives (i) to sample from probability distributions, 
and (ii) to condition values of variables in a program (\emph{conditioning})~\cite{GordonHNR14}.
(Recursive) probabilistic programs implement randomized algorithms \cite{MotwaniR95} 
such as QuickSort with randomized pivot selection~\cite{Hoare62}, 
the Sherwood variant of binary search~\cite{mcconnell2007analysis}, and security and privacy protocols~\cite{BartheKOB13}. 
With growing popularity, they found application, in AI and machine learning~\cite{Ghahramani15}, 
for expressing probabilistic generative models clearly and concisely.
A probabilistic program implements a generative model as a function, and its semantics is the probability distribution over program variables at the return statement, called the \emph{posterior distribution}.
To incorporate evidence of observed events in the model, the posterior distribution is \emph{conditioned} on observed data.
Conditioning is a first-class citizen in most \acp{PPL}:
it is represented \emph{within} programs with \emph{ad hoc} constructs
such as \texttt{observe} statements, which force random variables to take particular values.

The problem of computing an explicit representation of the posterior distribution is called \emph{inference}, in extension to Bayesian inference for graphical models. It is also referred to as \emph{querying}~\cite{GoodmanMRBT08} the program. A number of inference techniques exist~\cite{vandeMeentPY18}, depending on the semantics and the implementation of conditioning operators.
When using a \emph{rejection sampling} semantics, a program contains statements of the form \texttt{observe (e)},
where \texttt{e} is a Boolean condition.
Only probabilistic computations (i.e. trials) in which all \texttt{observe} conditions are true appear in the posterior distribution:
when a condition fails, the current computation is discarded and excluded from the posterior.
In these cases, we assume that \texttt{observe} \emph{triggers a new trial of the queried program}, hoping that the next run will satisfy all \texttt{observe} statements. 

Inference queries model reasoning patterns, in which observations represent beliefs, desires, goals, or choices.
As mentioned, most modern \acp{PPL} implement inference queries as functions and,
just like ordinary programming languages, 
they support invoking them recursively~\cite{vandeMeentPY18,WoodMM14,GoodmanMRBT08,Goodman14,TolpinMYW16}.
Recursive inference queries, or \emph{nested queries}, enable \emph{reasoning about reasoning}:
a probabilistic program representing an intelligent agent samples from the conditional distribution given by another probabilistic program, implemented as a nested inference query.
Nested queries are then a tool to model meta-reasoning
and planning in multi-agent systems~\cite{Evans17,ZhangA22,StuhlmullerG14}.

\begin{wrapfigure}[19]{r}{0.38\textwidth}
\includegraphics[width=0.38\textwidth]{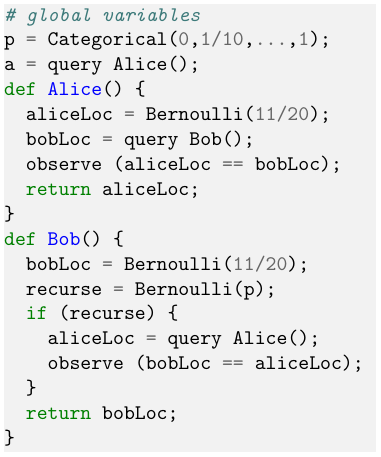}
\captionof{figure}{A coordination game adapted from~\cite{StuhlmullerG14}.}
\label{fig:coordination_game}
\end{wrapfigure}
In this setting, temporal properties concerning agent behaviors are also of interest, besides the posterior distribution.
Verifying formal properties on such programs is, however, a very challenging task,
as demonstrated by significant theoretical efforts~\cite{EtessamiY12,BrazdilEKK13,DubslaffBB12,WinklerGK22}.

\slimparagraph{Motivating Example.}
The program of Fig.~\ref{fig:coordination_game} 
is an instance of a \emph{Schelling coordination game}~\cite{Schelling1980}, and it is a variant of the Church program presented in \mbox{\citet{StuhlmullerG14}} (see also \cite{ZhangA22}). 
Two agents, Alice and Bob, have agreed to meet today, but they have not yet decided to which of two cafés to go. They have no means of direct communication, but they know each other's preferences. The decision-making process of each agent is modelled as an inference query.
Alice reasons as follows.
She samples a location according to her preferences (a binary choice with bias 0.55), and then asks herself: ``What is Bob thinking of at the moment?''.
She proceeds to sample Bob's choice by \emph{querying} his behavior (procedure \texttt{Bob()}).
Finally, she conditions on the outcome of the two choices being the same, which is her goal.
Bob's procedure differs in that he decides probabilistically whether to further query Alice's behavior---i.e., with some probability he forgets about Alice and picks a café only according to his preferences.
The two inference queries are then mutually recursive, modeling nested reasoning about each other's reasoning,
in a potentially infinite manner.
Global variable \texttt{p} parameterizes Bob's ``forgetfulness'', thus controlling recursion depth.
The semantics of the overall program is the posterior distribution of Alice's first inference query, conditioned on her reasoning about Bob's reasoning about Alice's reasoning\dots

We model this program as a \ac{pPDA} \cite{EsparzaKM04,BrazdilEKK13},
the probabilistic counterpart of \acp{PDA}.
Procedure calls and returns are modeled as transitions that respectively push and pop a stack symbol;
queries also push one stack symbol, and false observations pop all stack symbols until the first query symbol,
effectively \emph{unwinding} the stack, and then re-start the queried program in a loop,
thus implementing \emph{rejection sampling}~\cite{OlmedoGJKKM18}.
We encode finite-domain program variables, parameters and return values into states and stack symbols.

We are interested in verification of \emph{temporal properties} of this program.
For example, what is the probability that a sample call of procedure \texttt{Alice()} with \texttt{p} $\geq \lambda$ (\emph{pre-condition}) is rejected because of a disagreement on the caf\'e choice (\emph{post-condition})?
This is equivalent to checking that the outer query calling \texttt{Alice()} encounters a non-satisfied \texttt{observe()} statement in a sample trial.
To answer, we need to reason about only the observations that affect the outer query, and skip events inside inner queries: 
an execution trace of the overall program has multiple nested calls to \texttt{Alice()}, 
each one either returning correctly, or being rejected by a failed observation. 
Only one statement terminates the first call to \texttt{Alice()}.
Thus, we must express our specification in a formalism that can skip parts of execution traces depending on stack behavior.
\ac{LTL} is not expressive enough, because it only captures regular properties~\cite{AlurEM04}.
We need a logic with context-free modalities, such as \ac{POTL}~\cite{ChiariMP21a},
which achieves this with the $\lcunext$ operator.
The requirement above corresponds to the following \ac{POTL} formula:
\begin{equation}
  \llglob \, (\lcall \land \text{Alice} \land p \geq \lambda \implies \neg \lcunext \lobs).
  \label{eq:disagreement}
\end{equation}
Moreover, what is the probability that caf\'e 1 is eventually chosen?
We can answer by checking the satisfaction probability of formula
\begin{equation}
\lcunext (\mathtt{aliceLoc} == 1).
\label{eq:aliceLoc1}
\end{equation}

\slimparagraph{Our approach (Fig.~\ref{fig:qualitative-workflow-diagram}).}
We study the problem of probabilistic model checking of a \ac{pPDA} against specifications in a fragment of \ac{POTL}.
\ac{POTL} captures the first-order definable fragment of \acp{OPL}~\cite{ChiariMP21b}.
For model checking to be decidable, we introduce a new class of \ac{pPDA} whose traces are \acp{OPL}, called \emph{\acp{pOPA}}.
\acp{pOPA} model behaviors of modern \acp{PPL}---for example, a \ac{pOPA} corresponds to the motivating example of Fig.~\ref{fig:coordination_game}.

\begin{figure}[tb]
\centering
\begin{tikzpicture}
  [ node distance=0pt, font=\scriptsize, >=latex,
    action/.style={draw, rectangle},
    cond/.style={draw, chamfered rectangle, chamfered rectangle xsep=5em, chamfered rectangle ysep=0pt},
    round/.style={draw, circle},
    scale=.7
  ]
\node (phi) [action, align=center] {\acs{POTLF} formula\\(Sec.~\ref{sec:potl})};
\node (opba) [action, right=10pt of phi, align=center] {separated \acs{OPBA}\\(Sec.~\ref{sec:separated-opba})};
\node (sg) [action, right=10pt of opba, align=center] {Support Graph\\(Sec.~\ref{sec:support-graph})};
\node (prog) [action, below=7pt of phi.south west, anchor=north west, align=center] {Program};
\node (popa) [action, right=38pt of prog, align=center] {\acs{pOPA}\\(Sec.~\ref{sec:popa})};
\node (sc) [action, right=38pt of popa, align=center] {Support Chain\\(Sec.~\ref{sec:support-chain})};
\node (prod) [round, inner sep=2pt, right=6pt of sc, yshift=12pt] {$\times$};
\node (g) [action, right=10pt of prod, align=center] {Graph $G$\\(Sec.~\ref{sec:qualitative-mc})};
\node (ast) [action, right=30pt of g, align=center] {AS?};
\node (yes) [round, right=130pt of sg, align=center] {\cmark};
\node (no) [round, below=10pt of yes, align=center] {\xmark};
\node (quant) [action, right=10pt of no, align=center] {quant. \\ bounds};
\path (phi) edge[->] (opba)
      (opba) edge[->] (sg)
      (prog) edge[->] (popa)
      (popa) edge[->] (sc); 
\path[draw, ->] (sg.east) -| (prod.north);
\path[draw, ->] (sc.east) -| (prod.south);
\path (prod) edge[->] (g)
      (g) edge[->] node[align=center] {\acs{SCC}\\analysis} (ast);
\path (ast) edge[->] (yes)
      (ast) edge[->] (no)
      (no) edge[->] (quant);
\end{tikzpicture}
\caption{Overview of model checking of probabilistic programs against \acs{POTLF}.}
\label{fig:qualitative-workflow-diagram}
\end{figure}

\ac{POTL} formulas can be translated to \acp{OPBA}, the class of automata that accepts \acp{OPL} \cite{ChiariMPP23}.
To perform non-probabilistic \ac{POTL} model checking, one must build the \emph{support graph}~\cite{ChiariMPP23},
a finite graph with edges representing fragments of runs of an \ac{OPBA}
in between the transitions pushing and popping the same stack symbol.
We introduce the \emph{support chain} as the probabilistic counterpart of the support graph:
it is a finite Markov chain that encodes all runs of a \ac{pOPA} and preserves its probability distribution.

A \ac{pOPA} $\mathcal{A}$ and an \ac{OPBA} $\mathcal{B}$ decide whether to make a push or a pop move depending on \emph{\acp{PR}}
between labels.
If they agree on \acp{PR}, their stacks are synchronized, and we can define
a synchronized product $G$ of the support chain of $\mathcal{A}$ and the support graph of $\mathcal{B}$.
If $\mathcal{B}$ is deterministic, then a simple reachability analysis of the accepting \acp{BSCC} of $G$ suffices for qualitative model checking. 
However, \acp{OPBA} obtained from \ac{POTL} formulas are in general not deterministic.

We thus exploit a property called \emph{separation}~\cite{CartonM03}.
Separation in (flat) B\"uchi automata enables checking Markov Chains~\cite{CouvreurSS03} and \acp{RMC}~\cite{YannakakisE05}
against \acs{LTL} specifications in singly exponential time in formula length,
i.e., without the additional exponential blowup due to automata determinization.
A (flat) B\"uchi automaton is separated iff the languages it accepts starting from different states are disjoint.
A variant \cite{CouvreurSS03} of the Vardi-Wolper construction~\cite{VardiW86} for \acs{LTL} yields exponential-size separated B\"uchi automata~\cite{Wilke17}. 
As we shall see, things get more complicated for pushdown automata.
We study for the first time separation for pushdown automata 
and define \emph{separated} \ac{OPBA}.
We identify a fragment of \acs{POTL} called \acs{POTLF} that can be encoded as separated \ac{OPBA}.
We then identify three conditions that characterize \acp{SCC} of the synchronized product $G$.
These conditions hold for a \acs{SCC} of $G$ if and only if the \acs{SCC} subsumes \ac{pOPA} runs that are accepted by the input \emph{separated} \ac{OPBA}.
Thus, we obtain an algorithm for qualitative model checking of \acs{POTLF} formulas.
We then show how to extend it to solve quantitative problems.

\slimparagraph{Complexity.} 
Building the support chain requires computing the \emph{termination probabilities} of $\mathcal{A}$, 
These can be encoded in the \acf{ETR}, which is decidable in \textsc{pspace}~\cite{Canny88}.
Thus, we obtain a qualitative  model checking algorithm that runs in space polynomial in the size of $\mathcal{A}$,
and time exponential in formula length.
Quantitative model checking involves additionally solving on $G$ a set of equations strictly resembling those for termination probabilities, 
hence the algorithm runs in space exponential in formula length.
We prove that the bound for qualitative model checking is optimal.


\slimparagraph{Related Work.}
Termination and model checking of \emph{regular} specifications for \acp{pPDA} 
and the equivalent formalism of \acp{RMC} are studied in \cite{EtessamiY12,BrazdilEKK13}.
\citet{DubslaffBB12} prove that, upon imposing a visibility condition on \ac{pPDA}, model checking against Visibly Pushdown specifications is decidable.
\citet{WinklerGK22} introduce a model checking algorithm for the CaRet~\cite{AlurEM04} temporal logic 
with the purpose of model checking recursive probabilistic programs without conditioning.
They model programs through \acp{pVPA}, which are \acp{pPDA} whose states
can only have outgoing transitions of one type: either pushing, popping, or leaving the stack unchanged.
\acp{VPL}, on which CaRet and \acp{pVPA} are based, are strictly less expressive than \acp{OPL}~\cite{CrespiMandrioli12}.
Thus, the visibility conditions imposed on \acp{pPDA} in these works make them less expressive than \acp{pOPA}:
in particular, they are less effective in modelling nested queries and conditioning constructs.
Moreover, \citet{WinklerGK22} proceed by determinizing automata encoding CaRet formulas,
thus obtaining a 2-\textsc{exptime} (resp.\ \textsc{2-expspace}) complexity of qualitative (resp.\ quantitative) model checking,
as opposed to our \textsc{exptime} (resp.\ \textsc{expspace}) optimal bound (cf.\ Table~\ref{tab:related-work}).

\ac{POTL}~\cite{ChiariMP21a} is a linear-time temporal logic that captures the first-order fragment~\cite{ChiariMP21b} of \acp{OPL}~\cite{Floyd1963}.
The greater generality of \acp{OPL} results in \ac{POTL} being more expressive than \acs{VPL}-based logics \cite{ChiariMP21b}:
CaRet is defined on a one-to-one nesting relation \cite{jacm/AlurM09},
while the nesting relation of \ac{POTL} can be one-to-many or many-to-one.
This allows for expressing properties on exceptions~\cite{PontiggiaCP21}, as well as, most notably in our probabilistic setting, observations and conditioning.

However, the relation between CaRet and \acs{POTLF} is left as an open problem.
The \acs{POTLF} fragment does not contain past operators, 
while the only existing translation of CaRet formulas to \ac{POTL}
uses past operators~\cite{ChiariMP21b}, and hence does not hold for \acs{POTLF}.
On the other hand, despite the best of our efforts, we could not devise any CaRet formula which provably cannot be expressed in \acs{POTLF}.
Due to this technical nuance, we can only state that \acs{POTLF} is ``at least incomparable'' with CaRet: \acs{POTLF}
is still equipped with a more general one-to-many or many-to-one nesting relation.
The same issue holds regarding the relation between \ac{POTLF} and \ac{LTL}.
All existing CaRet formulas in the literature on probabilistic model checking~\cite{WinklerGK22,WinklerGK23} are expressible in \ac{POTLF}.

Non-probabilistic \acs{POTL} model checking is \textsc{exptime}-complete~\cite{ChiariMP21b},
and has been implemented through an automata construction and graph-theoretic analyses \cite{ChiariMPP23}.
In this paper, we exploit the automata construction from \cite{ChiariMPP23}
after showing it yields separated automata for \acs{POTLF}.

\begin{table}
\caption{Complexity of probabilistic pushdown model checking w.r.t.\ the specification formalism.}
\label{tab:related-work}
\centering
\begin{tabular}{l @{\hspace{1em}} l @{\hspace{1em}} l @{\hspace{1em}} l @{\hspace{1em}} l}
\toprule
Specification formalism & Model & Qualitative & Quantitative & Reference \\
\midrule
\ac{LTL} & \ac{RMC} & \textsc{exptime}-complete & \textsc{expspace} & \cite{YannakakisE05,EtessamiY12} \\
\ac{LTL} & \ac{pOPA} & \textsc{exptime}-complete & \textsc{expspace} & Section \ref{ssec:hardness} \\
CaRet & \ac{pVPA} & 2-\textsc{exptime} (\textsc{exptime}-hard) & 2-\textsc{expspace} & \cite{WinklerGK22} \\
Deterministic \ac{OPBA} & \ac{pOPA} & \textsc{pspace}  & \textsc{pspace} & Remark \ref{rem:deterministic-opba}  \\ 
Separated \ac{OPBA} & \ac{pOPA}    & \textsc{pspace}  & \textsc{pspace} & Section \ref{ssec:hardness} \\ 
\acs{POTLF}  & \ac{pOPA}           & \textsc{exptime}-complete & \textsc{expspace} & Section \ref{ssec:hardness} \\
\bottomrule
\end{tabular}
\end{table}


\slimparagraph{Contribution.}
We introduce
(a) a class of \ac{pPDA} whose traces are \acp{OPL} to model recursive probabilistic programs with \texttt{observe} statements;
(b) the class of \emph{separated} \ac{OPBA};
(c) \acs{POTLF}, a fragment of \acs{POTL} that can be encoded with separated \ac{OPBA};
(d) am algorithm for \ac{LTL} and \acs{POTLF} probabilistic model checking,
which we prove to be \textsc{exptime}-complete with regard to the qualitative problem,
and in \textsc{expspace} with regard to the quantitative problem.

\slimparagraph{Organization.}
We give an overview of \acp{OPL} and \acp{OPBA} in Section~\ref{sec:background}, and
introduce \ac{pOPA} in Section~\ref{sec:popa}.
In Section~\ref{sec:potl}, we present \acs{POTLF}, the temporal logic that we use to express formal requirements as \acp{OPL}.
We formalise separation for \acp{OPBA} and describe our novel model checking algorithm in Section~\ref{sec:mc}.
We conclude the paper in Section~\ref{sec:conclusions}.
Missing proofs and technicalities are reported in the appendix.

\section{Background}
\label{sec:background}
In Section~\ref{sec:opl}, we present \acp{OPL} and the automata that accept them,
which are the basis of our verification framework.
In Section~\ref{sec:chain}, we recall the concepts of chains and \acp{OPBA} supports, and in Section~\ref{sec:support-graph} we define the support graph of a \ac{OPBA}.
\subsection{Operator Precedence Languages}
\label{sec:opl}
Finite-word \acp{OPL} were originally introduced in the context of programming language parsing \cite{Floyd1963,MP18}.
Recently, infinite-word \acp{OPL} (\acsp{OOPL}) have been characterized through automata \cite{LonatiEtAl2015}.
\acused{OOPL}
We recall here this automata-theoretic characterization of \acp{OOPL}, which is more suitable for our model checking context.

Let $S$ be a finite set: $S^*$ (resp. $S^\omega$) denotes the set of finite (infinite) words on $S$.
We denote the empty word by $\varepsilon$.
Given a finite alphabet $\Sigma$, we define three \emph{\acfp{PR}}:
for $a, b \in \Sigma$, if $a \lessdot b$ we say $a$ \emph{yields precedence} to $b$,
if $a \doteq b$ we say $a$ is \emph{equal in precedence} to $b$,
and if $a \gtrdot b$ then $a$ \emph{takes precedence} from $b$.
We also use a delimiter $\#$ at the beginning of each infinite word,
and at both ends of each finite word,
and write $\Sigma_\# = \Sigma \cup \{\#\}$.
\begin{definition}[\cite{MP18}]
An \acf{OPM} $M$ is a total function $M : \Sigma_\#^2 \rightarrow \{\lessdot, \doteq, \gtrdot\}$
such that $M(\#, a) = \lessdot$ and $M(a, \#) = \gtrdot$ for all $a \in \Sigma$.
If $M$ is an \ac{OPM} on a finite alphabet $\Sigma$, then $(\Sigma, M)$ is an \emph{\acs{OP} alphabet}.
\end{definition}

\begin{remark}
To retain consistency with the existing literature on \acp{OPL},
we employ the \ac{PR} symbols $\lessdot, \doteq, \gtrdot$ originally introduced by R.~Floyd~\cite{Floyd1963}.
Despite the graphical resemblance, they have nothing to do with ordering relations;
hence \acp{PR} do not necessarily enjoy the properties of ordering relations,
such as reflexivity, antisymmetry and transitivity.
\end{remark}

\acp{OOPL} are accepted by a class of pushdown automata (\acp{OPBA}) that decide whether to push,
pop, or update stack symbols based on \acp{PR} among input symbols.
Thus, the context-free structure of OP words is completely determined by the \ac{OPM}.
\begin{definition}[\cite{LonatiEtAl2015}]
\label{def:opba}
An \acf{OPBA} is a tuple
$\mathcal{B} = (\Sigma, \allowbreak M, \allowbreak Q, \allowbreak I,
\allowbreak F, \allowbreak \delta)$
where $(\Sigma, M)$ is an \acs{OP} alphabet, $Q$ is a finite set of states,
$I \subseteq Q$ and $F \subseteq Q$ are resp.\ the sets of initial and final states,
and $\delta$ is a triple of transition relations
$\delta_\mathit{push}, \delta_\mathit{shift} \subseteq Q \times \Sigma \times Q$ and
$\delta_\mathit{pop} \subseteq Q \times Q \times Q$.
\end{definition}
\ac{OPBA} have a fixed set of stack symbols $\Gamma_\bot = \Gamma \cup \{\bot\}$:
$\bot$ is the initial stack symbol, and other symbols are in $\Gamma = \Sigma \times Q$.
A \emph{configuration} is a triple $\config{w}{q}{B}$ where $w \in \Sigma^\omega$ is the input word,
$q \in Q$ is the current state, and $B = \beta_1 \beta_2 \dots \bot$, with $\beta_1, \beta_2, \dots \in \Gamma$,
is the stack content.
We use $\symb{}$ to denote the terminal character of a stack symbol, 
i.e. $\symb{\bot} = \#$ and $\symb{[a, r]} = a$ for $[a, r] \in \Gamma$. 
We use $\tp()$ to denote the topmost stack symbol, i.e. $\tp(B) = \beta_1$, 
and, with a slight abuse of notation, also use $\symb{B} = \symb{\tp(B)}$.
\acp{OPBA} perform moves of three kinds:

\smallskip
\noindent \textbf{push ($q \apush{a} p$):} $\config{aw}{q}{B} \vdash \config{w}{p}{[a, q] B}$ if $B = \bot$ or
  $\tp(B) = [b, r]$ and $b \lessdot a$, and $(q, a, p) \in \delta_\mathit{push}$;\\
\noindent \textbf{shift ($q \ashift{a} p$):} $\config{aw}{q}{[b, r] B} \vdash \config{w}{p}{[a, r] B}$ if $b \doteq a$
  and $(q, a, p) \in \delta_\mathit{shift}$;\\
\noindent \textbf{pop ($q \apop{r} p$):} $\config{aw}{q}{[b, r] B} \vdash \config{aw}{p}{B}$ if $b \gtrdot a$
  and $(q, r, p) \in \delta_\mathit{pop}$.
\smallskip

Push moves add a new symbol on top of the stack, while pop moves remove the topmost symbol,
and shift moves only update the terminal character in the topmost symbol.
Only push and shift moves read an input symbol, while pop moves are effectively $\varepsilon$-moves.
In the graphical notation $q \apop{r} p$ above, 
$r$ is the state of the (popped) topmost stack symbol, and not an input symbol.
The \ac{PR} between the (terminal character of the) topmost stack symbol and the next input symbol determines the next move:
the stack behavior thus only depends on the input word rather than on the transition relation.

\begin{wrapfigure}[12]{r}{5.5cm}%
\vspace{-11pt}
\centering
\(
\begin{array}{r | c c c c c}
         & \lcall   & \lret   & \lqry    & \lobs    & \lstm \\
\hline
\lcall   & \lessdot & \doteq  & \lessdot & \gtrdot  & \lessdot \\
\lret    & \gtrdot  & \gtrdot & \gtrdot  & \gtrdot  & \gtrdot \\
\lqry    & \lessdot & \doteq  & \lessdot & \lessdot & \lessdot \\
\lobs    & \gtrdot  & \gtrdot & \gtrdot  & \gtrdot  & \gtrdot \\
\lstm    & \gtrdot  & \gtrdot & \gtrdot  & \gtrdot & \gtrdot \\
\end{array}
\)
\captionof{figure}{\acs{OPM} $M_\lcall$, omitting $\#$.
  Rows are labeled with symbols on top of the stack (left)
  and columns with the next input symbol (top):
  for example, $\lcall \doteq \lret$ and $\lret \gtrdot \lcall$.}
\label{fig:opm}
\end{wrapfigure}
A \emph{run} on an $\omega$-word $w$ is an infinite sequence of configurations
$\config{w}{q_0}{B_0} \vdash \config{w_1}{q_1}{B_1} \dots$.
A run is \emph{final} if for infinitely many indices $i \geq 0$ we have $q_i \in F$,
and \emph{initial} if $q_0 \in I$ and $B_0 = \bot$;
if a run is both initial and final, it is called \emph{accepting}.
An $\omega$-word $w$ is accepted by an \ac{OPBA} $\mathcal{B}$ if $\mathcal{B}$ has an accepting run on $w$.
By $L_\mathcal{B}(q, B)$
we denote the set of words $w \in \Sigma^\omega$ such that
$\mathcal{B}$ has a final run starting from $\config{w}{q}{B}$
in which no symbol in $B$ is ever popped.
The language accepted by $\mathcal{B}$ is $L_\mathcal{B} = \cup_{q \in I} L_\mathcal{B}(q, \bot)$.
A language $L \subseteq \Sigma^\omega$ is an \ac{OOPL} if $L = L_\mathcal{B}$ for some \ac{OPBA} $\mathcal{B}$.
Unlike generic pushdown automata, \acp{OPBA}---and thus \acp{OOPL}---%
are closed by the Boolean operations~\cite{LonatiEtAl2015}.

We devise the \ac{OPM} in Fig.~\ref{fig:opm} to represent traces of probabilistic programs with observations.
$\lcall$ and $\lret$ represent respectively function calls and returns,
$\lstm$ statements that do not affect the stack (e.g., assignments),
and $\lqry$ and $\lobs$ are query statements and triggered (false) observations.
The \acp{PR} are assigned so that an \ac{OPBA} always reads $\lcall$ with a push move
($\lcall \lessdot \lcall$, etc.),
and performs a pop move after reading $\lret$ through a shift move
($\lcall \doteq \lret$ and $\lret$ is in the $\gtrdot$ relation with other symbols):
this way, the \ac{OPBA} stack mimics the program's stack.
Moreover, $\lobs$ triggers pop moves that unwind the stack until a symbol with $\lqry$
is reached, in line with their rejection sampling semantics.

\begin{figure}[tb]
\centering
\begin{tikzpicture}
  [node distance=60pt, every state/.style={minimum size=0pt, inner sep=2pt}, >=latex, font=\scriptsize]
\node[state, initial by arrow, initial text=] (q0) {$q_0$};
\node[state] (q1) [right of=q0, initial right, initial text=, accepting] {$q_1$};
\node[state] (q2) [below of=q1, accepting] {$q_2$};
\node[state] (q3) [left of=q2, accepting] {$q_3$};
\path[->] (q0) edge[out=200, in=225, loop, below] node{$\lcall$} (q0)
          (q0) edge[double, out=250, in=275, loop, below] node{$q_0,q_1$} (q0)
          (q0) edge[above] node{$\lcall$} (q1)
          (q1) edge[double, out=20, in=45, loop, right] node{$q_0,q_1$} (q1)
          (q1) edge[out=250, in=110, left, pos=.2] node{$\lcall$} (q2)
          (q2) edge[dashed, out=70, in=290, right] node{$\lret$} (q1)
          (q2) edge[dashed, left, pos=.3] node{$\lret$} (q0)
          (q2) edge[dashed, out=315, in=340, loop, right] node[yshift=3]{$\lret$} (q2)
          (q2) edge[double, out=20, in=45, loop, right] node{$q_0,q_1$} (q2)
          (q2) edge[dashed, above] node{$\lret$} (q3)
          (q3) edge[out=160, in=135, loop, left] node[yshift=2]{$\lcall$} (q3)
          (q3) edge[out=200, in=225, double, loop, left] node[yshift=3]{$q_0,q_1$} (q3);
\end{tikzpicture}%
\begin{adjustbox}{max width=.65\textwidth}
\begin{tikzpicture}
\matrix (m) [matrix of math nodes, column sep=-4, row sep=-4]
{
\lcall     & \lcall              & \lret               & -                   & \lcall             & \lcall             & \lcall           & \lcall \\
           &                     &                     &                     &                    &                    &                  & {[\lcall, q_3]} \\
           &                     &                     &                     &                    &                    & {[\lcall, q_3]}  & {[\lcall, q_3]} \\
           &                     & {[\lcall, q_1]}     & {[\lret, q_1]}      &                    & {[\lcall, q_3]}    & {[\lcall, q_3]}  & {[\lcall, q_3]} \\
           & {[\lcall, q_0]}     & {[\lcall, q_0]}     & {[\lcall, q_0]}     & {[\lcall, q_0]}    & {[\lcall, q_0]}    & {[\lcall, q_0]}  & {[\lcall, q_0]} \\
\bot       & \bot                & \bot                & \bot                & \bot               & \bot               & \bot             & \bot \\[5]
q_0        & q_1                 & q_2                 & q_3                 & q_3                & q_3                & q_3              & q_3 \\
\lessdot   & \lessdot            & \doteq              & \gtrdot             & \lessdot           & \lessdot           & \lessdot         & \lessdot \\
};
\end{tikzpicture}
\end{adjustbox}
\caption{Left: \ac{OPBA} $\mathcal{B}_\lcall$.
Push, shift and pop transitions are depicted resp.\ as solid, dashed and double arrows;
initial states $q_0$ and $q_1$ are marked with an ingoing arrow.
Right: prefix of a run of $\mathcal{B}_\lcall$ showing, from the top, the input word,
the stack growing upwards, the current state, and the \acs{PR} between the topmost stack symbol and the look-ahead.
}
\label{fig:opba-example}
\end{figure}

\begin{example}[Running example]
\label{running-ex:opba}
Fig.~\ref{fig:opba-example} (left) shows \ac{OPBA} $\mathcal{B}_\lcall$ defined on the \ac{OPM} in Fig.~\ref{fig:opm}.
States $q_0$ and $q_1$ are initial, while $q_1$, $q_2$, and $q_3$ are final.
Let $L_D^{>0}$ be the Dyck language on $\{\lcall, \lret\}$ without the empty string, i.e.,
the language of all non-empty finite words such that all prefixes contain no more $\lret$s than $\lcall$s,
and $\lcall$s occur the same number of times as $\lret$s.
$\mathcal{B}_\lcall$ accepts the language $L_1 = (\lcall^* L_D^{>0})^\omega$
by looping between $q_0$, $q_1$ and $q_2$;
once in $q_2$, it nondeterministically guesses that the rest of the word is $\lcall^\omega$,
thus accepting $L_2 = (\lcall^* L_D^{>0})^+ \lcall^\omega$.
Hence, $\mathcal{B}_\lcall$ accepts the language $L_1 \cup L_2$.   
\end{example}

\subsection{Chains and Supports}
\label{sec:chain}
\emph{Chains} formalize how \acp{OPM} define the context-free structure of words accepted by an \ac{OPBA}.
Intuitively, a chain is the sub-word read by an \ac{OPBA} in between the pushing of a stack symbol and its corresponding pop move.
More formally:
\begin{definition}[\cite{LonatiEtAl2015}]
\label{def:chain}
A \emph{closed simple chain}
$
\ochain {c_0} {c_1 c_2 \dots c_\ell} {c_{\ell+1}}
$
is a finite word $c_0 c_1 c_2 \dots c_\ell c_{\ell+1}$
such that
$c_0 \in \Sigma_\#$,
$c_i \in \Sigma$ for $1 \leq i \leq \ell+1$,
and $c_0 \lessdot c_1 \doteq c_2 \dots c_{\ell-1} \doteq c_\ell \gtrdot c_{\ell+1}$.
Symbols $c_0$ and $c_{\ell+1}$ are called  the \emph{left} and \emph{right context} of the chain,
while the string in between is called the \emph{body}.

A \emph{closed composed chain}
$\ochain {c_0} {s_0 c_1 s_1 c_2 \dots c_\ell s_\ell} {c_{\ell+1}}$
is a finite word
$c_0 s_0 c_1 s_1 c_2  \dots \allowbreak c_\ell \allowbreak s_\ell \allowbreak c_{\ell+1}$,
such that
$\ochain {c_0}{c_1 c_2 \dots c_\ell}{c_{\ell+1}}$ is a simple chain, and
$s_i \in \Sigma^*$ is either the empty string
or is such that $\ochain {c_i} {s_i} {c_{i+1}}$ is a closed chain,
for $0 \leq i \leq \ell$.

An \emph{open simple chain} 
$\opchain {c_0} {c_1 c_2 c_3 \dots}$
is an $\omega$-word $c_0 c_1 c_2 c_3 \dots$ such that
$c_0 \in \Sigma_\#$,
$c_i \in \Sigma$ for $i \geq 1$,
and $c_0 \lessdot c_1 \doteq c_2 \doteq c_3 \dots$.

An \emph{open composed chain} is an $\omega$-word that can be of two forms:
\begin{itemize}
\item $\opchain{c_0}{s_0 c_1 s_1 c_2 \dots}$, where the \acp{PR} are as in closed chains
and $s_i \in \Sigma^*$ is either the empty string
or is such that $\ochain {c_i} {s_i} {c_{i+1}}$ is a closed chain, for $i \geq 0$;
\item
$\opchain{c_0}{s_0 c_1 s_1 c_2 \dots c_\ell s_\ell}$,
where $s_i \in \Sigma^*$ is either the empty string
or is such that $\ochain {c_i} {s_i} {c_{i+1}}$ is a closed chain, for $0 \leq i \leq \ell-1$,
and $s_\ell$ is an open chain.
\end{itemize}
Open chains have no right context.
\end{definition}

The portion of a run that reads a chain body is called the chain's \emph{support}.
It corresponds to the portion of a run from the move pushing a stack symbol to the one popping it (if any).
\begin{definition}[\cite{LonatiEtAl2015}]
\label{def:support}
\sloppy
Given an \ac{OPBA} $\mathcal{B}$,
a \emph{support} for a simple chain
$\ochain {c_0} {c_1 c_2 \dots c_\ell} {c_{\ell+1}}$
is a path in $\mathcal{B}$ of the form
$q_0
\apush{c_1}{q_1}
\ashift{c_2}{}
\dots
\ashift{} q_{\ell-1}
\ashift{c_{\ell}}{q_\ell}
\apop{q_0} {q_{\ell+1}}$.
The pop move is executed because of relation $c_\ell \gtrdot c_{\ell+1}$,
and its label is $q_0$, the state pushed at the beginning.

A \emph{support for a composed chain}
$\ochain {c_0} {s_0 c_1 s_1 c_2 \dots c_\ell s_\ell} {c_{\ell+1}}$
is a path in $\mathcal{B}$ of the form
\(
q_0
\asupp{s_0}{q'_0}
\apush{c_1}{q_1}
\asupp{s_1}{q'_1}
\ashift{c_2}{}
\dots
\ashift{c_\ell} {q_\ell}
\asupp{s_\ell}{q'_\ell}
\apop{q'_0}{q_{\ell+1}}
\)
where, for every $0 \leq i \leq \ell$:
if $s_i = \varepsilon$, then $q'_i = q_i$,
else $q_i \asupp{s_i} q'_i$ is a support for $\ochain {c_i} {s_i} {c_{i+1}}$.
We write $q_0 \asupp{x} q_{\ell+1}$ with $x = s_0 c_1 s_1 c_2 \dots c_\ell s_\ell$
if such a support exists.

Supports for open chains, called \emph{open supports}, are identical, except they do not end with a pop move:
the symbol pushed at the beginning of the support remains in the stack forever.
\end{definition}
\fussy

\begin{example}[Running example, cont. \ref{running-ex:opba}]
\label{running-ex:chains}
$\mathcal{B}_\lcall$ from Fig.~\ref{fig:opba-example} (left),
reads the body of the closed simple chain $\ochain{\lcall}{\lcall \, \lret}{\lcall}$
with the support $q_1 \apush{\lcall} q_2 \ashift{\lret} q_3 \apop{q_1} q_3$,
and that of the closed composed chain $\ochain{\lcall}{\lcall \, \lcall \, \lret \, \lret}{\lcall}$
with $q_0 \apush{\lcall} q_1 \asupp{s} q_2 \ashift{\lret} q_3 \apop{q_0} q_3$,
where $s = \lcall \, \lret$.
All closed supports read chains in $L_D^{>0}$, i.e., words of balanced $\lcall$s and $\lret$s.
In general, each support corresponds to the ``lifetime'' of a stack symbol. 
The run in Fig.~\ref{fig:opba-example} (right) contains the closed support $\ochain{\lcall}{\lcall \, \lret}{\lcall}$,
and several open supports:
e.g., the open chain $\opchain{\lcall}{\lcall \dots}$ starting in the last-but-one
position has support $q_3 \apush{\lcall} q_3 \asupp{\lcall \dots}$.
The whole input word is also an open chain:
$\# [ \lcall [[ \lcall \, \lret ] \lcall [ \lcall [ \lcall \dots$
and is read by an open support $q_0 \apush{\lcall} q_1 \asupp{s}$ where $s$ is the rest of the word. 
\end{example}

\subsection{The support graph}
\label{sec:support-graph}
Reachability problems in an \acp{OPBA} are solved by building its \emph{support graph},
which describes its limit behavior.
Nodes of the support graph are \emph{semi-configurations} in
$\mathcal{C} = Q \times \Sigma_\# \times \Sigma$.
Let $d = (q, b, a) \in \mathcal{C}$.
$d$ represents all configurations in which $q$ is the current state,
$b$ is the input symbol of the topmost stack symbol, and $a$ is a look-ahead.
$d$ is \emph{initial} iff $q \in I$ and $b = \#$.
We define $\State{d} = q$.
Edges of the support graph represent closed supports, or push (resp.\ shift) moves that push (resp.\ update)
stack symbols.
\begin{definition}[\cite{ChiariMPP23}]
\label{def:support-graph-opba}
Given a \ac{OPBA}
$\mathcal{B} = (\Sigma, \allowbreak M, \allowbreak Q, \allowbreak I,
\allowbreak F, \allowbreak \delta)$,
its \emph{support graph} is a pair
$(\mathcal{C}, \mathcal{E})$ where $\mathcal{E} \subseteq \mathcal{C}^2$
is partitioned into the following three sets:
\begin{itemize}
\item
\(
E_\mathit{push} =
\{
  ((q, b, a), (p, a, \ell)) \in \mathcal{C}^2 \mid
  b \lessdot a
  \land (q, a, p) \in \delta_\mathit{push}
\}
\)
\item
\(
E_\mathit{shift} =
\{
  ((q, b, a), (p, a, \ell)) \in \mathcal{C}^2 \mid
  b \doteq a
  \land (q, a, p) \in \delta_\mathit{shift}
\}
\)
\item
\(
E_\mathit{supp} =
\{
  ((q, b, a), (p, b, \ell)) \in \mathcal{C}^2 \mid
  q \asupp{ax} p
  \text{ for some chain } \ochain{b}{ax}{\ell}
\}
\)
\end{itemize}
\end{definition}
We write $c_1 \spush c_2$ (resp.\ $c_1 \sshift c_2$, $c_1 \ssupp c_2$)
meaning that the support graph has an edge
in $E_\mathit{push}$ (resp.\ $E_\mathit{shift}$, $E_\mathit{supp}$)
between semi-configurations $c_1$ and $c_2$.
Note that a support edge may be associated with multiple and even infinite supports in the underlying \ac{OPBA}.

In the following, we assume familiarity with common graph-theoretic concepts such as
\acf{SCC} and \acf{BSCC} \cite{BaierK08}.

\begin{definition}
\label{def:trim-support-graph}
We define a \ac{SCC} of the support graph as \emph{final} if it
contains at least one node whose state is final,
or a support edge that represents a support of a simple or composed chain
that includes a final state.
We call \emph{trim} a support graph deprived of nodes from which
a final \ac{SCC} is not reachable.
\end{definition}
The language accepted by $\mathcal{B}$ is non-empty iff a final \ac{SCC} is reachable from an initial semi-configuration.
It is possible to build the support graph and find its \acp{SCC} in polynomial time
according to \cite[Section 6.2]{ChiariMPP23}.
We say that a semi-configuration $(q, b, a)$ is \emph{pending} in a run $\rho$
if it is part of an open support, i.e., there is a configuration
$\config{ax}{q}{\beta B} \in \rho$, with $\symb{\beta} = b$,
such that $\beta$ is never popped in the rest of the run.
The trim support graph contains all and only semi-configurations that are pending in some final run.

We define a function $\sigma$ that relates runs to paths in the support graph.
Given run $\rho = \rho_0 \rho_1 \dots$ with $\rho_0 = \config{w}{q_0}{\bot}$,
$\sigma(\rho) = \sigma_0 \sigma_1 \dots$
is the path obtained by removing all closed chain supports from $\rho$,
leaving only their starting (before the push) and ending (after the pop) configurations,
and by then converting all configurations to semi-configurations.
Precisely, let $\rho_i = \config{a_i w_i}{q_i}{B_i}$ for all $i \geq 0$.
We have $\sigma_0 = (q_0, \#, a_0)$ and
if $\sigma$ maps $\rho_0 \dots \rho_i$ to $\sigma_0 \dots \sigma_i$,
then $\sigma_{i+1}$ is determined inductively as follows:
if $\symb{B_i} \lessdot a_i$ and the pushed stack symbol $[a_i, q_i]$
is never popped in $\rho$, or if $\symb{B_i} \doteq a_i$,
then $\sigma_{i+1} = (q_{i+1}, \symb{B_{i+1}}, a_{i+1})$;
if $\symb{B_i} \lessdot a_i$ and $\rho$ contains a support $q_i \asupp{a_ix} q_{i+k}$
then $\sigma_{i+1} = (q_{i+k}, \symb{B_{i+k}}, a_{i+k})$.
Observe that $\symb{B_{i+k}} = \symb{B_i}$, and $\symb{B_{i+1}} = a_i$.

\begin{figure}[tb]
\centering
\begin{tikzpicture}
  [baseline=(current bounding box.center),
   state/.style={draw, ellipse, inner sep=1pt}, >=latex, font=\scriptsize,
   support/.style={decoration={snake, amplitude=1pt, segment length=3pt}, decorate}
  ]
\node[state, initial by arrow, initial text=] (n0) {$q_0, \#, \lcall$};
\node[state] (n1) [right=20pt of n0] {$q_0, \lcall, \lcall$};
\node[state] (n3) [above of=n1, accepting] {$q_1, \lcall, \lcall$};
\node[state] (n2) [right=15pt of n3, accepting] {$q_3, \lcall, \lcall$};
\node[state, accepting] (n6) [below of=n0] {$q_2, \#, \lret$};
\node[state, initial by arrow, initial text=, accepting] (n4) [below of=n6] {$q_1, \#, \lcall$};
\node[state] (n5) [right=15pt of n4, accepting] {$q_3, \#, \lcall$};
\node[state] (n7) [right=15pt of n6, accepting] {$q_2, \lcall, \lret$};
\node[state] (n8) [below=15pt of n2] {$q_0, \lret, \lcall$};
\node[state] (n9) [below=5pt of n8] {$q_1, \lret, \lcall$};
\node[state] (n10) [below=5pt of n9] {$q_2, \lret, \lret$};
\node[state] (n11) [below=5pt of n10] {$q_3, \lret, \lcall$};
\path[->] (n0) edge (n1)
          (n0) edge[out=30, in=192] (n3)
          (n0) edge[support, out=120, in=60, loop] (n0)
          (n0) edge[support, out=210, in=150] (n4)
          (n0) edge[support, out=340, in=160] (n5)
          (n0) edge[support] (n6)
          (n1) edge[loop below] (n1)
          (n1) edge[support, loop below] (n1)
          (n1) edge[out=60, in=300] (n3)
          (n1) edge[support, out=60, in=300] (n3)
          (n1) edge[support, out=30, in=190] (n2)
          (n1) edge[support, out=230, in=90] (n7)
          (n2) edge[loop below] (n2)
          (n3) edge[support, out=170, in=185, loop] (n3)
          (n3) edge[support, out=240, in=120] (n1)
          (n3) edge[support] (n2)
          (n3) edge[out=197, in=163] (n7)
          (n3) edge[support, out=197, in=163] (n7)
          (n4) edge[support, out=30, in=330] (n0)
          (n4) edge[support] (n5)
          (n4) edge[support, out=170, in=160, loop] (n4)
          (n4) edge[support] (n6)
          (n5) edge[out=25, in=195] (n2)
          (n7) edge[dashed, out=0, in=180] (n8)
          (n7) edge[dashed, out=0, in=180] (n9)
          (n7) edge[dashed, out=0, in=180] (n10)
          (n7) edge[dashed, out=0, in=180] (n11);
\end{tikzpicture}%
\begin{tikzpicture}
  [baseline=(current bounding box.center),
   state/.style={draw, ellipse, inner sep=1pt}, >=latex, font=\scriptsize,
   support/.style={decoration={snake, amplitude=1pt, segment length=3pt}, decorate}
  ]
\node[state, initial by arrow, initial text=] (n0) {$q_0, \#, \lcall$};
\node[state] (n1) [right=15pt of n0] {$q_0, \lcall, \lcall$};
\node[state] (n2) [right=15pt of n1, accepting] {$q_3, \lcall, \lcall$};
\node[state] (n3) [above of=n1, accepting] {$q_1, \lcall, \lcall$};
\node[state, initial by arrow, initial text=, accepting] (n4) [below of=n0] {$q_1, \#, \lcall$};
\node[state] (n5) [right=15pt of n4, accepting] {$q_3, \#, \lcall$};
\path[->] (n0) edge (n1)
          (n0) edge (n3)
          (n0) edge[support, out=120, in=60, loop] (n0)
          (n0) edge[support, out=240, in=120] (n4)
          (n0) edge[support] (n5)
          (n1) edge[loop below] (n1)
          (n1) edge[support, loop below] (n1)
          (n1) edge[out=60, in=300] (n3)
          (n1) edge[support, out=60, in=300] (n3)
          (n1) edge[support] (n2)
          (n2) edge[loop above] (n2)
          (n3) edge[support, out=170, in=190, loop] (n3)
          (n3) edge[support, out=240, in=120] (n1)
          (n3) edge[support] (n2)
          (n4) edge[support, out=60, in=300] (n0)
          (n4) edge[support] (n5)
          (n4) edge[support, out=165, in=145, loop] (n4)
          (n5) edge (n2);
\end{tikzpicture}
\caption{Left: portion of the support graph of \ac{OPBA} $\mathcal{B}_\lcall$
from Fig.~\ref{fig:opba-example} reachable from initial semi-configurations.
We only include semi-configurations with a look-ahead that matches the label of some outgoing transition
(for instance, $(q_0, \lcall, \lret)$ is reachable from $(q_0, \#, \lcall)$,
but $q_0$ has no outgoing transitions labeled with $\lret$).
Right: the trim support graph.
Push edges are solid lines, shift edges are dashed, and support edges are wavy.
Some edges are both push and support edges.}
\label{fig:opba-support-graph}
\end{figure}

\begin{example}[Running example, cont. \ref{running-ex:chains}]
\label{running-ex:support-graph}
Fig.~\ref{fig:opba-support-graph} shows $\mathcal{B}_\lcall$'s full (left) and trim (right) support graphs.
From semi-configurations $(q, b, a)$ such that $b \gtrdot a$%
---all those with $b = \lret$ in the example---%
only pop moves occur:
since pop moves are not included in the support graph, they do not have any outgoing edges, and they are not part of the trim support graph.
Semi-configurations that are not pending are not included in the trim graph, because only these dead-end nodes are reachable from them.
Recall that support edges represent strings in $L_D^{>0}$
and are all final, because supports must reach $q_2$ to be closed.
Words whose runs end up in the \ac{SCC} made of the sole $(q_3, \lcall, \lcall)$ are in $L_2$,
while those remaining in other \acp{SCC} are in $L_1$.
The run in Fig.~\ref{fig:opba-example} (right) is an example of the former, since its $\sigma$-image is
$(q_0, \#, \lcall) (q_1, \lcall, \lcall) (q_3, \lcall, \lcall)^\omega$.
\end{example}


\section{Probabilistic Operator Precedence Automata}
\label{sec:popa}
We introduce \acfp{pOPA}, a class of \ac{pPDA} whose sets of traces are \acp{OOPL}. In the following, we denote as $\mathfrak{D}(S) = \{f : S \rightarrow [0,1] \mid \sum_{s \in S} f(s) = 1\}$
the set of probability distributions on a finite set $S$.

\begin{definition}
A \ac{pOPA} is a tuple
$\mathcal A = (\Sigma, \allowbreak M, \allowbreak Q, \allowbreak u_0,
\allowbreak \delta, \allowbreak \Lambda)$ where:
$(\Sigma, M)$ is an \acs{OP} alphabet;
$Q$ is a finite set of states (disjoint from $\Sigma$);
$u_0$ is the initial state;
$\Lambda : Q \rightarrow \Sigma$ is a state labelling function; and
$\delta$ is a triple of transition functions
$\delta_{\mathit{push}} : Q \rightarrow \mathfrak{D}(Q)$,
$\delta_{\mathit{shift}} : Q \rightarrow \mathfrak{D}(Q)$, and
$\delta_{\mathit{pop}} : (Q \times Q) \rightarrow \mathfrak{D}(Q)$,
such that pop transitions have the following condition, for all $u, s, v \in Q$:
\begin{equation}
\delta_\mathit{pop}(u, s)(v) > 0 \implies
  \forall a \in \Sigma : a \gtrdot \Lambda(u) \implies a \gtrdot \Lambda(v).
\label{pop-cond}
\end{equation}
\end{definition}

Besides the randomized transition function,
the main variation we introduce on Def.~\ref{def:opba}
is the labeling of states instead of transitions.
This is a standard approach in probabilistic model checking%
---i.e., flat finite-state 
Markov Chains do not read input strings, their traces are sequences of state labels---%
and it is required to retain the model checking problem decidable.
State labels are analogous to the partition in call/return states
in probabilistic \acp{VPA}~\cite{DubslaffBB12,WinklerGK22} which is, however,
more limiting, as it forces each state to either only push or pop stack symbols.
Instead, all \ac{pOPA} states can push, shift, or pop,
depending on the \ac{PR} between the topmost stack symbol and their label.
Condition \eqref{pop-cond} on pop moves is needed for labelings of \ac{pOPA}
runs to be \acp{OOPL} (cf.\ Prop.~\ref{prop:popa-opba}).
\ac{OPBA} use a look-ahead to decide what kind of move to do next,
but \ac{pOPA} use the current state label.
Thanks to condition \eqref{pop-cond}, all sequences of pop moves made by a \ac{pOPA}
are also performed by an \ac{OPBA} using the label of the last \ac{pOPA} state
in the sequence as its look-ahead.

The semantics of a \ac{pOPA} $\mathcal{A}$ is defined through
an infinite Markov chain $\Delta(\mathcal{A})$
whose set of vertices is $Q \times (\Gamma^* \{\bot\})$
where, as in \ac{OPBA}, $\bot$ is the initial stack symbol,
which is never pushed or popped,
and $\Gamma = \Sigma \times Q$ is the set of stack symbols.
We use $\symb{}$ and $\tp{}$ as defined for \ac{OPBA}.
The transition functions reflect the three kinds of moves of \ac{pOPA}.
For any stack contents $A \in \Gamma^* \{\bot\}$:

\smallskip
\noindent \textbf{push:} $(u, A) \apush{x} (v, [\Lambda(u), u] A)$
if $A = \bot$ or
  $\tp(A) = [a, s]$ and $a \lessdot \Lambda(u)$,
  and $\delta_{\mathit{push}}(u)(v) = x$;\\
\textbf{shift:} $(u, [a, s] A) \apush{x} (v, [\Lambda(u), s] A)$
  if $a \doteq \Lambda(u)$
  and $\delta_{\mathit{shift}}(u)(v) = x$;\\
\textbf{pop:} $(u, [a, s] A) \apush{x} (v, A)$
  if $a \gtrdot \Lambda(u)$
  and $\delta_{\mathit{pop}}(u, s)(v) = x$.
\smallskip

A run of $\mathcal{A}$ is an infinite path in $\Delta(\mathcal{A})$ that starts in $(u_0, \bot)$,
where $u_0$ is the initial state.
The first move is always a push.
We call $\mathit{Runs}(\mathcal{A})$ the set of such runs.
The probability space over $\mathit{Runs}(\mathcal{A})$
is obtained by the classic cylinder set construction for Markov chains~\cite{BaierK08}.

\begin{figure}[tb]
\begin{tikzpicture}
  [node distance=40pt, state/.style=state with output, every state/.style={minimum size=0pt, inner sep=2pt}, >=latex, font=\scriptsize]
\node[state, initial by arrow, initial text=] (u0) {$u_0$ \nodepart{lower} $\lcall$};
\node[state] (u1) [right of=u0] {$u_1$ \nodepart{lower} $\lcall$};
\node[state] (u2) [below of=u1] {$u_2$ \nodepart{lower} $\lret$};
\node[state] (u3) [below of=u0] {$u_3$ \nodepart{lower} $\lcall$};
\path[->] (u0) edge[above] node{$1$} (u1)
          (u1) edge[out=250, in=110, left, pos=.8] node[xshift=2]{$\nicefrac{1}{3}$} (u2)
          (u1) edge[above, in=125, out=150, loop, pos=1] node[xshift=.5em]{$\nicefrac{2}{3}$} (u1)
          (u2) edge[dashed, right] node[xshift=-2]{$1$} (u1)
          (u1) edge[double, in=35, out=60, loop, below, pos=.8] node[xshift=4pt]{$u_1\nicefrac{1}{2}$} (u1)
          (u1) edge[double, out=300, in=60, right, pos=.5] node{$u_1\nicefrac{1}{2}$} (u2)
          (u1) edge[double, left, pos=.3] node{$u_01$} (u3)
          (u3) edge[out=190, in=165, loop, above, pos=.8] node{$1$} (u3);
\end{tikzpicture}%
\begin{adjustbox}{max width=.7\textwidth}
\begin{tikzpicture}
\matrix (m) [matrix of math nodes, column sep=-4, row sep=-4]
{
\lcall     & \lcall              & \lret               & -                   & \lcall             & \lcall             & \lcall           & \lcall \\
           &                     &                     &                     &                    &                    &                  & {[\lcall, u_1]} \\
           &                     &                     &                     &                    &                    & {[\lcall, u_1]}  & {[\lcall, u_1]} \\
           &                     & {[\lcall, u_1]}     & {[\lret, u_1]}      &                    & {[\lcall, u_1]}    & {[\lcall, u_1]}  & {[\lcall, u_1]} \\
           & {[\lcall, u_0]}     & {[\lcall, u_0]}     & {[\lcall, u_0]}     & {[\lcall, u_0]}    & {[\lcall, u_0]}    & {[\lcall, u_0]}  & {[\lcall, u_0]} \\
\bot       & \bot                & \bot                & \bot                & \bot               & \bot               & \bot             & \bot \\
u_0        & u_1                 & u_2                 & u_1                 & u_1                & u_1                & u_1              & u_1 \\
\lessdot   & \lessdot            & \doteq              & \gtrdot             & \lessdot           & \lessdot           & \lessdot         & \lessdot \\
};
\end{tikzpicture}
\end{adjustbox}
\caption{Left: \ac{pOPA} $\mathcal{A}_\lcall$.
Push, shift and pop transitions are depicted resp.\ as solid, dashed and double arrows, labeled with their probability and, just for pop moves, with the state they pop.
Unreachable transitions are omitted (e.g., push edges from $u_2$). 
Right: prefix of a run of $\mathcal{A}_\lcall$ showing, from the top, its labeling through $\Lambda_\varepsilon$,
the stack growing upwards, the current state, and the \acs{PR} between the topmost stack symbol and the current state label.}
\label{fig:popa-example}
\end{figure}

Given a run $\rho = (u_0, A_0) (u_1, A_1) \dots$ of $\mathcal{A}$,
we define its labeling
$\Lambda_\varepsilon(\rho) = \Lambda_\varepsilon(u_0, A_0) \Lambda_\varepsilon(u_1, A_1) \dots$,
where $\Lambda_\varepsilon(u, A) = \varepsilon$ if $\symb{A} \gtrdot \Lambda(u)$,
and $\Lambda_\varepsilon(u, A) = \Lambda(u)$ otherwise.
Since a \ac{pOPA} has a unique initial state, all runs have a labeling which starts with $\Lambda(u_0)$.
Labels link \ac{pOPA} to \ac{OOPL}:%
\begin{proposition}
\label{prop:popa-opba}
The set of $\Lambda_\varepsilon$-labelings of all runs of a \ac{pOPA} is an \ac{OOPL}.
\end{proposition}
\begin{proof}
We can build an \ac{OPBA} accepting the language of $\Lambda_\varepsilon$-labelings
of runs of a \ac{pOPA} by taking the same set of states and transitions,
and simply moving state labels to input symbols of outgoing transitions.
The soundness of this construction for push and shift moves is trivial;
for pop moves we need condition~\eqref{pop-cond}.
In fact, pop moves of \ac{OPBA} use the next input symbol as a look-ahead
for checking the current \ac{PR},
so for a sequence of consecutive pop moves of the \ac{pOPA} such as
\(
(u_1, [a_1, s_1] \dots [a_n, s_n]A)
\dots (u_n, [a_n, s_n]A) (u_{n+1}, A)
\)
the \ac{OPBA} has a corresponding sequence
\(
\config{\Lambda(u_{n+1}) \dots}{u_1}{[a_1, s_1] \dots [a_n, s_n]A} \allowbreak
\dots \allowbreak
\config{\Lambda(u_{n+1}) \dots}{u_n}{[a_n, s_n]A} \allowbreak
\config{\Lambda(u_{n+1}) \dots}{u_{n+1}}{A}
\)
because $a_i \gtrdot \Lambda(u_i)$ implies $a_{i} \gtrdot \Lambda(u_{i+1})$
for all $1 \leq i \leq n$,
and therefore $a_i \gtrdot \Lambda(u_{n+1})$ for all $i$.


\end{proof}

\acp{pOPA} model recursive probabilistic programs similarly to how \ac{OPBA} \cite{ChiariMPP23}
and pushdown automata~\cite{AlurBE18} model non-probabilistic recursive programs.
With \ac{OPM} $M_\lcall$, push moves from states labeled with $\lcall$
simulate the allocation of a function frame on the program stack,
and pop moves performed after $\lret$ mimic the removal of the frame after the function returns.
States represent valuations of program variables,
which must have finite domains.
Observations are simulated 
with transitions that restart a queried function when an observation fails, in a loop, in line with their rejection sampling semantics.
In practice, a \ac{pOPA} implements a \texttt{query} statement by pushing a stack symbol
labeled with $\lqry$, and then one labeled with $\lcall$ representing the activation frame of the queried function,
which runs normally.
If the queried function (or one of the functions it calls) encounters an \texttt{observe} statement
with a false condition, the \ac{pOPA} goes in a state labeled with $\lobs$:
since $\lcall \gtrdot \lobs$, all symbols representing function stack frames get popped,
until the topmost $\lqry$ symbol is reached.
However, since $\lqry \lessdot \lobs$, the query stack symbol is not popped;
instead, a dummy $\lobs$-labeled stack symbol is pushed and immediately popped,
and a new stack symbol representing the stack frame of the queried function is pushed, and ``a new trial restarts''.
The \ac{pOPA} pops the $\lqry$ symbol only when the queried function reaches its return statement.
Queries can be nested: in this case, thanks to $\lqry \lessdot \lobs$,  an observation rejects the current trial only up to the innermost $\lqry$ symbol, i.e., it restarts only the function associated with the current inference query.
\begin{wrapfigure}[9]{r}{3.5cm}%
\includegraphics[width=3.5cm]{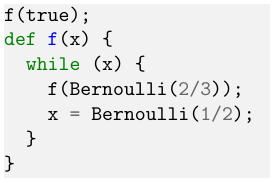}%
\captionof{figure}{Program $P_\mathtt{f}$.}
\label{fig:example-prog}
\end{wrapfigure}

\begin{example}[Running example, cont. \ref{running-ex:support-graph}]
\label{running-ex:popa-example}
Fig.~\ref{fig:popa-example} shows the \ac{pOPA} encoding program $P_\mathtt{f}$ (Fig.~\ref{fig:example-prog})
and one of its runs.
$P_\mathtt{f}$ consists of a function \texttt{f} that calls itself recursively in loop
that breaks with probability $\nicefrac{1}{2}$.
Recursive calls start the loop with probability $\nicefrac{2}{3}$.
States $u_1$ and $u_2$ model \texttt{f}: the two push moves originating in $u_1$
decide whether the next call to \texttt{f} enters the loop or returns, and pop moves whether the loop continues.
State $u_0$ models the initial call \texttt{f(true)}, and $u_3$ is a sink state:
if the initial call terminates, the \ac{pOPA} reaches $u_3$ and loops on it.
\end{example}

\slimparagraph{Termination Probabilities.}
To answer reachability questions for \acp{pOPA}, we must compute
the \emph{termination probabilities} of \emph{semi-con\-fi\-gu\-rations} in $Q \times \Gamma_\bot$.
We define $\pvar{u}{\alpha}{v}$ with $u, v \in Q$, $\alpha \in \Gamma$
as the probability that a \ac{pOPA} in state $u$ with $\alpha$ on top of its stack
eventually pops $\alpha$ and reaches state $v$\footnote{Note that shift transitions might update the input symbol of $\alpha$. In this cases, we consider the probability that the updated stack symbol is eventually popped.}.
More formally, it is the probability that the \ac{pOPA} has a sequence of transitions of the form
\(
(u, \alpha \bot) \dots (v, \bot)
\)
in which no \((r, \bot)\) appears before $(v, \bot)$ for any $r \in Q$.
If $\alpha = \bot$, since $\bot$ is never popped, we define $\pvar{u}{\bot}{v}$ as trivially zero for all $u, v \in Q$.

As \acp{pOPA} are a subclass of \acp{pPDA}, we exploit
previous results on \ac{pPDA} termination probabilities~\cite{BrazdilEKK13}.
\ac{pOPA} termination probabilities are the least non-negative solutions of the equation system
$\mathbf{v} = f(\mathbf{v})$, where $\mathbf{v}$ is the vector of triples
$\pvar{u}{\alpha}{v}$ for all $u, v \in Q$, $\alpha \in \Gamma_\bot$ and $f(\pvar{u}{\alpha}{v})$ is given by
\[
\begin{cases}
0 & \text{if $\alpha = \bot$} \\
\sum_{r, t \in Q} \delta_\mathit{push}(u)(r) \pvar{r}{[\Lambda(u), u]}{t} \pvar{t}{\alpha}{v} & \text{if $\alpha = [a, s]$ and $a \lessdot \Lambda(u)$} \\
\sum_{r \in Q} \delta_\mathit{shift}(u)(r) \pvar{r}{[\Lambda(u), s]}{v} & \text{if $\alpha = [a, s]$ and $a \doteq \Lambda(u)$} \\
\delta_\mathit{pop}(u, s)(v) & \text{if $\alpha = [a, s]$ and $a \gtrdot \Lambda(u)$}
\end{cases}
\]

Since \acp{pOPA} are equipped with state labels, a semi-configuration performs either only push moves, only shift moves, or only pop moves.
Intuitively, the probability of a semi-configuration to terminate in a particular state is, depending on the kind of moves the semi-configuration does, either:
(i) the probability of directly going to $v$ with a pop transition (pop moves);
(ii) the sum of the probabilities of semi-configurations reached with a shift move, weighted by the distribution probabilities of the shift transition (shift moves);
(iii) the sum of the probabilities of semi-configurations reached with chain supports, weighted by the probabilities of the supports (push moves).
In the last case, the concept of \emph{chain support} for a \ac{pOPA} has an analogous definition as for a \ac{OPBA} (Def.~\ref{def:support}), with state labels replacing input symbols. 

\begin{example}[Running example, cont \ref{running-ex:popa-example}]
\label{running-ex:termination-probs}
Some termination probabilities of $\mathcal{A}_\lcall$ (Fig.~\ref{fig:popa-example}) are:
\(
\pvar{u_1}{[\lcall, u_0]}{u_3} = \nicefrac{1}{2},
\)
\(
\pvar{u_1}{[\lcall, u_1]}{u_1} = \pvar{u_1}{[\lcall, u_1]}{u_2} = \nicefrac{1}{4},
\)
and
\(
\pvar{u_2}{[\lcall, u_0]}{u_3} =
\pvar{u_1}{[\lret, u_0]}{u_3} =
1.
\)
\end{example}

\begin{remark}
While these example values are rational, in general
solutions of $\mathbf{v} = f(\mathbf{v})$ are algebraic numbers
that may be irrational, and not even expressible by radicals~\cite{EtessamiY09}.
Questions such as whether a value in $\mathbf{v}$ is $<$, $>$, or $=$
to a rational constant can be answered by encoding the system
in \ac{ETR}~\cite{EtessamiY09},
which is decidable in polynomial space~\cite{Canny88} and exponential time~\cite{Renegar92}.
Various numerical methods can also approximate the solutions~\cite{EtessamiY09,WojtczakE07,WinklerK23a}. 
\end{remark}

\subsection{The Support Chain}
\label{sec:support-chain}

We introduce the \emph{support chain},
a finite Markov chain that describes the behavior of non-terminating \ac{pOPA} runs
while preserving their probability distribution.
This concept is analogous to the \emph{summary Markov chain} of \cite{YannakakisE05},
as well as the \emph{step chain} in \cite{WinklerGK22}.
It first appeared in~\cite{EsparzaKM04}.

The support chain of a \acp{pOPA} is similar to the support graph of a \acp{OPBA}
(Def.~\ref{def:support-graph-opba}), but augmented with probabilities.
First, we reintroduce the concept of support graph, adapting it to \acp{pOPA}. 

\begin{definition}
The \emph{support graph} of a \ac{pOPA}
$\mathcal{A} = (\Sigma, \allowbreak M, \allowbreak Q, \allowbreak u_0,
\allowbreak \delta, \allowbreak \Lambda)$ is a pair
$(\mathcal{C}, \mathcal{E})$ where $\mathcal{C} \subseteq Q \times \Gamma_\bot$
is the set of \emph{semi-configurations} and
$\mathcal{E}$ is the following triple of subsets of $\mathcal{C}^2$:
\begin{itemize}
\item
\(
E_\mathit{push} =
\{
  ((u, \alpha), (v, [\Lambda(u), u])) \in \mathcal{C}^2 \mid
  \symb{\alpha} \lessdot \Lambda(u)
  \land \delta_\mathit{push}(u)(v) > 0
\}
\)
\item
\(
E_\mathit{shift} =
\{
  ((u, [a, s]), (v, [\Lambda(u), s])) \in \mathcal{C}^2 \mid
  a \doteq \Lambda(u)
  \land \delta_\mathit{shift}(u)(v) > 0
\}
\)
\item
\(
E_\mathit{supp} =
\{
  ((u, \alpha), (v, \alpha)) \in \mathcal{C}^2 \mid
  \symb{\alpha} \lessdot \Lambda(u)
  \text{ and $\mathcal{A}$ has a support $u \asupp{}{} v$}
\}
\)
\end{itemize}
\end{definition}
We write $c \shortstackrel[-.1ex]{\mathit{type}}{\longrightarrow} c'$
iff $(c, c') \in E_\mathit{type}$ where $c, c' \in \mathcal{C}$ and
$\mathit{type}$ is one of $\mathit{push}$, $\mathit{shift}$, or $\mathit{supp}$.
Note that a semi-configuration has outgoing edges either only in $E_\mathit{shift}$, or only in $E_\mathit{push}$ and $E_\mathit{supp}$. Also,
$E_\mathit{shift}$ is always disjoint with $E_\mathit{push}$ and $E_\mathit{supp}$,
but $E_\mathit{push} \cap E_\mathit{supp}$ may contain edges of the form
$(u, [\Lambda(u), u]) \suppedge (v, [\Lambda(u), u])$.

A semi-configuration $(u, \alpha)$ is \emph{pending} in a \ac{pOPA} run $\rho$ if $\alpha$ is never popped after $\rho$ visits a configuration $(u, \alpha A)$ for some stack contents $A$.
The probability that a semi-configuration $(s, \alpha)$ is pending is defined as
\begin{equation}
\nex{s, \alpha} = 1 - \sum_{v \in Q} \pvar{s}{\alpha}{v}.
\label{eq:nex}
\end{equation}
We build the \emph{support chain} of $\mathcal{A}$ starting from its support graph,
by removing all vertices (and incident edges) that have zero probability of being pending.
The probabilities of the remaining edges are then conditioned on the event that
a run of the \ac{pOPA} visiting them never pops the stack symbols
in the semi-configurations they link.


\begin{minipage}{\textwidth}
\begin{definition}
\label{def:support-chain}
The \emph{support chain} $M_\mathcal{A}$ is a Markov chain with state set
$Q_{M_\mathcal{A}} = \{c \in \mathcal{C} \mid \nex{c} > 0\}$,
initial state $(u_0, \bot)$, 
and transition function $\delta_{M_\mathcal{A}}$ defined as:
\begin{itemize}
\item if $(u, \alpha) \sshift (v, \alpha')$, then
  $\delta_{M_\mathcal{A}}(u, \alpha)(v, \alpha') = \delta_{\mathit{shift}}(u)(v) \nex{v, \alpha'} / \nex{u, \alpha}$;
\item otherwise,
  \(
  \delta_{M_\mathcal{A}}(u, \alpha)(v, \alpha') =
    (P_\mathit{push} + P_\mathit{supp}) \nex{v, \alpha'} / \nex{u, \alpha}
  \)
  where
  \begin{itemize}
    \item $P_\mathit{push} = \delta_{\mathit{push}}(u)(v)$ if $(u, \alpha) \spush (v, \alpha')$,
          and $P_\mathit{push} = 0$ otherwise;
    \item $P_\mathit{supp} = \sum_{v' \in V} \delta_{\mathit{push}}(u)(v') \pvar{v'}{[\Lambda(u), u]}{v}$
          if $\alpha = \alpha'$ and $(u, \alpha) \ssupp (v, \alpha)$,\\
          and $P_\mathit{supp} = 0$ otherwise;
  \end{itemize}
  where $V = \{v' \in Q \mid (u, \alpha) \spush (v', [\Lambda(u), u])\}$.
\end{itemize}
\end{definition} 
\end{minipage}

Given a run $\rho = \rho_0 \rho_1 \dots$ of $\mathcal{A}$,
we define $\sigma(\rho) = \sigma_0 \sigma_1 \dots$
by replacing configurations with semi-configurations and short-cutting supports,
in the same way as for \ac{OPBA} runs.
Formally, let $\rho_i = (u_i, A_i)$ for all $i \geq 0$,
and let $k$ be the lowest index such that $\tp(A_k)$ is never popped
(i.e., $(u_k, \tp(A_k))$ is pending in $\rho$).
We have $\sigma_0 = (q_k, \tp(A_k))$ and, inductively,
if $\sigma$ maps $\rho_0 \dots \rho_i$ to $\sigma_0 \dots \sigma_i$,
then $\sigma_{i+1}$ can be determined as follows:
if $\symb{A_i} \lessdot \Lambda(u_i)$ and the pushed stack symbol $\tp(A_{i+1}) = [\Lambda(u_i), u_i]$
is never popped in $\rho$, or if $\symb{A_i} \doteq \Lambda(u_i)$,
then $\sigma_{i+1} = (u_{i+1}, \tp(A_{i+1}))$;
if $\symb{A_i} \lessdot \Lambda(u_i)$ and $\rho$ contains a support $u_i \ssupp u_{i+j}$
then $\sigma_{i+1} = (u_{i+j}, \tp(A_{i+j}))$.

While $\rho$ is a path in the infinite Markov chain $\Delta(\mathcal{A})$,
$\sigma(\rho)$ is a path in the \emph{finite} Markov chain $M_\mathcal{A}$.
The transition probabilities of $M_\mathcal{A}$ are defined so that the probability
distribution of runs of $\mathcal{A}$ is preserved under $\sigma$:
in Theorem~\ref{thm:support-chain} we prove that runs of $\mathcal{A}$
whose $\sigma$-image is not a path in $M_\mathcal{A}$ form a set of measure 0,
while all other sets of runs are mapped by $\sigma$ to sets of the same measure in $M_\mathcal{A}$.
This allows us to analyze the limit behavior of \acp{pOPA}
by using techniques for Markov chains on $M_\mathcal{A}$.
In particular, \acp{BSCC} of $M_\mathcal{A}$ determine the set of states visited infinitely
often with positive probability by runs of $\mathcal{A}$ (and their labelings).

\begin{theorem}
\label{thm:support-chain}
Let $(\Omega, \mathcal{F}, P)$ be the probability space of $\Delta(\mathcal{A})$.
$M_\mathcal{A}$ is a Markov chain with probability space
$(\Omega', \mathcal{F}', P')$ such that
$P(\Omega \setminus \sigma^{-1}(\Omega')) = 0$, and
given a set $F' \in \mathcal{F}'$, we have $F = \sigma^{-1}(F') \in \mathcal{F}$ and $P'(F') = P(F)$.
\end{theorem}

We leave the full proof to Appendix~\ref{sec:support-chain-proof}.

\begin{example}[Running example, cont. \ref{running-ex:termination-probs}]
\label{running-ex:supp-chain}
Fig.~\ref{fig:supp-chain} shows the support chain of $\mathcal{A}_\lcall$ (Fig.~\ref{fig:popa-example}).
It has two \acp{BSCC}:
$\{c_2\}$ represents runs in which $\mathcal{A}_\lcall$ cycles between $u_1$ and $u_2$ forever
(such as the one in Fig~\ref{fig:popa-example}, right),
and $\{c_4\}$ runs in which stack symbol $[\lcall, q_0]$ is popped,
and the \ac{pOPA} ends up in $u_3$.
We have $\nex{c_0} = \nex{c_3} = \nex{c_4} = 1$, and
$\nex{c_1} = \nex{c_2} = \nicefrac{1}{2}$.
\end{example}

\begin{figure}[bt]
\centering
\begin{tikzpicture}
  [state/.style={draw, ellipse, inner sep=1pt}, >=latex, font=\scriptsize,
   support/.style={decoration={snake, amplitude=1pt, segment length=3pt}, decorate}
  ]
\node[state, initial by arrow, initial text=, label={above:$c_0$}] (n0) {$u_0, \bot$};
\node[state] (n1) [right=15pt of n0, label={above:$c_1$}] {$u_1, [\lcall, u_0]$};
\node[state] (n2) [below of=n1, label={below:$c_2$}] {$u_1, [\lcall, u_1]$};
\node[state] (n7) [below of=n0, label={below:$c_3$}] {$u_3, \bot$};
\node[state] (n8) [left=15pt of n7, label={below:$c_4$}] {$u_3, [\lcall, u_3]$};
\path[->] (n0) edge[above] node{$\nicefrac{1}{2}$} (n1)
          (n1) edge[support, out=10, in=350, loop, left] node[xshift=1pt]{$\nicefrac{1}{3}$} (n1)
          (n1) edge[left] node{$\nicefrac{2}{3}$} (n2)
          (n2) edge[support, out=10, in=350, loop, left] node{$1$} (n2)
          (n2) edge[out=10, in=350, loop] (n2)
          (n0) edge[support, right] node{$\nicefrac{1}{2}$} (n7)
          (n7) edge[above] node{$1$} (n8)
          (n8) edge[loop above, above] node{$1$} (n8);
\end{tikzpicture}%
\caption{Support chain of \ac{pOPA} $\mathcal{A}_\lcall$ (Fig.~\ref{fig:popa-example}).
  Push and support edges are depicted resp.\ with solid and wavy lines
  (the self-loop on $c_2$ is both a push and a support edge).
}
\label{fig:supp-chain}
\end{figure}

\begin{remark}
\label{rem:deterministic-opba}
Theorem~\ref{thm:support-chain} allows us to apply some standard techniques
for the analysis of flat Markov chains to \acp{pOPA}.
For instance, let $\mathcal{D}$ be a \emph{deterministic} \ac{OPBA},
i.e., such that its transition relations are functions
$\delta_\mathit{push}, \delta_\mathit{shift} : Q \times \Sigma \rightarrow Q$ and
$\delta_\mathit{pop} : Q \times Q \rightarrow Q$.
The word $\Lambda_\varepsilon(\rho)$ labeling each run $\rho$ of $\mathcal{A}$
induces exactly one run $\tau$ in $\mathcal{D}$.
Moreover, we can define a bijection between $\sigma$-images of runs of $\mathcal{A}$ and $\mathcal{D}$.
This suggests a straightforward algorithm for checking \acp{pOPA} against
specifications given as deterministic \acp{OPBA},
analogous to the one for checking Markov chains against deterministic B\"uchi automata
(see e.g.\ \cite{BaierK08}).
It consists of building a synchronized product between $\mathcal{A}$'s support chain
and $\mathcal{D}$'s support graph.
Since complexity is dominated by building the support chain, we can state
\begin{theorem}
\label{thm:deterministic-opba}
Given a deterministic \ac{OPBA} $\mathcal{D}$ and a \ac{pOPA} $\mathcal{A}$,
the problem of deciding whether
$P(\Lambda_\varepsilon(\mathit{Runs}(\mathcal{A})) \cap L_\mathcal{D}) \geq \varrho$
lies in \textsc{pspace}.
\end{theorem}
\end{remark}


\section{Precedence Oriented Temporal Logic}
\label{sec:potl}
\acused{POTLF}
We report the fragment \acs{POTLF} of \acs{POTL}~\cite{ChiariMP21a} for which we study model checking.
Let $AP$ be a finite set of atomic propositions,
and $AP_s \subseteq AP$ such that $\# \in AP_s$ a set of \emph{structural labels},
which represent the ``role'' of each word position (such as $\lcall, \lret$, etc.) and we write in \textbf{bold} face.
For $t \in \{d, u\}$ and $\mathrm{a} \in AP$, the syntax of \acs{POTLF} the following:
\[
    \varphi \coloneqq \mathrm{a}
    \mid \neg \varphi
    \mid \varphi \lor \varphi
    \mid \lnextsup{t} \varphi
    \mid \lcnext{t} \varphi
    \mid \lguntil{t}{\chi}{\varphi}{\varphi}.
\]

The semantics is defined on \emph{\acs{OP} words}.
An \acs{OP} word is a tuple $w = (\mathbb{N}, \allowbreak <, \allowbreak M_{AP_s}, \allowbreak P)$,
where $\mathbb{N}$ is the set of natural numbers,
$<$ is the usual linear order on them,
$M_{AP_s}$ is an \ac{OPM} defined on $AP_s$,
and $P : AP \rightarrow \powset{\mathbb{N}}$ is a function associating each atomic
proposition to the set of word positions where it holds,
such that exactly one structural label holds in each position%
---formally, for all $i \in \mathbb{N}$, we have $|\{a \in AP_s \mid i \in P(a)\}| = 1$.
By convenience, we extend $M_{AP_s}$ directly to word positions:
for any $i, j \in \mathbb{N}$ and $\mathbf{l}_1, \mathbf{l}_2 \in AP_s$ we write $i \sim j$ iff
$i \in P(\mathbf{l}_1)$ and $j \in P(\mathbf{l}_2)$ and $\mathbf{l}_1 \sim \mathbf{l}_2$,
for some $\sim \in \{\lessdot, \doteq, \gtrdot\}$.
The semantics is defined on the \emph{chain} relation $\chain$,
which is induced by \ac{OPM} $M_{AP_s}$ and the labeling defined by $P$,
such that $\chain(i, j)$ for $i, j \in \mathbb{N}$ iff $i$ and $j$ are the left and right
contexts of the same chain (cf.\ Def.~\ref{def:chain}).

While the \acs{LTL} until is defined on paths of consecutive positions,
\acs{POTLF} \emph{summary} until operators follow both the linear ordering relation
and the $\chain$ relation, skipping parts of a word in the latter case.
The downward version of these paths navigates down the nesting structure of the $\chain$ relation,
descending towards inner functions,
while the upward version only goes up, towards containing functions frames.
We define the resulting type of paths as follows:
\begin{definition}[\cite{ChiariMP21b}]
\label{def:summary}
The \emph{downward summary path} between positions $i$ and $j$, denoted $\pi_\chain^d(w, i, j)$,
is a set of positions $i = i_1 < i_2 < \dots < i_n = j$ such that, for each $1 \leq p < n$,
\[
i_{p+1} =
\begin{cases}
  k & \text{if $k = \max\{ h \mid h \leq j \land \chain(i_p,h) \land (i_p \lessdot h \lor i_p \doteq h)\}$ exists;} \\
  i_p + 1 & \text{otherwise, if $i_p \lessdot (i_p + 1)$ or $i_p \doteq (i_p + 1)$.}
\end{cases}
\]
We write $\pi_\chain^d(w, i, j) = \emptyset$ if no such path exists.
The upward counterpart $\pi_\chain^u(w, i, j)$ is defined by substituting $\gtrdot$ for $\lessdot$.
\end{definition}

The semantics of \acs{POTLF} formulas is defined on single word positions.
Let $w$ be an \acs{OP} word, and $\mathrm{a} \in AP$;
we set $\sim^d = \mathord{\lessdot}$ and $\sim^u = \mathord{\gtrdot}$.
Then, for any position $i \in \mathbb{N}$ of $w$, $t \in \{d, u\}$:
\begin{itemize}
  \item $(w, i) \models \mathrm{a}$ iff $i \in P(\mathrm{a})$;
  \item $(w, i) \models \neg\varphi$ iff $(w,i)\not\models\varphi$;
  \item $(w, i) \models \varphi_1\lor\varphi_2$ iff $(w,i)\models\varphi_1$ or
    $(w,i)\models\varphi_2$;
  \item $(w,i) \models \lnextsup{t} \varphi$ iff
    $(w,i+1) \models \varphi$ and $i \sim^t (i+1)$ or $i \doteq (i+1)$;
  \item $(w,i) \models \lcnext{t} \varphi$
    iff $\exists j > i$ such that $\chain(i,j)$,
    ($i \sim^t j$ or $i \doteq j$), and $(w,j) \models \varphi$;
  \item $(w,i) \models \lguntil{t}{\chi}{\varphi_1}{\varphi_2}$ iff
    $\exists j \geq i$ such that $\pi_\chain^t(w, i, j) \neq \emptyset$,
    $(w, j) \models \varphi_2$ and $\forall j' \in \pi_\chain^t(w, i, j)$ such that $j' < j$
    we have $(w, j') \models \varphi_1$;
\end{itemize}
We additionally employ $\land$ and $\implies$ with the usual semantics.
We define the language denoted by a formula $\varphi$ as
$L_\varphi = \{w \in \powset{AP}^\omega \mid (w, 1) \models \varphi\}$.

\begin{minipage}{\textwidth}
The following \acs{OP} word, whose $\chain$ relation is shown by edges,
represents an execution trace:
\centering
\begin{tikzpicture}
\matrix (m) [matrix of math nodes, column sep=-4, row sep=-4]
{
  \#
  & \color{blue} \lessdot & \lcall
  & \color{blue} \lessdot & \lqry
  & \color{blue} \lessdot & \lcall
  & \color{blue} \lessdot & \lcall, \mathrm{B}
  & \color{purple} \gtrdot & \lobs
  & \color{purple} \gtrdot & \lcall
  & \color{blue} \lessdot & \lcall
  & \color{orange} \doteq & \lret
  & \color{purple} \gtrdot & \lret \dots \\
  0
  & & 1
  & & 2
  & & 3
  & & 4
  & & 5
  & & 6
  & & 7
  & & 8
  & & 9 \\
};
\draw[blue] (m-1-5) to [out=25, in=155, pos=.8] node[yshift=-3pt]{$\lessdot$} (m-1-13);
\draw[blue] (m-1-5) to [out=25, in=155, pos=.2] node[yshift=-2pt]{$\lessdot$} (m-1-11);
\draw[purple] (m-1-7) to [out=25, in=155] node[yshift=-2pt]{$\gtrdot$} (m-1-11);
\draw[orange] (m-1-13) to [out=25, in=155] node[yshift=-4pt]{$\doteq$} (m-1-19);
\end{tikzpicture}
\end{minipage}

With \ac{OPM} $M_\lcall$, each $\lcall$ in a trace is in the $\chain$ relation with the $\lret$
of the same function invocation (e.g., pos.\ 6 and 9) and with the $\lcall$s to functions it invokes,
which are nested into the invoking $\lcall$-$\lret$ pair.
If a function is interrupted by a failed observe statement,
its call is in the $\chain$ relation with the $\lobs$ event
(e.g., pos.\ 3 and 5).
An $\lobs$ may terminate multiple functions,
so multiple $\lcall$s may be in the $\chain$ relation with it.

We briefly give an intuition of \ac{POTLF} semantics;
we refer the reader to \cite{ChiariMP21b} for a more detailed presentation.
The $d$ operators can follow $\lessdot$ and $\doteq$ edges,
going down toward nested function calls.
Operator $\ldnext$ follows successor edges, so e.g., $\ldnext \lcall$ holds in position 2.
Operator $\lcdnext$ follows the $\chain$ relation,
and can therefore skip whole function bodies.
When evaluated in a $\lcall$, $\lcdnext (\lret \land \psi)$ states that $\psi$
holds when the function returns normally, expressing a post-condition.
The downward until operator iterates $\ldnext$ and $\lcdnext$,
and $\pi_\chain^d$ paths can only enter function frames, and not exit.
They can be employed to express specifications local to a function:
if $\lcduntil{\top}{\psi}$ is evaluated on a $\lcall$ position,
it means that $\psi$ holds somewhere in that $\lcall$'s function frame%
---it is an \ac{LTL} \emph{eventually} limited to a function invocation.
For instance, $\lcduntil{\top}{\mathrm{B}}$
holds in positions 1 to 4.
If $i$ and $j$ are $\lcall$s, $\pi_\chain^d(w, i, j)$ contains all $\lcall$s
to functions that are in the stack when the one in $j$ is invoked.
Thus, a formula like $\lcduntil{\neg \mathrm{A}}{\mathrm{B}}$,
where $\mathrm{A}$ and $\mathrm{B}$ are two functions,
means that $\mathrm{B}$ is invoked sometime while the current function is active,
and $\mathrm{A}$ is not on the stack.
E.g., it holds in position 1 due to path 1-2-3-4.

The $u$ operators are symmetric, but they navigate traces up towards outer functions.
For instance, $\lcunext \lobs$ holds in pos.\ 3 because $\chain(3,5)$ and $\lobs$ holds in 5,
and $\lcall \gtrdot \lobs$.
Since $\lcall \doteq \lret$, the $\lcunext$ operator can be used to express pre/post-conditions,
both when a function terminates normally, and or due to a false \texttt{observe} statement.
Thus, when evaluated on a $\lcall$ to a function $\mathtt{f}$,
formula \eqref{eq:aliceLoc1} $\lcunext (\mathtt{aliceLoc} == 1)$
is true iff variable $\mathtt{aliceLoc}$ equals 1 when $\mathtt{f}$ terminates.
In the motivating example (Fig.~\ref{fig:coordination_game}),
we evaluate it on the first function call,
thus checking whether $\mathtt{aliceLoc} == 1$ at the end of the program.

We mix \acs{POTLF} operators with the more familiar \ac{LTL} operators,
which can be checked with our approach too.
For instance, in formula \eqref{eq:disagreement} from Section~\ref{intro},
we use the \ac{LTL} \emph{globally} (or \emph{always}) operator $\llglob \varphi$,
meaning that $\varphi$ holds forever after the position in which it is evaluated:
$(w,i) \models \llglob \varphi$ iff $\forall j \geq i$ we have $(w,j) \models \varphi$.
Thus the formula means that in all time instants (globally) that are $\lcall$s to function $\text{Alice}$,
if $p \geq \lambda$ then that function instance is not terminated by a triggered observation,
because no \texttt{observe} statement within that function call is false.

The construction by~\citet{ChiariMPP23} yields \acp{OPBA} that are not deterministic,
hence we cannot apply Theorem~\ref{thm:deterministic-opba} (Remark~\ref{rem:deterministic-opba}).
As we show next, to perform probabilistic model checking of \ac{POTLF} specifications 
we rely on a weaker form of determinism called separation.


\section{Model Checking \texorpdfstring{\acs{pOPA}}{pOPA} against \texorpdfstring{\acs{POTLF}}{POTLfX}}
\label{sec:mc}
The main contribution of this paper is solving the probabilistic \acs{POTLF} model checking problem:
\begin{definition}[Problem statement]
\label{def:potl-mc-problem}
Let $AP$ be a finite set of atomic propositions,
$(\powset{AP}, M_{AP})$ an \acs{OP} alphabet,
$\mathcal{A}$ a \ac{pOPA} and $\varphi$ a \acs{POTLF} formula
on the same alphabet, and $\varrho \in [0, 1]$ a rational number.
The \emph{qualitative (resp.\ quantitative) probabilistic \acs{POTLF} model checking}
problem amounts to deciding if
\[
P(\Lambda_\varepsilon(\mathit{Runs}(\mathcal{A})) \cap L_\varphi) \geq \varrho
\]
where $\varrho = 1$ (resp.\ $\varrho$ is a rational number in $[0, 1)$).
\end{definition}

For this purpose, we introduce an algorithm that exploits \emph{separation} to model check \acp{pOPA} against \ac{POTLF} and \acs{LTL} specifications.
A (flat) B\"uchi automaton is separated iff the languages it accepts starting
from different states are disjoint~\cite{CartonM03,Wilke17}.
Separated automata have a \emph{backward-deterministic} (or \emph{co-deterministic}) transition relation,
a property that has been exploited to obtain optimal---i.e., singly exponential---%
algorithms for model checking Markov chains against \acs{LTL} properties~\cite{CouvreurSS03}.
There exists a translation of \acs{LTL} formulae into separated B\"uchi automata of exponential size, which can be directly employed for probabilistic model checking.
Alternative techniques incur in one additional exponential blowup
because they proceed with determinizing the B\"uchi automaton~\cite{Safra88,BaierK08} into a \ac{DRA},
or translate the formula directly into a \ac{DRA}
of size doubly exponential in formula length~\cite{EsparzaKS20}.

The rest of the section is organized as follows:
in Section~\ref{sec:separated-opba}, we define a separation property for \ac{OPBA} and prove that \ac{POTLF} formulae can be translated into separated \acp{OPBA};
in Section~\ref{sec:qualitative-mc}, we describe our algorithm for qualitative model checking of \acp{pOPA} against separated \acp{OPBA} and prove its correctness;
in Section~\ref{sec:quantitative-mc}, we give the quantitative model checking algorithm;
in Section~\ref{ssec:hardness}, we summarize our results and their complexity;
finally, in Section~\ref{sec:analysis-motivating-example},
we employ our algorithm to analyze the motivating example from Section~\ref{intro}.



\subsection{Separated \texorpdfstring{\acp{OPBA}}{omegaOPBA}}
\label{sec:separated-opba}

In the following definitions, we fix an \acs{OP} alphabet $(\Sigma, M)$ and identify $\bot\bot$ with just $\bot$.
We denote the language of \emph{finite} words read by an \ac{OPBA} $\mathcal{B}$ that starts from state $q$ and stack contents $B$,
and reaches state $q'$ with an empty stack as
\[L^f_\mathcal{B}(q, B, q') = \{x \in \Sigma^* \mid \config{x \#}{q}{B} \vdash^* \config{\#}{q'}{\bot}\}\]

We call an \ac{OPBA} $\mathcal{B}$ \emph{trim} if for every transition
there exists a word $w \in \Sigma^\omega$ such that $\mathcal{B}$
visits such transition during a final run on $w$.
We call \emph{complete} an \ac{OPBA} that has a final run on every $\omega$-word in $\Sigma^\omega$.

\begin{definition}
\label{def:separated}
An \ac{OPBA} $\mathcal{B}$ is \emph{separated} iff
for every word $w \in \Sigma^\omega$,
stack symbols $\beta_1, \beta_2 \in \Gamma_\bot$
such that $\symb{\beta_1} = \symb{\beta_2}$,
and states $q_1, q_2 \in Q$,
$w \in L_\mathcal{B}(q_1, \beta_1 \bot)$ and $w \in L_\mathcal{B}(q_2, \beta_2 \bot)$
implies $q_1 = q_2$.
\end{definition}

While separated B\"uchi automata have a backward-deterministic transition relation,
separated \ac{OPBA} have a backward-deterministic support graph.
The definition of backward determinism for support graphs is, however, sightly more complex,
because support edges are labeled with words rather than characters.
\begin{definition}
\label{def:backward-det}
The support graph of an \ac{OPBA} $\mathcal{B}$ is backward deterministic if, given a node $d = (q, b, \ell)$:
\begin{itemize}
\item for each $b' \in \Sigma$, there is at most one state $q' \in Q$ such that
  the support graph contains an edge $(q', b', b) \spush d$ if $b' \lessdot b$,
  and $(q', b', b) \sshift d$ if $b' \doteq b$;
\item for each chain $\ochain{b}{ax}{\ell}$ with $ax \in \Sigma^+$,
  there is at most one state $q' \in Q$ such that the support graph contains an edge $(q', b, a) \ssupp d$,
  and $\mathcal{B}$ has a run going from configuration
  $\config{ax\ell}{q'}{\beta \bot}$ to $\config{\ell}{q}{\beta \bot}$
  for some stack symbol $\beta$ such that $\symb{\beta} = b$
  and $\beta$ is never popped during the run.
\end{itemize}
\end{definition}

The second point in Def.~\ref{def:backward-det} is equivalent to saying
that there exists at most one state $q'$ such that $\mathcal{B}$ has a support of the form
$q' \apush{a} q'' \dots q''' \apop{q'} q$, which reads $x$ in the portion between $q''$ and $q'''$.

\begin{proposition}
\label{prop:sep-backward}
If a trim \ac{OPBA} is separated, then its trim support graph is backward deterministic.
\end{proposition}

Before proving Proposition~\ref{prop:sep-backward}, we illustrate Def.~\ref{def:backward-det} with an example.

\begin{example}[Running example, cont. \ref{running-ex:supp-chain}]
\label{running-ex:separated}
The \ac{OPBA} in Fig.~\ref{fig:opba-example} is separated and, in fact,
its trim support graph in Fig.~\ref{fig:opba-support-graph} (right) is backward deterministic.
For instance, consider node $(q_3, \lcall, \lcall)$, which has two incoming push edges,
both corresponding to push moves that read $\lcall$.
If we consider the notation from Def.~\ref{def:backward-det},
with both $b' = \#$ or $b' = \lcall$ we have only $q' = q_3$.
Consider node $(q_0, \lcall, \lcall)$, which has two incoming support edges:
the languages of finite strings that the \ac{OPBA} reads with these supports are disjoint.
The support edge coming from $(q_1, \lcall, \lcall)$ represents support
$q_1 \apush{\lcall} q_2 \ashift{\lret} q_0 \apop{q_1} q_0$,
hence with the chain $\ochain{\lcall}{\lcall \, \lret}{\lcall}$
we have $b = \lcall$, $ax = \lcall \, \lret$, $\ell = \lcall$, and $q' = q_1$.
Thus, the \ac{OPBA} has a run from $\config{\lcall \, \lret \, \lcall}{q_1}{[\lcall, q_0]\bot}$
to $\config{\lcall}{q_0}{[\lcall, q_0]\bot}$ in which $[\lcall, q_0]$ is never popped.
For the sake of readability, we do not report all supports represented by the support edge looping on $(q_0, \lcall, \lcall)$.
However, we note that they read the language $L_{supp} = \lcall \, L_D^{>0} \, \lret$.
Any string in the language begins with at least two $\lcall$ symbols, 
and it has the same number of $\lcall$ and $\lret$, because all chains in a closed support must be closed.
Moreover, every prefix of a string contains strictly more $\lcall$ than $\lret$,
otherwise the support would be closed by the prefix already, and the remaining part would not be part of the support.
 For any string $\lcall \, \lcall \, x$ in this language, the \ac{OPBA} has a run from
 $\config{\lcall \, \lcall \, x \, \lcall}{q_0}{[\lcall, q_0]\bot}$
 to $\config{\lcall}{q_0}{[\lcall, q_0]\bot}$ in which $[\lcall, q_0]$ is never popped.

\end{example}

\begin{proof}[Proof of Proposition~\ref{prop:sep-backward}]
We fix an \ac{OPBA} $\mathcal{B}$.
Let $d = (q, b, \ell)$ be a node in the trim support graph of $\mathcal{B}$,
and let $b' \in \Sigma$.
Suppose by contradiction that there are two states $p_1$ and $p_2$
such that there are edges $(p_1, b', b) \spush d$ and $(p_2, b', b) \spush d$
in the trim support graph.
Since the support graph is trim, $d$ reaches a final \ac{SCC},
and so do $(p_1, b', b)$ and $(p_2, b', b)$.
Thus, there is a word $b w \in \Sigma^\omega$
such that $bw \in L_\mathcal{B}(p_1, \beta_1 \bot)$ and
$bw \in L_\mathcal{B}(p_2, \beta_2 \bot)$ for some stack symbols $\beta_1, \beta_2$
such that $\symb{\beta_1} = \symb{\beta_2} = b'$,
and the support starting with $b$ is open in the final run on $bw$.
By Definition~\ref{def:separated} we have $p_1 = p_2$.
The case for shift edges is analogous.

Suppose, again by contradiction, that there are two states $p_1$ and $p_2$
such that there are edges $(p_1, b, a) \ssupp d$ and $(p_2, b, a) \ssupp d$
in the trim support graph and $\mathcal{B}$ has two runs going resp.\ from configuration
$\config{ax\ell}{p_1}{\beta_1 \bot}$ to $\config{\ell}{q}{\beta_1 \bot}$ and
$\config{ax\ell}{p_2}{\beta_2 \bot}$ to $\config{\ell}{q}{\beta_2 \bot}$
for some stack symbols $\beta_1, \beta_2$ such that $\symb{\beta_1} = \symb{\beta_2} = b$
and $\beta_1$ and $\beta_2$ are never popped during both runs.
Then $\mathcal{B}$ has two supports of the form
$p_1 \apush{a} p'_1 \ashift{} \dots p''_1 \apop{p_1} q$ and
$p_2 \apush{a} p'_2 \ashift{} \dots p''_2 \apop{p_1} q$.
Since the support graph is trim, there exists an $\omega$-word $w \in \Sigma^\omega$
starting with $\ell$ such that $w \in L_\mathcal{B}(q, \beta \bot)$
for some stack symbol $\beta$ with $\symb{\beta} = b$.
Then, while reading word $axw$, $\mathcal{B}$ can read $ax$ with both supports.
Thus, $axw \in L_\mathcal{B}(p_1, \beta \bot)$ and $axw \in L_\mathcal{B}(p_2, \beta \bot)$,
which contradicts Definition~\ref{def:separated}.
\end{proof}

\acs{POTL} formulas can be encoded as \ac{OPBA} with the construction by~\citet{ChiariMPP23}.
We prove that, when restricted to \acs{POTLF}, the resulting \ac{OPBA} are separated:
\begin{theorem}
\label{thm:potlf-separated-opba}
Given a \acs{POTLF} formula $\varphi$, a complete separated \ac{OPBA} $\mathcal{B}$
of size exponential in the length of $\varphi$ can be built such that $L_\mathcal{B} = L_\varphi$.
\end{theorem}
\begin{proof}[Proof (sketch)]
The \ac{OPBA} $\mathcal{B}$ is such that its states contain the set of subformulas of $\varphi$
that hold in the word position it is about to read, plus some bookkeeping information.
We prove by induction on the syntactic structure of $\varphi$ that,
given an edge $(q', b', \ell') \suppedge (q, b, \ell)$ of the support graph,
whether a subformula appears in $q'$ is fully determined by formulas
and bookkeeping information in $q$, together with $b$.
We leave the full proof to Appendix~\ref{sec:opba-construction-proof}.
\end{proof}

The same construction yields non-separated \ac{OPBA} for other \acs{POTL} operators;
we illustrate this fact with an example in Section~\ref{sec:potl-construction-not-sep} of the appendix.
We do not know if an alternative construction that yields separated \ac{OPBA} exists,
but we note that it would have to differ significantly from Vardi-Wolper-style tableaux ones
such as those for \acs{POTL}~\cite{ChiariMPP23}, CaRet~\cite{AlurEM04}, and \acs{NWTL}~\cite{lmcs/AlurABEIL08}.

\citet{CouvreurSS03} give a construction that translates purely future \acs{LTL} formulas into separated B\"uchi automata.
It can be easily adapted to \acp{OPBA}:
the sets of all, initial, and final states are the same as for B\"uchi automata,
and we add all push, shift, and pop transitions allowed by the \ac{OPM} according to the rules given in \cite{CouvreurSS03}.
With an argument similar to the proof of Theorem~\ref{thm:potlf-separated-opba},
we can show that \acp{OPBA} obtained in this way are separated.
\begin{theorem}
\label{thm:ltl-separated-opba}
Given a future \acs{LTL} formula $\varphi$, there exists a complete separated \ac{OPBA} $\mathcal{B}$
of size exponential in the length of $\varphi$ such that $L_\mathcal{B} = L_\varphi$.
\end{theorem}
Theorems \ref{thm:potlf-separated-opba} and \ref{thm:ltl-separated-opba}
allow us to use the algorithms for model checking \acp{pOPA} against \acp{OPBA} specifications
that we give in Sections \ref{sec:qualitative-mc} and \ref{sec:quantitative-mc}
to check \ac{POTLF} and \ac{LTL} properties.

\subsection{Qualitative Model Checking}
\label{sec:qualitative-mc}

In this section, we present the qualitative model checking algorithm,
and then show in Section~\ref{sec:quantitative-mc} the necessary additional steps
needed to solve the quantitative problem.

\subsubsection{Algorithm}
\label{sec:qualitative-mc-construction}

We fix a \ac{pOPA}
$\mathcal{A} = (\Sigma, \allowbreak M, \allowbreak Q_\mathcal{A}, \allowbreak u_0,
\allowbreak \delta^\mathcal{A}, \allowbreak \Lambda)$
encoding the program, and a complete separated \ac{OPBA}
$\mathcal{B} = (\Sigma, \allowbreak M, \allowbreak Q_\mathcal{B}, \allowbreak I,
\allowbreak F, \allowbreak \delta^\mathcal{B})$
encoding the specification, on the same \acs{OP} alphabet $(\Sigma, M)$.
We check whether
$P(\Lambda_\varepsilon(\mathit{Runs}(\mathcal{A})) \cap L_\mathcal{B}) = 1$.

\slimparagraph{Building graph $G$.}
The first step consists of building a graph $G$ that can be thought of a \emph{synchronized product}
between the support chain of $\mathcal{A}$ and the support graph of $\mathcal{B}$.

Let $\rho = (u_0, A_0) (u_1, A_1) \dots$ be a suffix of a run of $\mathcal{A}$,
and $\sigma(\rho) = (u_{i_0}, \alpha_{i_0}) \allowbreak (u_{i_1}, \alpha_{i_1}) \allowbreak \dots$
be a path in $\mathcal{A}$'s support graph, where $i_0 = 0$,
all $i_k$ for $k > 0$ are increasing indices of positions in $\rho$,
and each $\alpha_{i_h}$, $h \geq 0$, is the topmost symbol of $A_{i_h}$.
Consider a run \(\tau = \config{x_0x_1\dots}{p_0}{B_0} \config{x_1x_2\dots}{p_1}{B_1} \dots\)
of $\mathcal{B}$ on $\Lambda(\rho) = x_0x_1\dots$.
Since $\mathcal{A}$ and $\mathcal{B}$ share the same OP alphabet,
$\rho$ and $\tau$ are synchronized, i.e., they both perform the same kind of move
(push, shift, or pop) at the same time, and for each $i > 0$,
if we let $A_i = \alpha_0 \alpha_1 \dots \alpha_{n_i} \bot$ and $B_i = \beta_0 \beta_1 \dots \beta_{n_i} \bot$
we have $\symb{\alpha_j} = \symb{\beta_j}$ for all $0 \leq j \leq n_i$.
Thus, $\sigma(\rho)$ and $\sigma(\tau)$ are such that edges between
each pair of consecutive semi-configurations are of the same kind (push, shift or support).

Let $\sigma(\tau) = (p_{i_0}, b_{i_0}, \ell_{i_0}) (p_{i_1}, b_{i_1}, \ell_{i_1}) \dots$,
where $b_{i_h} = \symb{B_{i_h}}$ and $\ell_{i_h}$ is the look-ahead for all $h \geq 0$.
Suppose we know all of $\sigma(\tau)$ except $d_{i_0} = (p_{i_0}, b_{i_0}, \ell_{i_0})$.
Clearly we have $b_{i_0} = \symb{\alpha_{i_0}}$ and $\ell_{i_0} = \Lambda(u_{i_0})$.
Moreover, by Prop.~\ref{prop:sep-backward}
the support graph of $\mathcal{B}$ is backward deterministic,
and we can determine $p_{i_0}$ from $b_{i_0}$, $\ell_{i_0}$,
$(p_{i_1}, b_{i_1}, \ell_{i_1})$ and,
if $(u_{i_0}, \alpha_{i_0})$ and $(u_{i_1}, \alpha_{i_1})$ are linked by a support edge,
from the underlying chain, according to Def.~\ref{def:backward-det}.

Thus, we build a graph $G$ as the synchronized product of
the support chain $M_\mathcal{A}$ of $\mathcal{A}$ and the support graph of $\mathcal{B}$.
Let $(\mathcal{C}, \mathcal{E})$ be the support graph of $\mathcal{A}$.
Nodes of $G$ are pairs $((u, \alpha), p) \in \mathcal{C} \times Q_\mathcal{B}$.
For each node $((u_1, \alpha_1), p_1)$ of $G$,
we add an edge in $G$ between
$((u_0, \alpha_0), p_0)$ and $((u_1, \alpha_1), p_1)$
by choosing $p_0$ according to Def.~\ref{def:backward-det},
with $b = \Lambda(u_0)$ and $b' = \symb{\alpha_0}$.
While this is trivial for push and shift edges,
for support edges we find $p_0$ by means of reachability techniques for \acp{OPBA}, as we show next.
We build a synchronized product $\hat{\mathcal{A}}$ between $\mathcal{A}$ and $\mathcal{B}$,
which is an \ac{OPBA}
with states in $Q_\mathcal{A} \times Q_\mathcal{B}$, and transition relation $\hat{\delta}$ such that
\begin{itemize}
\item $((u, p), \Lambda(u), (u', p')) \in \hat{\delta}_\mathit{push}$ if
    $\delta^\mathcal{A}_\mathit{push}(u)(u') > 0$,
    and $(p, \Lambda(u), p') \in \delta^\mathcal{B}_\mathit{push}$;
\item $((u, p), \Lambda(u), (u', p')) \in \hat{\delta}_\mathit{shift}$ if
    $\delta^\mathcal{A}_\mathit{shift}(u)(u') > 0$,
    and $(p, \Lambda(u), p') \in \delta^\mathcal{B}_\mathit{shift}$;
\item $((u, p), (u'', p''), (u', p')) \in \hat{\delta}_\mathit{pop}$ if
    $\delta^\mathcal{A}_\mathit{pop}(u, u'')(u') > 0$,
    and $(p, p'', p') \in \delta^\mathcal{B}_\mathit{pop}$.
\end{itemize}
We then build the support graph of $\hat{\mathcal{A}}$ in polynomial time \cite{ChiariMPP23} and,
iff it contains a support edge between $((u_0, p_0), \#, \Lambda(u_0))$
and $((u_1, p_1), \#, \ell)$ for any $\ell \in \Sigma$,
then we know that $((u_1, p_1), \bot)$ is reachable from $((u_0, p_0), \bot)$.
Thus, we add to $G$ an edge between $((u_0, \alpha_1), p_0)$ and $((u_1, \alpha_1), p_1)$.
If any final state of $\mathcal{B}$ is visited while checking the support,
the edge is marked as \emph{final}.

\slimparagraph{Analyzing graph $G$.}
After defining graph $G$, we give three properties 
that can be practically checked
to determine whether the language generated by a \ac{BSCC} of $M_\mathcal{A}$
is accepted by $\mathcal{B}$ almost surely.

We define the projection on the first component of a node $(c, p)$ of $G$ as $\pi_\mathcal{A}(c, p) = c$.
We extend $\pi_\mathcal{A}$ to sets of nodes and paths of $G$ in the obvious way.
The following properties hold:

\begin{lemma}
\label{lemma:paths-G-supp}
1) For every infinite path $t$ of $G$, $\pi_\mathcal{A}(t)$ is a path of $M_\mathcal{A}$.\\
2) For every infinite path $r$ of $M_\mathcal{A}$ there is a path $t$ of $G$ such that $r = \pi_\mathcal{A}(t)$.
\end{lemma}
\begin{proof}
1) trivially follows from the definition of $G$.
Concerning 2), for each path $r$ in $M_\mathcal{A}$ there is some run $\rho$ of $\Delta(\mathcal{A})$
such that $r = \sigma(\rho)$.
Since $\mathcal{B}$ is complete, it has a final run $\tau$ on $\Lambda(\rho)$.
By construction, $G$ contains a path obtained by pairing together one-by-one
the semi-configurations in $\sigma(\rho)$ and $\sigma(\tau)$.
\end{proof}
Therefore, an \ac{SCC} of $M_\mathcal{A}$ yields one or more \acp{SCC} of $G$,
which are also \acp{SCC} of $\mathcal{B}$.

Let $K$ be a \ac{BSCC} of $M_\mathcal{A}$.
If we find an \ac{SCC} $\hat{C}$ of $G$ such that
\begin{enumerate*}[label=\textbf{(\alph*)}]
  \item \label{item:prop-a} $\pi_\mathcal{A}(\hat{C}) = K$,
  \item \label{item:prop-b} edges linking nodes of $\hat{C}$ cover all paths in $K$
(including all supports represented by support edges),
  \item \label{item:prop-c} $\hat{C}$ contains nodes or edges that are final for $\mathcal{B}$,
then we know that the labels of all runs of $\mathcal{A}$
generated by $K$ yield final runs of $\mathcal{B}$.
\end{enumerate*}
Theorem~\ref{thm:conditions-H} provides a practical way to find \acp{SCC} of $G$
that satisfy \ref{item:prop-a}--\ref{item:prop-c}:%
\begin{theorem}
\label{thm:conditions-H}
Let $K$ be a \ac{BSCC} of $M_\mathcal{A}$.
$G$ contains a unique \ac{SCC} $\hat{C}$ that satisfies \ref{item:prop-a}--\ref{item:prop-c}.
$\hat{C}$ is the only \ac{SCC} of $G$ that satisfies the following properties:
\begin{enumerate}[label=(\arabic*)]
\item \label{item:G-cond-proj}
  $\pi_\mathcal{A}(\hat{C}) = K$;
\item \label{item:G-cond-final}
  $\hat{C}$ contains at least one final node or edge;
\item \label{item:G-cond-ancestor}
  no other \ac{SCC} $C$ such that $K = \pi_\mathcal{A}(C)$ is an ancestor of $\hat{C}$.
\end{enumerate}
\end{theorem}

We justify Theorem~\ref{thm:conditions-H} through an intuitive argument.
We then give the full, more involved proof in Section~\ref{sec:qualitative-mc-proof}.

The existence of $\hat{C}$ follows from $\mathcal{B}$ being complete:
consider a run $\hat{\rho}$ of $\mathcal{A}$ such that $\sigma(\hat{\rho})$
includes all nodes in $K$ and edges connecting them.
Since $\mathcal{B}$ is complete, it has a run that is final for
$w = \Lambda_\varepsilon(\hat{\rho})$, which forms an \ac{SCC} of $\mathcal{B}$.

The uniqueness of $\hat{C}$ follows form $\mathcal{B}$ being separated:
suppose $G$ has another \ac{SCC} $\hat{C}' \neq \hat{C}$
that satisfies conditions \ref{item:prop-a}--\ref{item:prop-c}.
Due to conditions \ref{item:prop-a} and \ref{item:prop-b}, $\hat{C}'$ and $\hat{C}$ must differ
in states of $\mathcal{B}$ their nodes contain.
This means that $w$ has two final runs
$\tau_1$ and $\tau_2$ of $\mathcal{B}$ whose $\sigma$-images differ in at least
two semi-configurations $d_1 = (p_1, b_1, \ell_1)$ and $d_2 = (p_2, b_2, \ell_2)$.
We have $b_1 = b_2$ and $\ell_1 = \ell_2$ because they only depend on $w$, so $p_1 \neq p_2$.
Thus, for some stack symbols $\beta_1, \beta_2$ with $\symb{\beta_1} = \symb{\beta_2}$
we have
$w \in L_\mathcal{A}(p_1, \beta_1 \bot)$ and $w \in L_\mathcal{A}(p_2, \beta_2 \bot)$,
which contradict Def.~\ref{def:separated}.

Condition \ref{item:prop-b} is ensured by property \ref{item:G-cond-ancestor}.
The proof employs a result from \cite[Theorem\ 5.10]{CourcoubetisY95} and is quite involved,
but we give an intuition in Example~\ref{exa:mc}.

\begin{figure}[bt]
\centering
\begin{tikzpicture}
  [state/.style={draw, ellipse, inner sep=1pt}, >=latex, font=\scriptsize,
   final/.style={very thick}
  ]
\node[state, initial by arrow, initial text=] (n0) {$c_0,q_0$};
\node[state] (n6) [right=15pt of n0, final] {$c_1,q_1$};
\node[state] (n7) [below of=n6, final] {$c_2,q_1$};
\node[state] (n1) [right=15pt of n6] {$c_1,q_0$};
\node[state, initial right, initial text=, final] (n8) [right=35pt of n1] {$c_0,q_3$};
\node[state, final] (n9) [below of=n8] {$c_1,q_3$};
\node[state] (n2) [below of=n1] {$c_2,q_0$};
\node[state] (n3) [below of=n2, final] {$c_2,q_3$};
\node[state] (n4) [below of=n0, final] {$c_3,q_3$};
\node[state] (n5) [below of=n4, final] {$c_4,q_3$};
\path[->] (n0) edge[out=20, in=160] (n1)
          (n0) edge (n6)
          (n6) edge[out=20, in=165, below, pos=.4] node[yshift=1pt]{$e_1$} (n1)
          (n1) edge[final, out=195, in=340, below, pos=.9] node[yshift=1pt]{$e_2$} (n6)
          (n6) edge[out=240, in=215, loop, left] node{$e_1$} (n6)
          (n1) edge[out=10, in=350, loop, final, above] node{$e_2$} (n1)
          (n1) edge[right] node{$e_3$} (n2)
          (n1) edge[right] node{$e_3$} (n7)
          (n2) edge[out=10, in=350, loop, final, above] node{$e_2,e_3$} (n2)
          (n2) edge[final, out=195, in=340, below, pos=.1] node[yshift=1pt]{$e_2,e_3$} (n7)
          (n7) edge[out=20, in=165, below, pos=.4] node[yshift=1pt]{$e_1$} (n2)
          (n7) edge[out=240, in=215, loop, left] node{$e_1$} (n7)
          (n7) edge[below] node{$e_1$} (n3)
          (n2) edge[final, right] node{$e_2$} (n3)
          (n3) edge[out=10, in=350, loop, above] node[yshift=-12pt]{$e_3$} (n3)
          (n0) edge[final] (n4)
          (n4) edge (n5)
          (n5) edge[out=170, in=190, loop] (n5)
          (n8) edge (n9)
          (n9) edge node[xshift=10pt]{$e_3$} (n3);
\end{tikzpicture}
\caption{Graph $G$ obtained from $\mathcal{A}_\lcall$'s support chain (Fig.~\ref{fig:supp-chain}) and $\mathcal{B}_\lcall$'s trim support graph (Fig.~\ref{fig:opba-support-graph}, right).
  Final nodes and edges are drawn with a thick line.
  We only include the non-trivial part of $G$.
  The nodes obtained by combining the initial configuration $c_0$ of $\mathcal{A}_\lcall$
  with the other two states $q_1, q_2$ of $\mathcal{B}_\lcall$ do not reach any \ac{SCC}: $(c_0, q_2)$ has no outgoing edge, and $(c_0, q_1)$ is connected only with $(c_1,q_2)$, which has no outgoing edges.
}
\label{fig:graph-G}
\end{figure}
\begin{example}[Running example, cont. \ref{running-ex:separated}]
\label{exa:mc}
$\mathcal{B}_\lcall$ (Fig.~\ref{fig:opba-example}) is separated,
and its language satisfies \acs{POTLF} formula $\lcduntil{\top}{(\lcall \land \ldnext \lret)}$.
We do not use the \ac{OPBA} obtained with the construction of \cite{ChiariMPP23},
because $\mathcal{B}_\lcall$ is smaller and more polished.
We check $\mathcal{A}_\lcall$ against $\mathcal{B}_\lcall$:
the resulting graph $G$ is shown in Fig.~\ref{fig:graph-G}.
Edges labeled with $e_1$ come from simple supports
$u_1 \apush{} u_2 \ashift{} u_1 \apop{u_1} u_1$ of $\mathcal{A}_\lcall$
and $q_1 \apush{\lcall} q_2 \ashift{\lret} q_i \apop{q_1} q_i$, with $i \in \{0,1\}$,
of $\mathcal{B}_\lcall$;
edges labeled with $e_2$ come from the composed supports
$u_1 \apush{} u_1 \asupp{} u_2 \ashift{} u_1 \apop{u_1} u_1$ of $\mathcal{A}_\lcall$
and $q_0 \apush{\lcall} q_1 \asupp{} q_2 \ashift{\lret} q_i \apop{q_0} q_i$, with $i \in \{0,1\}$,
of $\mathcal{B}_\lcall$;
edges labeled with $e_3$ come the push self-loop on $u_1$
and push moves reading $\lcall$ in $\mathcal{B}_\lcall$.

\ac{BSCC} $\{c_4\}$ of the support chain appears only once,
paired with a final state of $\mathcal{B}_\lcall$,
and is reachable only from the initial node $(c_0, q_0)$.
$\{c_2\}$ appears in the two \acp{SCC} $\hat{C} = \{(c_2, q_0), (c_2, q_1)\}$
and $\hat{C}' = \{(c_2, q_3)\}$.
$\hat{C}$ precedes $\hat{C}'$
and contains all kinds of edges $e_1$, $e_2$ and $e_3$,
while $\hat{C}'$ lacks $e_1$ and $e_2$.
This is no coincidence: since $\{c_2\}$ is a \ac{BSCC},
any node in $G$ descending from it must still contain $c_2$.
Due to backward determinism, only one edge of type $e_2$ is incident in $(c_2, q_3)$,
and it is the one connecting $\hat{C}$ to $\hat{C}'$: it cannot be internal to $\hat{C}'$.
Thus, $\{c_2\}$ is accepted by states $q_1$ and $q_2$, not $q_3$.

The \acp{SCC} that satisfy properties \ref{item:G-cond-proj}--\ref{item:G-cond-ancestor} are $\{ c_4 \}$ and $\hat{C}$. The only initial node reaching them is
$(c_0, q_0)$.
In particular, they are not reachable from $(c_0, q_3)$ nor $(c_0, q_2)$,
therefore $\mathcal{A}_\lcall$ satisfies the property encoded by $\mathcal{B}_\lcall$ almost surely.
\end{example}

In conclusion, the algorithm for qualitative model checking proceeds as follows:
\begin{itemize}
\item build $M_\mathcal{A}$ and find its \acp{BSCC};
\item build $G$;
\item find the \acp{SCC} of $G$ that satisfy properties
  \ref{item:G-cond-proj}--\ref{item:G-cond-ancestor};
\item $P(\Lambda_\varepsilon(\mathit{Runs}(\mathcal{A})) \cap L_\mathcal{B}) = 1$ iff none of the \acp{SCC} detected in the previous step is reachable from a node $((u_0, \bot), p)$ with $p \notin I$,
where $u_0$ is the initial state of $\mathcal{A}$, and $I$ is the set of initial states of $\mathcal{B}$.
\end{itemize}
Building $M_\mathcal{A}$ requires computing termination probabilities
to determine which semi-con\-fi\-gu\-ra\-tions belong to $M_\mathcal{A}$, so it lies in \textsc{pspace}.
The remaining steps are polynomial in the sizes of $\mathcal{A}$ and $\mathcal{B}$.

\subsubsection{Correctness Proof}
\label{sec:qualitative-mc-proof}

We now prove the correctness of the qualitative model checking algorithm more formally.
Our proof is related to the one for \acs{LTL} qualitative model checking of \acp{RMC} in \cite[Section 7]{EtessamiY12},
but it differs significantly because we target \acp{pOPA} and \acp{OPBA} instead of \acp{RMC} and \acs{LTL}.
In particular, the model checking algorithm by \citet{EtessamiY12}
employs an \emph{ad hoc} bit-vector-based construction for \acs{LTL},
making it significantly less general than ours,
which targets the more expressive automata class of separated \acp{OPBA}.


Recall that we consider a \ac{pOPA} $\mathcal{A}$, its support chain $M_\mathcal{A}$,
a complete separated \ac{OPBA} $\mathcal{B}$, its support graph $(\mathcal{C}, \mathcal{E})$,
and the graph $G$ built as described in Section~\ref{sec:qualitative-mc-construction}.

For nodes $(c, p) \in G$ we define $\pfx(c, p)$ as the probability of runs $\rho$ of $\Delta(\mathcal{A})$ such that
(i) an instance of $c$ is pending in $\rho$,
and (ii) the word labeling $\rho$ yields a final run $\tau$ in $\mathcal{B}$
passing through state $p$ at the same time $\rho$ passes through such instance of $c$.
More formally, given $c = (q, \alpha) \in \mathcal{C}$ and $p \in Q_\mathcal{B}$,
and a run $\rho = \rho_0 \rho_1 \dots \rho_i \dots$ where $\rho_i = (q, \alpha A)$ for some $i \geq 0$,
$\pfx(c, p)$ is the probability that no symbol in $\alpha A$ is ever popped for all $\rho_j$, $j \geq i$,
and that $\Lambda_\varepsilon(\rho_i \rho_{i+1} \dots) \in L_\mathcal{B}(p, \beta \bot)$
for some $\beta \in \Gamma_\bot^\mathcal{B}$ such that $\symb{\beta} = \symb{\alpha}$.

\begin{lemma}
\label{lemma:nex-in-H}
If $G$ contains an edge $(c_1, p_1) \gedge (c_2, p_2)$ such that $\pfx(c_2, p_2) > 0$,
then $\pfx(c_1, p_1) > 0$.
\end{lemma}
\begin{proof}
If $G$ contains such an edge, then $M_\mathcal{A}$ contains an edge $c_1 \suppedge c_2$,
which can be taken with positive probability,
i.e., $\mathcal{A}$ contains a transition (if the edge is a push or shift edge)
or a summary (if the edge is a summary edge) that has positive probability $x$ of being taken
from semi-configuration $c_1$.
Thus, if $\pfx(c_2, p_2) > 0$, then $\nex{c_2} > 0$, and $\nex{c_1} \geq x \cdot \nex{c_2} > 0$.
Since, by construction of $G$, the support graph of $\mathcal{B}$ does have an edge $p_1 \suppedge p_2$,
then $(c_1, p_1) \gedge (c_2, p_2)$ can be taken with positive probability,
and $\pfx(c_1, p_1) > 0$.
(Note that, in general, $\pfx(c_1, p_1) \leq \nex{c_1}$,
because the weight of the edge $(c_1, p_1)$ may be shared among multiple edges in $G$.)
\end{proof}  

From Lemma~\ref{lemma:nex-in-H} immediately follows:
\begin{corollary}
\label{cor:scc-h}
In every \ac{SCC} of $G$, $\pfx$ is positive for either all of its nodes or none.
\end{corollary}

Thus, nodes of $G$ with positive $\pfx$ are those in an \ac{SCC} where all nodes have positive $\pfx$,
and their ancestors.
We use $\pfx$ to formalize conditions \ref{item:prop-a}--\ref{item:prop-c} from Section~\ref{sec:qualitative-mc}.
In particular, nodes of $G$ contained in \acp{SCC} that satisfy \ref{item:prop-a}--\ref{item:prop-c} and their ancestors
are exactly those with $\pfx > 0$.
We call $H$ the sub-graph of $G$ obtained by removing nodes with $\pfx = 0$.
$H$ can be built by finding bottom \acp{SCC} of $G$ with $\pfx > 0$ and all their ancestors.
If all nodes of the form $((u_0, \bot), q_0)$ in $H$,
where $u_0$ is the initial state of $\mathcal{A}$,
are such that $q_0$ is an initial state of $\mathcal{B}$,
then almost all runs of $\Delta(\mathcal{A})$ satisfy the specification encoded by $\mathcal{B}$.
This fact can be used to decide the qualitative model checking problem.

Thus, we need a way to identify nodes of $G$ that are part of $H$.
We prove that Theorem~\ref{thm:conditions-H} gives three necessary and sufficient conditions:

\begin{lemma}
\label{lemma:conditions-H}
Let $C$ be an \ac{SCC} of $G$.
The following conditions are equivalent:
\begin{enumerate}[label=(\Roman*)]
  \item \label{item:scc-cond-H}
    $C$ is a \ac{BSCC} of $H$;
  \item \label{item:scc-cond-G}
    There exists a \ac{BSCC} $K$ of $M_\mathcal{A}$ such that $C$ is the only \ac{SCC} of $G$
    that satisfies properties \ref{item:G-cond-proj}--\ref{item:G-cond-ancestor} of Theorem~\ref{thm:conditions-H};
  \item \label{item:scc-cond-a-c}
    There exists a \ac{BSCC} $K$ of $M_\mathcal{A}$ such that $C$ is the only \ac{SCC} of $G$
    that satisfies properties \ref{item:prop-a}--\ref{item:prop-c}.
\end{enumerate}
\end{lemma}

We start by proving the equivalence of conditions \ref{item:scc-cond-H} and \ref{item:scc-cond-G}.

\begin{lemma}[\ref{item:scc-cond-H} implies \ref{item:G-cond-proj}]
\label{lemma:H-cond-proj}
If $C$ is a bottom \ac{SCC} of $H$, then $\pi_\mathcal{A}(C)$ is a bottom \ac{SCC} of $M_\mathcal{A}$.
\end{lemma}
\begin{proof}
By the construction of $G$, it is clear that for every \ac{SCC} $C$ of $G$ (and hence $H$)
the projection of all its nodes on the first component yields an \ac{SCC}
$K = \pi_\mathcal{A}(\hat{C})$ of $M_\mathcal{A}$.

Consider a node $(c, p)$ of $C$.
We have $\pfx(c, p) > 0$, i.e., there is positive probability that $c$ appears pending in a run $\rho$ of $\Delta(\mathcal{A})$,
and the word label of $\rho$ starting from $c$ yields a final run of $\mathcal{B}$ starting from $p$.
Thanks to Theorem~\ref{thm:support-chain},
we know that $\sigma(\rho)$ is absorbed into a bottom \ac{SCC} $K$ of $M_\mathcal{A}$.
However, $\sigma(\rho)$ must be the projection of a path in $G$ (Lemma~\ref{lemma:paths-G-supp})
which is absorbed in an \ac{SCC} $C'$ of $G$ which, by Corollary~\ref{cor:scc-h},
is also a bottom \ac{SCC} in $H$.
Since $C'$ and $C$ share $(c, p)$, they must be the same \ac{SCC},
and their projection $\pi_\mathcal{A}(C)$ is the bottom \ac{SCC} $K$ of $M_\mathcal{A}$.
\end{proof}

\begin{lemma}[\ref{item:scc-cond-H} implies \ref{item:G-cond-final}]
\label{lemma:H-final}
Every node of $H$ can reach a final node or edge in $H$.
\end{lemma}
\begin{proof}
If a node $(c, p)$ of $G$ is in $H$,
then there is a positive probability that semi-configuration $c$ appears pending in a run $\rho$ of $\Delta(\mathcal{A})$,
and $\Lambda(\rho)$ yields a final run $\tau$ in $\mathcal{B}$.
If $\tau$ is final, then it visits a final state of $\mathcal{B}$ infinitely often.
Recall that $\sigma(\rho)$ and $\sigma(\tau)$ are synchronized,
i.e., edges taken by $\sigma(\rho)$ and $\sigma(\tau)$ in $\mathcal{A}$'s and $\mathcal{B}$'s
respective support graphs have the same type.
Thus, $\sigma(\rho)$ and $\sigma(\tau)$ can be paired node-by-node, yielding a path in $G$.
All nodes of this path represent a final run of $\mathcal{B}$,
because each suffix of a final run of $\mathcal{B}$ is also final,
due to the B\"uchi acceptance condition.
Hence they are also in $H$, and the path must be absorbed by one of $H$'s bottom \acp{SCC}.
If a final state of $\mathcal{B}$ appears in $\sigma(\tau)$,
then it also appears in the \ac{SCC}, and we are done.
Otherwise, it must be part of a support represented by an edge between two nodes of the \ac{SCC} in $H$.
\end{proof}
As a consequence of Lemma~\ref{lemma:H-final}, all \acp{BSCC} of $H$
contain a final node or edge.
This completes the proof that \ref{item:scc-cond-H} implies \ref{item:G-cond-final}.

The proof that condition \ref{item:scc-cond-H} implies
property~\ref{item:G-cond-ancestor} of Theorem~\ref{thm:conditions-H} is more involved,
because the backward determinism condition of the support graph of $\mathcal{B}$
implies that for each support edge $c_1 \ssupp c_2$ in $M_\mathcal{A}$
and node $(c_2, p_2)$ in $G$ there may be multiple $p$'s such that
$(c_1, p) \ssupp (c_2, p_2)$ is an edge in $G$.
Moreover, transitions in $M_\mathcal{A}$ may represent both push and support edges,
which may result in separate edges in $G$.
Thus, we define a multigraph $M'_\mathcal{A}$ that has the same nodes,
push and shift edges as $M_\mathcal{A}$, but some edges may appear with a multiplicity higher than one.
Consider a support edge $(c_1, c_2)$ of $M_\mathcal{A}$, and two runs $\rho_1$ and $\rho_2$
of $\mathcal{A}$ that link $c_1$ and $c_2$ with a chain support.
Given a semi-configuration $d_2$ of $\mathcal{B}$, since $\mathcal{B}$ is separated,
there is only one semi-configuration $d_1$ that is linked by $\Lambda_\varepsilon(\rho_1)$ to $d_2$.
Thus, we define an equivalence relation between runs so that $\rho_1$ and $\rho_2$
are equivalent iff, for each semi-configuration $d$ of $\mathcal{B}$,
$\Lambda_\varepsilon(\rho_1)$ links to $d$ the same semi-configuration as $\Lambda_\varepsilon(\rho_2)$.
For each support edge $(c_1, c_2)$ of $M_\mathcal{A}$, in $M'_\mathcal{A}$ $c_1$ and $c_2$ are linked
by as many edges as the number of equivalence classes of runs of $\mathcal{A}$ between $c_1$ and $c_2$.
We proceed similarly if $(c_1, c_2)$ also represents a push edge
(in this case $\rho_1$ and $\rho_2$ only consist of a push move).
This way, for each edge $(c_1, c_2)$ of $M'_\mathcal{A}$ and semi-configuration $d_2$ of $\mathcal{B}$,
there is exactly one semi-configuration $d_1$ of $\mathcal{B}$ linked to $d_2$,
and the support graph of $\mathcal{B}$ is backward deterministic with respect to
all edges of $M'_\mathcal{A}$ (not just pure push and shift edges).
Since $\mathcal{B}$ is complete, each string that is a support of $\mathcal{A}$ between
$c_1$ and $c_2$ is represented by one of the edges between them in $M'_\mathcal{A}$.
Thus, the probability associated with the support edge between $c_1$ and $c_2$ in $M_\mathcal{A}$
is divided among all edges between them in $M'_\mathcal{A}$.

\begin{lemma}
\label{lemma:SCC-ancestor}
Let $C$ be an \ac{SCC} of $G$ and $K = \pi_\mathcal{A}(C)$ the corresponding \ac{SCC} of $M'_\mathcal{A}$.
The following statements are equivalent.
\begin{enumerate}[label=(\roman*)]
\item \label{item:all-edges-in-C}
For each edge $c_1 \gedge c_2$ of $K$ and node $(c_2, p_2)$ of $C$,
there is a state $p_1$ such that the edge $(c_1, p_1) \gedge (c_2, p_2)$ is in $C$.
\item \label{item:paths-K-C}
Every finite path in $K$ is the projection on the first component of some path in $C$.
\item \label{item:C-no-ancestor}
No \ac{SCC} $C'$ of $G$ such that $K = \pi_\mathcal{A}(C')$ is an ancestor of $C$.
\end{enumerate}
\end{lemma}
\begin{proof}
This proof employs an idea from \cite[Theorem 5.10]{CourcoubetisY95}.
In \cite{CourcoubetisY95}, the proof relies on the fact that for each node $(c_2, p_2)$ in $G$,
if there is an arc $(c_1, c_2)$ in $K$ there is exactly one $d_1$
such that $G$ has an arc between $(c_1, p_1)$ and $(c_2, p_2)$.
This is not true in our case, because different support edges $(c_1, c_2)$ in $M'_\mathcal{A}$
may link several nodes $(c_1, p_1)$ with different $p_1$'s to the same $(c_2, p_2)$.
However, the argument still holds thanks to the universal quantification on edges between nodes in $K$
in \ref{item:all-edges-in-C}.
If we fix one of such edges, then exactly one node $(c_1, p_1)$ is linked to each node $(c_2, p_2)$ of $G$.

The proof that \ref{item:all-edges-in-C} is equivalent to \ref{item:paths-K-C} remains roughly the same,
with the difference that anytime a path in $K$ is fixed, it cannot be identified
with just the sequence of nodes it visits, but also by the edges it takes.

The equivalence between \ref{item:all-edges-in-C} and \ref{item:C-no-ancestor} has some more caveats,
so we re-prove it here.
Suppose that \ref{item:all-edges-in-C} is false, and for some node $(c_2, p_2)$ in $C$,
$G$ has an arc $(c_1, p_1) \rightarrow (c_2, p_2)$ such that $(c_1, p_1)$ is not in $C$ despite
$(c_1, c_2)$ being an edge in $K$.
All edges incoming and outgoing from $c_1$ must have an edge resp.\ incoming and outgoing
from $(c_1, p_1)$ in $G$, and so do all edges of $M'_\mathcal{A}$ reachable from them.
Since $c_1$ is in $K$, it is connected to all other nodes in $K$,
and $(c_1, d_1)$ must be part of another \ac{SCC} $D$ of $G$ such that $K = \pi_\mathcal{A}(D)$.
The edge $(c_1, p_1) \rightarrow (c_2, p_2)$ connects $D$ to $C$, so $D$ is an ancestor of $C$.

Conversely, suppose $C$ has an ancestor \ac{SCC} $D$.
Then $G$ contains some edge $(c_1, p_1) \rightarrow (c_2, p_2)$ such that $(c_2, p_2)$ is in $C$
but $(c_1, p_1)$ is not, and there is a path from a node $(c_3, p_3)$ in $D$ to $(c_1, p_1)$.
Node $(c_1, p_1)$ is linked to $(c_2, p_2)$ by one or more edges
between $c_1$ and $p_1$ in $M'_\mathcal{A}$.
However, all of such edges associate $(c_2, p_2)$ to $(c_1, p_1)$ only,
and since $(c_1, p_1)$ is not part of $C$, property \ref{item:all-edges-in-C} is violated for them.
\end{proof}

The following lemma is instrumental in proving that condition \ref{item:scc-cond-H} implies
property~\ref{item:G-cond-ancestor}:
\begin{lemma}
\label{lemma:all-in-C}
Let $C$ be a \ac{SCC} of $G$ that satisfies properties \ref{item:G-cond-proj}--\ref{item:G-cond-ancestor} of Theorem~\ref{thm:conditions-H}, and $K = \pi_\mathcal{A}(C)$.
Then each node $(c, p) \in G$ such that $c \in K$ and $\pfx(c, p) > 0$ is in $C$.
\end{lemma}
\begin{proof}
By Lemma~\ref{lemma:conditions-H}, $C$ contains a final node or edge.
First, suppose it is a node $(c^f, p^f)$.
We first prove that $\pfx(c^f, p^f) > 0$.
Suppose for the sake of contradiction that $\pfx(c^f, p^f) = 0$.
$K$ is a bottom \ac{SCC} of $M_\mathcal{A}$, hence runs starting from $c^f$ occur with positive probability.
Thus, the only way of having $\pfx(c^f, p^f) = 0$ is if almost all runs starting from it
are labeled with a word not accepted by $\mathcal{B}$.
Node $(c^f, p^f)$ is reachable from all nodes in $C$ through a path
whose projection on the first component is one or more paths in $M'_\mathcal{A}$ which only visit nodes in $K$.
Since $K$ is a bottom \ac{SCC} of $M_\mathcal{A}$,
one of these paths is eventually visited almost surely, and thus $(c^f, p^f)$ is visited infinitely often,
which contradicts the assumption that almost all runs starting from it are not accepted by $\mathcal{B}$.
Hence, almost all runs starting from $(c^f, p^f)$ must at some point leave $C$.
Take any of these paths, and let $(c^e, p^e)$ be its first node outside $C$.
Since $K$ is a bottom \ac{SCC}, $c^e$ must still be in $K$,
and there is a path in $M'_\mathcal{A}$ between $c^f$ and $c^e$.
By property \ref{item:paths-K-C} of Lemma~\ref{lemma:SCC-ancestor},
there is another path in $C$ that is the projection on the first component of the former.
Hence, for each run $\rho$ of $\mathcal{A}$ starting from $c^f$,
there exists a path in $C$ whose projection is $\sigma(\rho)$,
which visits $(c^f, p^f)$ infinitely often almost surely.
Thus, $\pfx(c^f, p^f) > 0$.

Let $(c, p) \in G$ be any node such that $c = (u, \alpha) \in K$ and $\pfx(c, p) > 0$.
Then there is a run $\rho$ of $\mathcal{A}$ starting from $c$ such that
$\Lambda_\varepsilon(\rho) \in L_\mathcal{B}(p, \beta \bot)$
for some stack symbol $\beta$ such that $\symb{\beta} = \symb{\alpha}$.
$\sigma(\rho)$ is the projection of a path $\theta$ in $C$ starting from $(c, p)$.
The case $(c, p) = (c^f, p^f)$ is trivial.
If this is not the case, but $\theta$ reaches $(c^f, p^f)$,
then we can prove by induction that $(c, p) \in C$,
because the path between $c$ and $c^f$ must be in $K$.

Otherwise, suppose $\theta$ never reaches $(c^f, p^f)$.
Then $\theta$ reaches another \ac{SCC} $C'$ of $G$ disjoint from $C$.
Since $K$ is a bottom \ac{SCC}, $\pi_\mathcal{A}(\theta)$ must visit all nodes in $K$ with probability 1.
Thus, $K = \pi_\mathcal{A}(C')$.
Since $K$ is a bottom \ac{SCC}, the whole $\sigma(\rho)$ must only contain nodes in $K$.
So, $\sigma(\rho)$ is also the projection on the first component of some path in $C$
(or property \ref{item:paths-K-C} of Lemma~\ref{lemma:SCC-ancestor}
would be violated at some point in $\sigma(\rho)$).
Let $(c, p') \in C$ be the first node of this path.
Since paths whose projection is in $C$ visit $(c^f, p^f)$ infinitely often with probability 1,
we have $\Lambda_\varepsilon(\rho) \in L_\mathcal{B}(q', \beta' \bot)$,
for some stack symbol $\beta'$ such that $\symb{\beta'} = \symb{\alpha}$.
But since $\mathcal{B}$ is separated, we must have $p' = p$,
which contradicts the claim that $C'$ and $C$ are disjoint.

The case in which $C$ contains no final node, but a final edge $(c^f_1, p^f_1) \gedge (c^f_2, p^f_2)$
can be analyzed by noting that the set of finite words it represents is captured by one or more edges between
$c^f_1$ and $c^f_2$ in $M'_\mathcal{A}$ that can all be taken with positive probability.
All arguments made for $(c^f, p^f)$ can then be adapted by considering
paths in $G$ that start with $(c^f_1, p^f_1) \gedge (c^f_2, p^f_2)$ or reach $(c^f_1, p^f_1)$.
After reaching $(c^f_1, p^f_1)$, the final edge can then be taken with positive probability, reaching $(c^f_2, p^f_2)$.
Thus, we can prove that the final edge $(c^f_1, p^f_1) \gedge (c^f_2, p^f_2)$ is visited infinitely often almost surely
by runs that are projections of a path in $C$, and the rest of the proof can be adapted consequently.
\end{proof}

We can finally prove Lemma~\ref{lemma:conditions-H}.
\begin{proof}
We already proved that condition \ref{item:scc-cond-H} implies
properties \ref{item:G-cond-proj} and \ref{item:G-cond-final} of Theorem~\ref{thm:conditions-H}
in Lemmas~\ref{lemma:H-cond-proj} and \ref{lemma:H-final}.
%
By contradiction, suppose that $C$ is a \ac{BSCC} of $H$,
but property \ref{item:G-cond-ancestor} does not hold:
then by Lemma~\ref{lemma:SCC-ancestor} there is a path $\tau$ within nodes in $K = \pi_\mathcal{A}(C)$
that is the projection of no path in $C$.
Since $K$ is a bottom \ac{SCC} of $M_\mathcal{A}$, and hence of $M'_\mathcal{A}$,
$\tau$ is visited infinitely often almost surely by every run of $\Delta(\mathcal{A})$ passing through a node in $K$.
Thus, a path through any node of $C$ will almost surely eventually reach a node in $\tau$ that is the projection
of a node $(c, p)$ of $G$ that is not in $C$.
However, $\pfx(c, p) > 0$, so $(c, p)$ is also a node of $H$ reachable from $C$,
which contradicts the assumption that $C$ is a \ac{BSCC} of $H$.
Thus, we have proved that condition \ref{item:scc-cond-H} implies
properties \ref{item:G-cond-proj}--\ref{item:G-cond-ancestor} of Theorem~\ref{thm:conditions-H}.

The converse follows from Lemma~\ref{lemma:all-in-C}:
since $C$ contains all nodes with positive $\pfx$,
it must be the only \ac{SCC} of $H$ whose projection on the first component is $K$.
In fact, suppose another \ac{SCC} $D$ of $G$ contains a node $(c, p)$ with $\pfx(c, p) > 0$ and $c \in K$.
Since $\pfx(c, p) > 0$, $(c, p)$ must be also in $C$, and since $C$ is strongly connected, $C \subseteq D$.
Similarly, since $D$ is also strongly connected, $D \subseteq C$.

Further, suppose $C$ is not bottom.
Then there is another \ac{SCC} $D$ of $H$ reachable from it.
However, $\pi_\mathcal{A}(C) = K$, and $K$ is a bottom \ac{SCC}.
So, $\pi_\mathcal{A}(D) \subseteq K$, which implies $D = C$ by our previous argument.

We conclude by proving that conditions \ref{item:scc-cond-G} and \ref{item:scc-cond-a-c} are equivalent.
Equivalence of \ref{item:prop-a} and \ref{item:prop-c} to, respectively,
properties \ref{item:G-cond-proj} and \ref{item:G-cond-final} is trivial.
Property \ref{item:prop-b} is equivalent to \ref{item:G-cond-ancestor} due to the equivalence of properties
\ref{item:paths-K-C} and \ref{item:C-no-ancestor} of Lemma~\ref{lemma:SCC-ancestor},
and the fact that $M'_\mathcal{A}$ has edges covering all possible supports and push moves,
and states visited by $\mathcal{B}$ while reading the strings they represent.
\end{proof}

\subsection{Quantitative Model Checking}
\label{sec:quantitative-mc}


First, we apply the qualitative model checking algorithm
described in Section~\ref{sec:qualitative-mc} to build graph $H$.
Performing quantitative model checking amounts to computing the probability mass of paths in $H$ that are accepted by $\mathcal{B}$.
More formally, for each semi-configuration of the support chain $c = (u, \alpha) \in \mathcal{C}$ and $q \in Q_\mathcal{B}$
we consider probabilities $\nex{c, q}$. 
$\nex{c, q}$ is the probability mass of runs $\rho$ of $\mathcal{A}$
such that $\Lambda_\varepsilon(\rho) \in L_\mathcal{B}(q, \beta\bot)$ for some $\beta$ with $\symb{\beta} = \symb{\alpha}$,
among those who start in semi-configuration $c$ and for which $c$ is pending.

We show that these probabilities are the solutions of a system of polynomial equations.
Then, if we want to check whether $P(\Lambda_\varepsilon(\mathit{Runs}(\mathcal{A})) \cap L_\varphi) \geq \varrho$
for a rational number $\varrho$, we check whether $\sum_{q \in I_\mathcal{B}} \nex{(u_0,\bot), q} \geq \varrho$
with a solver for \ac{ETR}.

In this system, we must identify the probabilities of edges in $H$.
We cannot just re-use probabilities from the support chain $M_\mathcal{A}$
because each support edge may appear in $H$ multiple times.
Thus, we must compute the weight of each occurrence of each support edge in $H$
as a fraction of its corresponding edge in $M_\mathcal{A}$.

We start by encoding the probabilities $\nex{c}$ for all $c \in \mathcal{C}$.
We do so by adding the equation (1a) $\mathbf{v} = f(\mathbf{v})$ of termination probabilities
for all semi-configurations and states of $\mathcal{A}$,
together with (1b) $\mathbf{v} \geq 0$.
Then, we encode each probability $\nex{c}$ as variable $y_c$
according to \eqref{eq:nex}, Section~\ref{sec:support-chain}.
For each $c = (s,\alpha) \in \mathcal{C}$ we add:
\begin{align*}
\text{(1c)} \quad y_c = 1 - \sum_{v \in Q_\mathcal{A}} \pvar{s}{\alpha}{v},
&&
\text{(1d)} \quad
\begin{cases}
y_c > 0 & \text{if $\nex{c} > 0$} \\
y_c = 0 & \text{if $\nex{c} = 0$}
\end{cases}
\end{align*}
Cases in (1d) are according to the qualitative analysis that we did while building the support chain.

Recall that for nodes $(c, q) \in H$, $\pfx(c, q)$ is the probability
that an instance of $c$ is pending in a run $\rho$ of $\Delta(\mathcal{A})$,
and that the word labeling such a run yields a final run $\tau$ in $\mathcal{B}$
passing through state $q$ at the same time $\rho$ passes through $c$.
Then, we introduce variables $z_{c,q}$ for all $(c, q) \in H$ to encode the probabilities $\pfx(c, q)$
\emph{conditioned} on the fact that a run of $\Delta(\mathcal{A})$ passing through $c$ does not terminate.
Precisely, we add a set of constraints such that the (unique) solution associated with each variable $z_{c,q}$ 
is the quantity $\frac{\pfx(c, q)}{\nex{c}} = \nex{c, q}$.

Clearly, for each $c \in \mathcal{C}$, the probability $\nex{c}$ is shared among all nodes $(c, p)$ of $H$ with $p \in Q_\mathcal{B}$:
hence, $\sum_{p \in Q_\mathcal{B}} \pfx(c,p) = \nex{c}$, from which we derive
\[
\text{(2a)} \quad
\sum_{p \in Q_\mathcal{B}} z_{c,p} = 1.
\]
Other constraints depend on the type of edges outgoing from each node $(c, q)$ of $H$.
Let $c = (u, \alpha)$.
\begin{itemize}
\item If $\Lambda(u) \doteq \symb{\alpha}$, only \emph{shift} edges are outgoing from $(c, q)$.
Thus, $\pfx(c, q)$ represents the probability that a run takes one of such edges targeting, say, node $(c', q')$,
and then continues from there without terminating.
We can write this as
\[\pfx(c, q) = \smashoperator{\sum_{(c',q') \in H \mid c \sshift c'}} \delta^{\mathcal{A}}_{\mathit{shift}}(u)(v) \pfx(c', q')\]
where $v$ is the \ac{pOPA} state of semi-configuration $c'$ (i.e., $c' = (v, \alpha')$).
Hence, by using the fact that $ \delta_{M_\mathcal{A}}(c)(c') = \delta^{\mathcal{A}}_{\mathit{shift}}(u)(v) \nex{c'} / \nex{c}$ (Definition~\ref{def:support-chain}), we obtain the constraint
\[
\text{(2b)} \quad
z_{c,q} = \smashoperator{\sum_{(c',q') \in H \mid c \sshift c'}} \delta_{M_\mathcal{A}}(c)(c') \, z_{c',q'}
\]
where $\delta_{M_\mathcal{A}}$ is the transition function of the support chain $M_\mathcal{A}$.

\item If $\Lambda(u) \lessdot \symb{\alpha}$,
a computation may leave $c$ with a \emph{push} or a \emph{support} edge.
Consider an edge $(c, q) \gedge (c', q')$.
If $c \spush c'$ can be a push edge in $M_\mathcal{A}$, we have $c' = (v, [\Lambda(u), u])$;
the probability that a run leaves $(c, q)$ through this edge and does not terminate
is simply $\delta_\mathit{push}(u)(v)$ multiplied by the probability $\pfx(c',q')$
that it does not terminate after continuing from $(c', q')$.

If $c \ssupp c'$ can be a support edge, then $c' = (v, \alpha)$,
and the probability that a run leaves $(c, q)$ through this edge and does not terminate is proportional to the product of the probability $\delta_{\mathit{push}}(u)(v')$ (for all $v' \in Q_\mathcal{A}$) that it performs the push move starting the support,
the probability $\pvar{v'}{[\Lambda(u), u]}{v}$ that the support ends in $v$,
and again $\pfx(c',q')$.
A support edge in $M_\mathcal{A}$ may, however, correspond to multiple edges in $H$,
and its probability $P_\mathit{supp}(c,c')$ is distributed among multiple edges.
We must thus add a term $\mu(c,q,c',q')$ that represents the fraction of $P_\mathit{supp}(c,c')$
associated with the edge $(c,q)$ to $(c',q')$ in $H$.

$\pfx(c, q)$ is the sum of the contributions of all such edges:
\[
\pfx(c,q) = \smashoperator{\sum_{\substack{(c',q') \in H \, \mid \, c \spush c' \\ (q, \, \Lambda(u), \, q') \in \delta^\mathcal{B}_\mathit{push}}}} P_\mathit{push}(c,c') \pfx(c',q')
+ \smashoperator{\sum_{(c', q') \in H \mid c \ssupp c'}} P_\mathit{supp}(c,c') \mu(c,q,c',q') \pfx(c',q')
\]
where $P_\mathit{push}(c,c')$ and $P_\mathit{supp}(c,c')$ are from Definition~\ref{def:support-chain}.

We can thus add the following constraints to the system:
\[
\text{(2c)} \quad
z_{c,q} = \smashoperator{\sum_{\substack{(c',q') \in H \, \mid \, c \spush c' \\ (q, \, \Lambda(u), \, q') \in \delta^\mathcal{B}_\mathit{push}}}} P_\mathit{push}(c,c') \frac{\nex{c'}}{\nex{c}} z_{c',q'}
+ \smashoperator{\sum_{(c', q') \in H \mid c \ssupp c'}} P_\mathit{supp}(c,c') \frac{\nex{c'}}{\nex{c}} \mu(c,q,c',q') z_{c',q'}
\]
\end{itemize}
Note that the equations identifying variables $z_{c,q}$ are linear in these variables. 
Now we need to add more (polynomial) constraints to identify variables $\mu$.
First note that, given a support edge $c \ssupp c'$ in $M_\mathcal{A}$ and $q' \in Q_\mathcal{B}$ such that $(c', q') \in H$, we have
\begin{equation}
\smashoperator{\sum_{q \mid (c,q) \in H}} \mu(c,q,c',q') = 1.
\label{eq:weight-sum}
\end{equation}

Then, we need to identify the values of $\mu(c,q,c',q')$ for all $c,c' \in \mathcal{C}$ and $q,q' \in Q_\mathcal{B}$.
To do so, we take the synchronized product $\hat{\mathcal{A}}$ between $\mathcal{A}$ and $\mathcal{B}$
we introduced in Section~\ref{sec:qualitative-mc} to build graph $G$,
and we add to all of its edges the probabilities of the corresponding transitions of $\mathcal{A}$.
We obtain a weighted \ac{OPBA}, whose weights on edges do not necessarily sum up to 1.
We then define a vector $\mathbf{h}$ of variables $h(\hat{c}, 
\hat{v})$ for each semi-configuration $\hat{c} = \bigl((u,q), [a, (w, s)] \bigr)$ 
and state $\hat{v} = (v,p)$ of $\hat{\mathcal{A}}$.
Formally, $h(\hat{c}, \hat{v})$ is the weight of those paths in $\hat{\mathcal{A}}$
such that (i) $\mathcal{A}$ has a sequence of transitions $\rho = (u, [a, w] A)\dots(v, A)$ 
in which no symbol in $A$ is ever popped,
and (ii) $\mathcal{B}$ has a sequence of transitions from configuration $\config{\Lambda_\varepsilon(\rho)}{q}{[a,s]B}$
to configuration $\config{\varepsilon}{p}{B}$ for some $B$, in which again no symbol of $B$ is ever popped.
These weights are the solutions of a system (3a) $\mathbf{h} = g(\mathbf{h})$
where $g$ is the equation system of $\hat{\mathcal{A}}$'s termination probabilities (cf. system $f$, Section~\ref{sec:popa}),
together with
(3b) $\mathbf{h} \geq 0$ and
\[
\text{(3c)} \sum_{q \in Q_{\mathcal{B}}} h(\hat{c}, \hat{v}) = \pvar{u}{[a,w]}{v}
\]
for all $\hat{c} = \bigl((u,q), [a, (w, s)] \bigr)$, and $\hat{v} = (v,p)$,  
such that $c = (u,[a,w]) \in \mathcal{C}$, $v \in Q_\mathcal{A}$, $p,s \in Q_\mathcal{B}$,
and $(p, \bot)$ is reachable from $(s, \bot)$ in $\mathcal{B}$.
Constraints (3c) make sure that the weights associated with all supports in $\hat{\mathcal{A}}$ that correspond to a support in $\mathcal{A}$ sum up to the probability of that support reaching its pop transition.

Finally, we identify $\mu(c,q,c',q')$, where $c = (u,\alpha)$, $c' = (u',\alpha)$ and $a = \Lambda(u)$, with
\[
\text{(3d)} \ 
P_\mathit{supp}(c,c') \mu(c,q,c',q') = \smashoperator{\sum_{\hat{u} = \hat{\delta}_{push}\bigl((u,q), a\bigr)}} h\bigl((\hat{u}, [a, (u,q)]), (u', q')\bigr)
\]
(recall that semi-con\-fi\-gu\-ra\-tions linked by support edges have the same stack symbol).

\begin{lemma}
\label{lemma:quantitative-unique}
The system of constraints (1a)--(1d), (2a)--(2c) and (3a)--(3d) has a unique solution.
\end{lemma}
\begin{proof}
First, consider equations (1a)--(1d).
The vector of termination probabilities of $\mathcal{A}$ is the least fixed-point solution of (1a) $\mathbf{v} = f(\mathbf{v})$~\cite{BrazdilEKK13}.
Thanks to a result called the \emph{Unique Fixed Point Theorem}~\cite[Theorem 18]{EtessamiY12},
this fixed point can be uniquely identified by adding constraint (1b) together with
$\sum_{v \in Q_\mathcal{A}} \pvar{s}{\alpha}{v} \leq 1$ for each $c = (s,\alpha) \in \mathcal{C}$,
and $\sum_{v \in Q_\mathcal{A}} \pvar{s'}{\alpha'}{v} < 1$ for each $c' = (s',\alpha') \in \mathcal{C}$
such that $\nex{c'} > 0$.
Note that (1c) and (1d) imply the latter two sets of constraints.

Once we fix the solutions to (1a)--(1d), the system (2a)--(2c) is linear in variables $\mathbf{z}$.
It is possible to prove that this system has a unique solution following the proof of \cite[Lemma 35]{EtessamiY12},
hence we only give a hint of the proof.
Recall that, to prove Lemma~\ref{lemma:SCC-ancestor}, we introduced the multigraph $M'_\mathcal{A}$,
in which push and support edges appear multiple times, one for each equivalence class of push moves or supports
represented by such edges that link backwards the same states of $\mathcal{B}$.
Similarly, we turn $H$ into a multi-graph $H'$, adding multiple edges according to the same equivalence relation employed for $M'_\mathcal{A}$.
We then observe that the coefficients of the $\mathbf{z}$ variables in (2b) and (2c) are the sums of the weights
associated to edges of $H'$.
Although $H'$ is not, strictly speaking, a Markov chain, the uniqueness of the solution of (2a)--(2c) can be proved
by applying to $H'$ a technique developed for Markov chains in the proof of \cite[Proposition 5.11]{CourcoubetisY95}.

Equations (3a)--(3c) are polynomial and monotonic, and we can prove that they have a unique solution that is the least fixed point of (3a).
Although the Unique Fixed Point Theorem~\cite[Theorem 18]{EtessamiY12} does not apply directly,
its proof works for this system with little modifications.
In particular, the constraints (3c) play the same role as (1c) and (1d) in ensuring the uniqueness of the solutions of (1b):
the termination probabilities $\pvar{u}{[a,w]}{v}$ in (3c) serve as a bound for sums of variables $\mathbf{h}$
representing the termination probabilities of supports in $\hat{\mathcal{A}}$, instead of 1.

Finally, there is exactly an equation (3d) for each variable $\mu$ and,
once we fix the solutions to all previous equation systems, they are linear equations:
their solution is hence also unique.
\end{proof}

By construction of the equation system and Lemma~\ref{lemma:quantitative-unique} we can state
\begin{theorem}
\label{thm:correctness-quantitative}
The system of constraints (1a)--(1d), (2a)--(2c) and (3a)--(3d) has a unique solution,
in which each variable $z_{c,q}$, for all nodes $(c, q) \in H$, takes the value $\frac{\pfx(c, q)}{\nex{c}}$.
\end{theorem}

The equation system is polynomial, so quantitative model checking queries can be solved as sentences in \ac{ETR}.
In particular, if we want to know whether
$P(\Lambda_\varepsilon(\mathit{Runs}(\mathcal{A})) \cap L_\varphi) \geq \varrho$,
we add constraint $\sum_{q \in I_\mathcal{B}} z_{(u_0,\bot), q} \geq \varrho$
and check for satisfiability.

The system in equation (3a) is of size polynomial in $\mathcal{A}$ and $\mathcal{B}$,
so such queries can be decided in space polynomial in them.
If $\mathcal{B}$ encodes a \acs{POTLF} formula, then its size is exponential in formula length.
Thus, quantitative model checking can be done in space exponential in formula length.

\begin{example}[Running example, cont. \ref{exa:mc}]
\begin{figure}
\centering
\begin{tikzpicture}
  [state/.style={draw, ellipse, inner sep=1pt}, >=latex, font=\scriptsize,
   final/.style={}
  ]
\node[state, initial by arrow, initial text=] (n0) {$c_0,q_0$};
\node[state] (n6) [right=15pt of n0, final] {$c_1,q_1$};
\node[state] (n7) [below of=n6, final] {$c_2,q_1$};
\node[state] (n1) [right=15pt of n6] {$c_1,q_0$};
\node[state] (n2) [below of=n1] {$c_2,q_0$};
\node[state] (n4) [below of=n0, final] {$c_3,q_3$};
\node[state] (n5) [below of=n4, final] {$c_4,q_3$};
\path[->] (n0) edge[out=20, in=160, above] node{$\nicefrac{1}{2}$} (n1)
          (n0) edge[below] node[yshift=2pt]{$\nicefrac{1}{2}$} (n6)
          (n6) edge[out=20, in=165, below, pos=.4] node[yshift=3pt]{$\nicefrac{1}{6}$} (n1)
          (n1) edge[final, out=195, in=340, below, pos=.9] node[yshift=1pt]{$\nicefrac{1}{6}$} (n6)
          (n6) edge[out=240, in=215, loop, left] node{$\nicefrac{1}{6}$} (n6)
          (n1) edge[out=10, in=350, loop, final, above] node{$\nicefrac{1}{6}$} (n1)
          (n1) edge[right] node{$\nicefrac{2}{3}$} (n2)
          (n1) edge[right] node{$\nicefrac{2}{3}$} (n7)
          (n2) edge[out=10, in=350, loop, final, above] node{$\nicefrac{5}{6}$} (n2)
          (n2) edge[final, out=195, in=340, below, pos=.1] node[yshift=1pt]{$\nicefrac{5}{6}$} (n7)
          (n7) edge[out=20, in=165, below, pos=.4] node[yshift=3pt]{$\nicefrac{1}{6}$} (n2)
          (n7) edge[out=240, in=215, loop, left] node{$\nicefrac{1}{6}$} (n7)
          (n0) edge[final, left] node{$\nicefrac{1}{2}$} (n4)
          (n4) edge[left] node{$1$} (n5)
          (n5) edge[out=170, in=190, loop, left] node{$1$} (n5);
\end{tikzpicture}
\caption{Graph $H$ obtained after analyzing graph $G$ from Fig.~\ref{fig:graph-G}, with edges labeled with their weight in the equation system for quantitative model checking.}
\label{fig:graph-h}
\end{figure}

According to the qualitative model checking algorithm (cf.\ Example~\ref{exa:mc}),
$\mathcal{A}_\lcall$ satisfies the property encoded by $\mathcal{B}_\lcall$ almost surely.
However, we still perform quantitative model checking, and show that the solution to the equation system contains the value 1.

Fig.~\ref{fig:graph-h} shows graph H, with edges labeled with their weight in the equation system.
The sub-graphs made of nodes $\{(c_0, q_0), (c_1, q_1), (c_1, q_0), (c_2, q_1), (c_2, q_0)\}$
and $\{(c_0, q_0), (c_3, q_3), \allowbreak (c_4, q_3)\}$
are respectively responsible for the acceptance of languages $L_1$ and $L_2$ from Example~\ref{running-ex:opba}.

Note that $H$ is not a Markov chain: weights of edges outgoing from a node may not sum to 1.
However, equation \eqref{eq:weight-sum} holds.
For instance, consider the support self-loop on $c_1$ in the support chain (Fig.~\ref{fig:supp-chain}),
which has probability $\nicefrac{1}{3}$.
On $(c_1, q_1)$, its probability is split among the self-loop and the edge incoming from $(c_1, q_0)$,
both weighting $\nicefrac{1}{6}$.
Symmetrically, the weights of the self-loop and the edge incoming from $(c_1, q_1)$ sum to $\nicefrac{1}{3}$.

Based on the weights, we can compute the following probabilities associated with nodes:
\begin{align*}
z_{c_0, q_0} &= 1 & z_{c_3, q_3} &= 1 & z_{c_4, q_3} &= 1 & z_{c_1, q_1} &= \nicefrac{1}{6} \\
z_{c_1, q_0} &= \nicefrac{5}{6} & z_{c_2, q_1} &= \nicefrac{1}{6} & z_{c_2, q_0} &= \nicefrac{5}{6}
\end{align*}
According to constraints (2a), probabilities associated with nodes sharing the same \ac{pOPA} semi-configuration sum to 1.
E.g., $z_{c_2, q_0} + z_{c_2, q_1} = 1$.
Note that probabilities associated with nodes in \acp{BSCC} of $H$ are not necessarily $1$, as one might intuitively expect.
We have $z_{c_0, q_0} = 1$, confirming that $\mathcal{A}_\lcall$
satisfies the property encoded by $\mathcal{B}_\lcall$ almost surely.
\end{example}

\subsection{Complexity}
\label{ssec:hardness}
Qualitative \ac{LTL} model checking was proven \textsc{exptime}-hard by \citet{YannakakisE05},
which suggests a similar bound for \acs{POTLF}.
\textsc{exptime}-hardness for \acs{POTLF} qualitative model checking, however,
does not follow directly from this result,
because \ac{POTL} does not include \ac{LTL} operators directly,
and the only known translation of \ac{LTL} into \ac{POTL} \cite{ChiariMP21b}
employs past \ac{POTL} operators not included in \acs{POTLF}.

We thus prove \textsc{exptime}-hardness of \acs{POTLF} qualitative model checking directly
by a reduction from the acceptance problem for alternating linear-space-bounded Turing Machines.
This reduction has been employed by \citet{EtessamiY05,BrazdilEKK13,YannakakisE05},
but was originally introduced by~\citet{BouajjaniEM97} to obtain hardness results for (non-probabilistic) pushdown automata,
so we leave it to Appendix~\ref{sec:hardness-proof}.

\begin{lemma}
\label{lemma:exptime-hardness}
Qualitative model checking of \acs{POTLF} formulas against \acp{pOPA} is \textsc{exptime}-hard.
\end{lemma}

By joining the results of Theorems~\ref{thm:potlf-separated-opba} and
\ref{thm:ltl-separated-opba} with Theorems~\ref{thm:conditions-H} and \ref{thm:correctness-quantitative}
and Lemma~\ref{lemma:exptime-hardness} we state:
\begin{theorem}
Let $\mathcal{A}$ be a \ac{pOPA}, $\mathcal{B}$ a complete separated \ac{OPBA},
$\varphi$ a \acs{POTLF} formula, $\psi$ an \ac{LTL} formula, all on the same \acs{OP} alphabet,
and $\varrho$ a rational constant.
\begin{itemize}
\item Deciding $P(\Lambda_\varepsilon(\mathit{Runs}(\mathcal{A})) \cap L_\mathcal{B}) \geq \varrho$
is in \textsc{pspace}.
\item Deciding
$P(\Lambda_\varepsilon(\mathit{Runs}(\mathcal{A})) \cap L_\varphi) \geq \varrho$ or
$P(\Lambda_\varepsilon(\mathit{Runs}(\mathcal{A})) \cap L_\psi) \geq \varrho$
is \textsc{exptime}-complete if $\varrho = 1$ and in \textsc{expspace} if $\varrho < 1$.
\end{itemize}
\end{theorem}

\subsection{Analysis of the Schelling Coordination Game}
\label{sec:analysis-motivating-example}

We have implemented the qualitative and quantitative model checking algorithms in the prototype tool POPACheck~\cite{PontiggiaBC25}, as part of the open-source OPPAS~\cite{oppas} model checking suite.
The tool automatically models probabilistic programs as \acp{pOPA},
and supports specifications given as \ac{LTL} and \ac{POTLF} formulas, 
for which it automatically builds an equivalent \ac{OPBA}.

We ran the tool on the motivating example in Fig.~\ref{fig:coordination_game}
on a PC with a 2.1GHz 12th Gen Intel Core i7 CPU and 32~GB of RAM.
We report here the experimental results obtained with the latest version of our tool.
With respect to \cite{PontiggiaBC25}, the tool is significantly faster, generates less states,
and generates much less equations by detecting qualitatively all variables which have solution zero.
The resulting \ac{pOPA} from the program has 288 states, and the system that encodes its termination probabilities
(cf.\ Section~\ref{sec:popa}) consists of 786 nonlinear polynomial equations.
We tried to solve the system with two state-of-the-art \acs{SMT} solvers supporting \ac{ETR}
(or \textsc{qf\_nra}): Z3 4.12 \cite{MouraB08} and cvc5 1.1 \cite{BarbosaBBKLMMMN22},
but both timed out after 20 minutes.
This is in line with previous results obtained by \citet{WinklerK23a},
and confirms that equation systems encoding \ac{pPDA} termination probabilities
are generally hard to solve.
Instead, our tool solves such equation systems by approximating solutions through numerical methods
\cite{EtessamiY09,WinklerK23a}, which took less than 1 second on the program of Fig.~\ref{fig:coordination_game}.

We then checked the program against \ac{POTLF} formulas reported in Section~\ref{intro}.
Qualitative model checking ran in 1.37~s on formula~\eqref{eq:disagreement},
finding out that it does not hold with probability 1.
We then ran quantitative model checking, finding out in 5~s that
the formula holds with probability $\sim0.895$.
The tool also found out that formula \eqref{eq:aliceLoc1} from Section~\ref{intro}
holds with probability $\sim0.610$ (qualitative model checking took 0.34~s, quantitative took 1~s).



\section{Conclusion}
\label{sec:conclusions}
We introduced a novel optimal automata-based algorithm for model checking regular and context-free properties on \acp{pPDA}.
We check properties in \ac{LTL} and \acs{POTLF}, a fragment of the context-free temporal logic \acs{POTL}.
\acs{POTLF} expresses properties definable as \acp{OPL},
a language class strictly larger than \acp{VPL},
which thus allows for higher expressiveness than \ac{VPL}-based approaches~\cite{WinklerGK22}
(e.g., conditioning through \texttt{observe} statements).

The analysis of the motivating example shows that our proposed automata-theoretic framework can formalize various interesting program properties,
such as Hoare-style pre/post conditions, even in the expressive setting of nested probabilistic programs.
However, many practical obstacles are yet to tackle to achieve practical scalability and applicability of our approach.
In a separate publication~\cite{PontiggiaBC25}, we describe the tool POPACheck, present more benchmarks and applications, and contribute a numerical semi-algorithm for the construction of the support chain. 

Support for full \ac{POTL} also requires further work.
Promising approaches are \ac{OPBA} determinization through \emph{stair-parity}
acceptance conditions~\cite{LodingMS04,WinklerGK22},
which would also allow for model checking non-deterministic \ac{OPBA} specifications.
Determinization, however, would yield a non-optimal doubly exponential \ac{POTL} model checking algorithm.
Nonetheless, there exist efficient \emph{direct} constructions translating LTL formulae into \acp{DRA}~\cite{EsparzaKS20}, with which the doubly exponential blowup rarely occurs in practice~\cite{KretinskyMSZ18}. 
Future work includes investigating whether these techniques extend to context-free logics such as \ac{POTL}.
Another promising approach is encoding specifications as \emph{unambiguous automata},
which have been recently exploited to define an optimal (singly exponential)
algorithm for \acs{LTL} model checking on Markov chains~\cite{BaierK00023}.

\begin{acks}
This work was partially funded by the Vienna Science and Technology Fund (WWTF)
grant [10.47379/ICT19018] (ProbInG),
and WWTF project ICT22-023 (TAIGER),
and by the EU Commission in the Horizon Europe research and innovation programme
\includegraphics[width=1em]{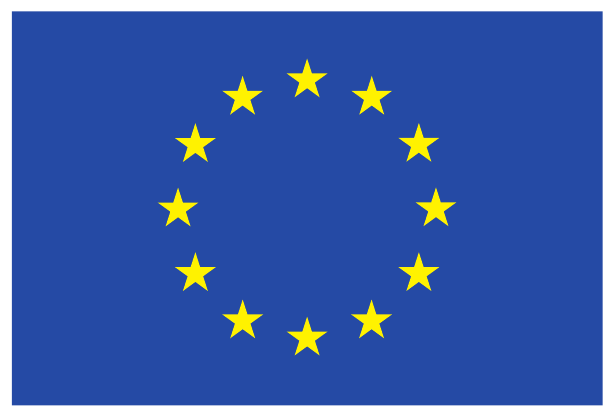}
under grant agreements No.\ 101034440 (Marie Sk\l{}odowska-Curie Doctoral Network LogiCS@TU Wien)
and No.\ 101107303 (MSCA Postdoctoral Fellowship CORPORA).
\end{acks}

\bibliographystyle{ACM-Reference-Format}
\bibliography{biblio,prob}

\clearpage
\appendix
\section{Appendix: Omitted Proofs}

\subsection{Construction of separated \texorpdfstring{\ac{OPBA}}{omegaOPBA} from \texorpdfstring{\acs{POTLF}}{POTLfX} formulas}
\label{sec:opba-construction-proof}

In this section, we report the construction of \ac{OPBA}
that accept the language of a \acs{POTL} formulas originally given in \cite{ChiariMPP23},
and show that, when restricted to \acs{POTLF} operators,
it yields separated \ac{OPBA}.
Theorem~\ref{thm:potlf-separated-opba} follows from
Theorem~\ref{thm:potlf-separated-appendix} and \cite[Theorem 5.6]{ChiariMPP23}.
Finally, in Section~\ref{sec:potl-construction-not-sep} we show that the construction of \cite{ChiariMPP23}
yields \acp{OPBA} that are not separated when applied to \acs{POTL} operators not included in \acs{POTLF}.

The construction yields a generalized \ac{OPBA},
which differs from normal \ac{OPBA} in that it has a set of sets of final states
$\mathbf{F}$ instead of just one.
Runs are final if they visit at least one state from each final set
$F \in \mathbf{F}$ infinitely often.
Standard methods based on counters
convert B\"uchi automata
with a generalized acceptance condition into
ordinary B\"uchi automata
with a polynomial size increase~\cite{ChiariMPP23}.
Unfortunately, these methods do not preserve 
the properties of completeness and separation.
In Section~\ref{sec-potl-generalized-to-ordinary-Buechi-cond},
we present another transformation from~\cite{CartonM03} which preserves the required properties.
It generates automata with exponential size in the number of acceptance sets: 
since the acceptance sets are at most $|\varphi|$,
the alternative transformation does not undermine
the optimality of the \ac{OPBA} construction.
Moreover, if separation and completeness are to be preserved,
the exponential blowup in the number of acceptance sets cannot be avoided~\cite[Proposition 18]{CartonM03}.
Alternatively, algorithms presented in the rest of the paper can be easily adapted
to the generalized acceptance condition by checking for the presence of a state from each set
$F \in \mathbf{F}$ whenever they need to check for the presence of a final state.

Given a \acs{POTL} formula $\varphi$, we build a generalized \ac{OPBA}
\(
  \mathcal{B} =
  ( \powset{AP}, \allowbreak
  M_{AP}, \allowbreak
  Q, \allowbreak
  I, \allowbreak
  \mathbf{F}, \allowbreak
  \delta )
\)
that accepts all and only models of $\varphi$ by following the construction given in \cite{ChiariMPP23}.
We report the construction here and show that, limited to \acs{POTLF} operators, 
it yields separated \ac{OPBA}.

We first introduce a few auxiliary operators:
\begin{itemize}
\item $\zeta_L$ forces the position where it holds to be the first one of a chain body;
\item $\lcnext{\prf}$ for each $\prf \in \{\lessdot, \doteq, \gtrdot\}$,
  is such that $(w,i) \models \lcdnext \psi$ iff
  $\exists j > i : \chain(i,j) \land i \pr j \land (w,j) \models \psi$.
\end{itemize}
The \emph{closure} of $\varphi$, $\clos{\varphi}$, is the smallest set satisfying the following constraints:
\begin{itemize}
\item
  $\varphi \in \clos{\varphi}$,
\item
  $AP \subseteq \clos{\varphi}$,
\item
  if $\psi \in \clos{\varphi}$ and $\psi \neq \neg \theta$, then $\neg \psi \in \clos{\varphi}$
  (we identify $\neg \neg \psi$ with $\psi$);
\item
  if $\neg \psi \in \clos{\varphi}$, then $\psi \in \clos{\varphi}$;
\item
  if any of the unary temporal operators (e.g., $\ldnext$, $\lcdnext$, \dots)
  is in $\clos{\varphi}$, its operand is in $\clos{\varphi}$;
\item
  if any of the binary operators (e.g., $\lor$, $\lcduntil{}{}$, \dots)
  is in $\clos{\varphi}$, and $\psi$ and $\theta$ are its operands,
  then $\psi, \theta \in \clos{\varphi}$;
\item
  if $\lcdnext \psi \in \clos{\varphi}$, then
  $\zeta_L, \lcnext{\lessdot} \psi, \lcnext{\doteq} \psi \in \clos{\varphi}$;
\item
  if $\lcunext \psi \in \clos{\varphi}$, then
  $\zeta_L, \lcnext{\gtrdot} \psi, \lcnext{\doteq} \psi \in \clos{\varphi}$;
\item
  if $\luntil{t}{\psi}{\theta} \in \clos{\varphi}$ for $t \in \{d, u\}$, then
  $\lnextsup{t}(\luntil{t}{\psi}{\theta}), \lcnext{t}(\luntil{t}{\psi}{\theta}) \in \clos{\varphi}$.
\end{itemize}

The set of atoms $\atoms{\varphi}$ contains all \emph{consistent} subsets of $\clos{\varphi}$,
i.e., all $\Phi \subseteq \clos{\varphi}$ that satisfy a set of
\emph{Atomic consistency Constraints} $\aconstr$.
$\aconstr$ contains the following constraints: for any $\Phi \in \atoms{\varphi}$,
\begin{enumerate}[label=(\alph*), series=aconstr]
\item \label{aconstr:mc-neg}
  $\psi \in \Phi$ iff $\neg \psi \notin \Phi$ for every $\psi \in \clos{\varphi}$;
\item \label{aconstr:mc-lor}
  $\psi \lor \theta \in \Phi$, iff $\psi \in \Phi$ or $\theta \in \Phi$, or both;
\item \label{aconstr:lcdnext}
  $\lcdnext \psi \in \Phi$
  iff ($\lcnext{\lessdot} \psi \in \Phi$ or $\lcnext{\doteq} \psi \in \Phi$);
\item \label{aconstr:lcunext}
  $\lcunext \psi \in \Phi$
  iff ($\lcnext{\doteq} \psi \in \Phi$ or $\lcnext{\gtrdot} \psi \in \Phi$);
\item \label{aconstr:until}
  $\luntil{t}{\psi}{\theta} \in \Phi_c$, with $t \in \{d, u\}$, iff either:
  \begin{itemize}
  \item $\theta \in \Phi_c$, or
  \item $\lnextsup{t}(\luntil{t}{\psi}{\theta}), \psi \in \Phi_c$, or
  \item $\lcnext{t}(\luntil{t}{\psi}{\theta}), \psi \in \Phi_c$.
  \end{itemize}
\end{enumerate}

Moreover, we define the set of \emph{pending formulas} as
\[
\clospn{\varphi} = \big\{ \theta \in \clos{\varphi} \mid \theta \in \{
\zeta_L,
\lcnext{\prf} \psi 
\}
\, \text{for $\prf \in \{\lessdot, \doteq, \gtrdot\}$
and $\psi \in \clos{\varphi}$}\big\}.
\]
and the set of \emph{in-stack formulas} as
\(
\closst{\varphi} =
\{ \lcnext{\prf} \psi \in \clos{\varphi} \mid \prf \in \{\lessdot, \doteq, \gtrdot\} \}.
\)

The states of $\mathcal{B}$ are the set
$Q = \atoms{\varphi} \times \powset{\clospn{\varphi}} \times \powset{\closst{\varphi}}$,
and its elements, which we denote with Greek capital letters,
are of the form $\Phi = (\Phi_c, \Phi_p, \Phi_s)$,
where $\Phi_c$, called the \emph{current} part of $\Phi$,
is the set of formulas that hold in the next position that $\mathcal{B}$ is going to read;
$\Phi_p$, or the \emph{pending} part of $\Phi$, is a set of \emph{temporal obligations};
and $\Phi_s$ is the \emph{in-stack} part, keeping track of when certain formulas are in some stack symbol.

The initial set $I$ contains all states of the form $(\Phi_c, \Phi_p, \emptyset) \in Q$
such that $\zeta_L \in \Phi_p$ iff $\# \not\in \Phi_c$,
and such that $\varphi \in \Phi_c$, i.e., the input formula $\varphi$ holds at initial states.
The sets of accepting states in $\mathcal{F}$ will be introduced later.

Temporal obligations are enforced by the transition relation $\delta$,
defined as the set of all transitions that satisfy a set of \emph{$\delta$-rules}, $\dconstr$.
We introduce $\dconstr$ gradually for each operator:
$\delta_\mathit{push}$ and $\delta_\mathit{shift}$ are the largest subsets of
$Q \times \powset{AP} \times Q$ satisfying all rules in $\dconstr$,
and $\delta_\mathit{pop}$ is the largest subset of $Q \times Q \times Q$ satisfying all rules in $\dconstr$.
First, we introduce two $\dconstr$ rules that are always present and are not bound to a particular operator.

Each state of $\mathcal{B}$ guesses the \acp{AP} that will be read next.
So, $\dconstr$ always contains the rule that
\begin{enumerate}[series=dconstr]
\item \label{rule:delta-push-shift}
  for any $(\Phi, a, \Psi) \in \delta_\mathit{push/shift}$,
  with $\Phi, \Psi \in Q$ and $a \in \powset{AP}$, we have $\Phi_c \cap AP = a$
\end{enumerate}
(by $\delta_\mathit{push/shift}$ we mean $\delta_\mathit{push} \cup \delta_\mathit{shift}$).
\emph{Pop} moves, on the other hand, do not read input symbols,
and $\mathcal{B}$ remains at the same position when performing them:
$\dconstr$ contains the rule
\begin{enumerate}[resume*=dconstr]
\item \label{rule:delta-pop}
  for any $(\Phi, \Theta, \Psi) \in \delta_\mathit{pop}$ it must be $\Phi_c = \Psi_c$.
\end{enumerate}
The auxiliary operator $\zeta_L$ is governed by the following rules in $\dconstr$:
\begin{enumerate}[resume*=dconstr]
\item \label{rule:zeta-l-pop-shift}
  if $(\Phi, a, \Psi) \in \delta_\mathit{shift}$
  or $(\Phi, \Theta, \Psi) \in \delta_\mathit{pop}$, for any $\Phi, \Theta, \Psi$ and $a$,
  then $\zeta_L \not\in \Phi_p$;
\item \label{rule:zeta-l-push}
  if $(\Phi, a, \Psi) \in \delta_\mathit{push}$, then $\zeta_L \in \Phi_p$.
\end{enumerate}
If $\ldnext \psi \in \clos{\varphi}$ for some $\psi$, $\dconstr$ contains this rule:
\begin{enumerate}[resume*=dconstr]
\item \label{rule:ltnext}
  for all $(\Phi, a, \Psi) \in \delta_\mathit{push/shift}$,
  it must be that $\ldnext \psi \in \Phi_c$ iff ($\psi \in \Psi_c$
  and either $a \lessdot b$ or $a \doteq b$, where $b = \Psi_c \cap AP$).
\end{enumerate}
Replace $\lessdot$ with $\gtrdot$ for the upward counterpart.

If $\lcnext{\doteq} \psi \in \clos{\varphi}$,
its satisfaction is ensured by the following rules in $\dconstr$:
\begin{enumerate}[resume*=dconstr]
\item \label{rule:lcnext-doteq-start}
  Let $(\Phi, a, \Psi)\in \delta_\mathit{push/shift}$:
  then $\lcnext{\doteq} \psi \in \Phi_c$ iff $\lcnext{\doteq} \psi, \zeta_L \in \Psi_p$;
\item \label{rule:lcnext-doteq-pop}
  let $(\Phi, \Theta, \Psi) \in \delta_\mathit{pop}$:
  then $\lcnext{\doteq} \psi \not\in \Phi_p$
  and ($\lcnext{\doteq} \psi \in \Theta_p$ iff $\lcnext{\doteq} \psi \in \Psi_p$);
\item \label{rule:lcnext-doteq-shift}
  let $(\Phi, a, \Psi) \in \delta_\mathit{shift}$:
  then $\lcnext{\doteq} \psi \in \Phi_p$ iff $\psi \in \Phi_c$.
\end{enumerate}
\noindent If $\lcnext{\lessdot} \psi \in \clos{\varphi}$,
then $\dconstr$ contains the following rules:
\begin{enumerate}[resume*=dconstr]
\item \label{rule:lcnext-lessdot-push-shift}
  Let $(\Phi, a, \Psi) \in \delta_\mathit{push/shift}$:
  then $\lcnext{\lessdot} \psi \in \Phi_c$ iff $\lcnext{\lessdot} \psi, \zeta_L \in \Psi_p$;
\item \label{rule:lcnext-lessdot-pop}
  let $(\Phi, \Theta, \Psi) \in \delta_\mathit{pop}$:
  then $\lcnext{\lessdot} \psi \in \Theta_p$ iff ($\zeta_L \in \Psi_p$ and (either
  \begin{enumerate*}
  \item
    $\lcnext{\lessdot} \psi \in \Psi_p$ or
  \item
    $\psi \in \Phi_c$));
  \end{enumerate*}
\item \label{rule:lcnext-lessdot-shift}
  let $(\Phi, a, \Psi) \in \delta_\mathit{shift}$: then $\lcnext{\lessdot} \psi \not\in \Phi_p$.%
\footnote{This rule is not present in \cite{ChiariMPP23},
but it is required for the support graph to be backward deterministic.
The rule does not interfere with the proofs of Lemmas~A.1 and 5.4 from \cite{ChiariMPP23},
so all correctness claims remain valid.}
\end{enumerate}
\noindent If $\lcnext{\gtrdot} \psi \in \clos{\varphi}$, in $\dconstr$ we have:
\begin{enumerate}[resume*=dconstr]
\item \label{rule:lcnext-gtrdot-start}
  Let $(\Phi, a, \Psi) \in \delta_\mathit{push/shift}$:
  then $\lcnext{\gtrdot} \psi \in \Phi_c$ iff $\lcnext{\gtrdot} \psi, \zeta_L \in \Psi_p$;
\item \label{rule:lcnext-gtrdot-pop}
  let $(\Phi, \Theta, \Psi) \in \delta_{pop}$:
  ($\lcnext{\gtrdot} \psi \in \Theta_p$ iff $\lcnext{\gtrdot} \psi \in \Psi_p$)
  and ($\lcnext{\gtrdot} \psi \in \Phi_p$ iff $\psi \in  \Phi_c$);
\item \label{rule:lcnext-gtrdot-shift}
  let $(\Phi, a, \Psi) \in \delta_\mathit{shift}$: then $\lcnext{\gtrdot} \psi \not\in \Phi_p$.
\end{enumerate}

Moreover, if $\psi = \lcnext{\prf} \theta \in \clos{\varphi}$, $\dconstr$ also contains the following rules:
\begin{enumerate}[resume*=dconstr]
\item \label{rule:mc-stack-push}
  for any $(\Phi, a, \Theta) \in \delta_\mathit{push}$,
  ($\psi \in \Phi_p$ or $\psi \in \Phi_s$) iff $\psi \in \Theta_s$;
\item \label{rule:mc-stack-shift}
  for any $(\Phi, a, \Theta) \in \delta_\mathit{shift}$,
  $\psi \in \Phi_s$ iff $\psi \in \Theta_s$;
\item \label{rule:mc-stack-pop}
  for any $(\Phi, \Theta, \Psi) \in \delta_\mathit{pop}$,
  ($\psi \in \Phi_s$ and $\psi \in \Theta_s$) iff $\psi \in \Psi_s$.
\end{enumerate}

Thus, we can define the acceptance sets
$\bar{F}_{\lcnext{\prf} \psi} = \{\Phi \in Q_\omega \mid \lcnext{\prf} \psi \not\in \Phi_p \cup \Phi_s\}$,
for $\prf \in \{\doteq, \gtrdot\}$ and
$\bar{F}_{\lcnext{\lessdot} \psi} = \{\Phi \in Q_\omega \mid \lcnext{\lessdot} \psi \not\in \Phi_s \land (\lcnext{\lessdot} \psi \not\in \Phi_p \lor \psi \in \Phi_c)\}$
that are in $\mathcal{F}$.
The acceptance set for $\lcduntil{\psi}{\theta}$ is the following:
\[
\bar{F}_{\lcduntil{\psi}{\theta}} =
\bar{F}_{\lcnext{\doteq} (\lcduntil{\psi}{\theta})}
\cap \bar{F}_{\lcnext{\lessdot} (\lcduntil{\psi}{\theta})}
\cap \{\Phi \in Q_\omega \mid \lcduntil{\psi}{\theta} \not\in \Phi_c \lor \theta \in \Phi_c\}.
\]
The one for $\lcuuntil{\psi}{\theta}$ is obtained
by substituting $\gtrdot$ for $\lessdot$ and $u$ for $d$.

The following holds:
\begin{lemma}[{\cite[Lemma 5.3]{ChiariMPP23}}]
\label{lemma:mc-omega-stack}
For any $\omega$-word $w = \# x y$ on $(\powset{AP}, M_{AP})$,
let $\config{y}{\Phi}{\gamma}$ be $\mathcal{B}$'s configuration after reading $x$.

If $\psi = \lcnext{\prf} \theta$ for $\prf \in \{\lessdot, \doteq, \gtrdot\}$,
then there exists a stack symbol $[a, \Theta] \in \gamma$ such that $\psi \in \Theta_p$
iff $\psi \in \Phi_s$.
\end{lemma}
And from it follows:
\begin{corollary}
\label{cor:no-in-stack}
If a semi-configuration $(\Phi, b, \ell)$ in the support graph of $\mathcal{B}$ is such that $\Phi_s \neq \emptyset$,
then no final \ac{SCC} is reachable from it.
\end{corollary}
\begin{proof}
The final sets $\bar{F}_{\lcnext{\prf} \psi}$ for all $\prf \in \{\lessdot, \doteq, \gtrdot\}$
are such that states with non-empty in-stack part are not final.
Moreover, the support graph does not contain pop edges.
Thus, each semi-configuration represents a configuration of $\mathcal{B}$
in which the stack contents will never be popped.
Thus, by Lemma~\ref{lemma:mc-omega-stack} all semi-configurations reachable
from a state with a non-empty in-stack part also have a non-empty in-stack part.
\end{proof}
Corollary~\ref{cor:no-in-stack} allows us to state that, after removing from the support graph
all nodes that cannot reach a final \ac{SCC}, no nodes containing a state with non-empty in-stack part remain.
A run of $\mathcal{B}$ can only contain such states within closed supports,
which appear as support edges in the support graph.

We define the set of sub-formulas $\subf(\psi)$ of a formula $\psi$
as the smallest set such that:
\begin{itemize}
\item $\psi \in \subf(\psi)$;
\item if any of the unary operators (e.g., $\neg$, $\ldnext$, $\lcnext{\lessdot}$, \dots)
  is in $\subf(\psi)$, and $\psi$ is its operand, then $\psi \in \subf(\psi)$;
\item if any of the binary operators (e.g., $\land$, $\lor$, $\lcduntil{}{}$, \dots)
  is in $\subf(\psi)$, and $\psi_1$ and $\psi_2$ are its operands,
  then $\psi_1, \psi_2 \in \subf(\psi)$.
\end{itemize}
The set of \emph{strict} sub-formulas of $\varphi$ is
$\ssubf(\varphi) = \subf(\varphi) \setminus \{\varphi\}$.

From the proofs of Theorems 4.4 and 5.5 in \cite{ChiariMPP23} follows
\begin{theorem}
\label{thm:mc-correctness}
For each $\theta \in \subf(\varphi)$,
accepting computations of $\mathcal{B}$ are such that for each position $i$ in the input word $w$ we have
$(w,i) \models \theta$ iff $\theta \in \Phi_c(i)$,
where $\Phi_c(i)$ is the state reached by $\mathcal{B}$ right before reading position $i$ of $w$.
\end{theorem}

We now prove a general corollary of Theorem~\ref{thm:mc-correctness}:
\begin{corollary}
\label{cor:current-determined}
Let $w$ be an $\omega$-word such that
$w \in L_\mathcal{B}(\Phi^1, \beta_1 \bot)$ and
$w \in L_\mathcal{B}(\Phi^2, \beta_2 \bot)$,
with $\symb{\beta_1} = \symb{\beta_2} = b$.
Then we have $\Phi^1_c = \Phi^2_c$.
\end{corollary}
\begin{proof}
Let $\beta \neq \bot$.
By the construction of $\mathcal{B}$, there exist %
states $\Psi^1, \Psi^2 \in Q$
such that $bw \in L_\mathcal{B}(\Psi^1, \bot)$ and $bw \in L_\mathcal{B}(\Psi^2, \bot)$,
and the respective runs reach $\Phi^1$ and $\Phi^2$ after reading $b$.

Such runs are accepting, so due to Theorem~\ref{thm:mc-correctness},
for all $\theta \in \subf(\varphi)$ we have
$\theta \in \Phi^1_c$ and $\theta \in \Phi^2_c$ iff $(bw, 2) \models \theta$.
Since we only consider future formulas, $(bw, 2) \models \theta$ iff $(w, 1) \models \theta$,
hence $\theta \in \Phi_c$ iff $(xw, 1) \models \theta$.

If $\beta = \bot$, the proof is analogous:
we need not consider $bw$, but we proceed directly with $w$.
\end{proof}

Now, we prove the following:
\begin{theorem}
\label{thm:potlf-separated-appendix}
The \ac{OPBA} $\mathcal{B}$ built for a formula $\varphi$
as described in this section is separated.
\end{theorem}
\begin{proof}
We prove the following claim:
Let $w$ be an $\omega$-word such that
$w \in L_\mathcal{B}(\Phi^1, \beta_1 \bot)$ and
$w \in L_\mathcal{B}(\Phi^2, \beta_2 \bot)$,
with $\symb{\beta_1} = \symb{\beta_2} = b$.
Then we have $\Phi^1 = \Phi^2$.

In the following, let $a$ be the label of the first position of $w$.
We proved that $\Phi^1_c = \Phi^2_c$ in Corollary~\ref{cor:current-determined}.
Moreover, the claim implies that the support graph of $\mathcal{B}$
contains two nodes $(\Phi^1, b, a)$ and $(\Phi^2, b, a)$.
Since both runs starting from $\Phi^1$ and $\Phi^2$ are final,
by Corollary~\ref{cor:no-in-stack} we have $\Phi^1_s = \Phi^2_s = \emptyset$.

It remains to prove that $\Phi^1_p = \Phi^2_p$.
We analyze each operator in $\clospn{\varphi}$ separately.

Because of $\dconstr$ rules \ref{rule:zeta-l-pop-shift} and \ref{rule:zeta-l-push},
$\zeta_L \in \Phi^1_p$ (resp.\ $\zeta_L \in \Phi^2_p$)
iff $\symb{\beta_1} \lessdot a$ (resp.\ $\symb{\beta_2} \lessdot a$).
Since $\symb{\beta_1} = \symb{\beta_2}$, we have
$\zeta_L \in \Phi^1_p$ iff $\zeta_L \in \Phi^2_p$.

For the chain next operators, we need to consider the edges
$e_1 = (\Phi^1, b, a) \suppedge (\Psi^1, b', \ell)$ and
$e_2 = (\Phi^2, b, a) \suppedge (\Psi^2, b', \ell)$ generated in the support graph
by the two final runs on $w$ starting from $\Phi^1$ and $\Phi^2$.
Note that $e_1$ and $e_2$ are always of the same type (among push, shift, and support)
because they share $a$ and $b$.

\begin{itemize}
\item $[\lcnext{\doteq} \psi]$
  Recall $\lcnext{\doteq} \psi \not\in \Psi^1_s, \Psi^2_s$ by Corollary~\ref{cor:no-in-stack}.
  Thus, if $e_1, e_2$ are push edges, by rule \ref{rule:mc-stack-push}
  we have $\lcnext{\doteq} \psi \not\in \Phi^1_p, \Phi^2_p$.

  If $e_1, e_2$ are shift edges, by rule \ref{rule:lcnext-doteq-shift},
  $\lcnext{\doteq} \psi \in \Phi^1_p$ iff $\psi \in \Phi^1_c$ and
  $\lcnext{\doteq} \psi \in \Phi^2_p$ iff $\psi \in \Phi^2_c$,
  but we already proved that $\psi \in \Phi^1_c$ iff $\psi \in \Phi^2_c$,
  hence $\lcnext{\doteq} \psi \in \Phi^1_p$ iff $\lcnext{\doteq} \psi \in \Phi^2_p$.

  Let $e_1, e_2$ be support edges.
  The first push transition in the support pushes the stack symbol popped by the last pop move, which leads to $\Psi$.
  Then, by rule \ref{rule:lcnext-doteq-pop} we have
  $\lcnext{\doteq} \psi \in \Phi^1_p$ iff $\lcnext{\doteq} \psi \in \Psi^1_p$.
  Let $e'_1 = (\Psi^1, b', \ell) \suppedge (\Theta^1, b'', \ell'')$
  be the edge of the support graph following $e_1$ in the run on $w$.
  If $e'_1$ is a push or a shift edge, whether $\lcnext{\doteq} \psi \in \Psi^1_p$
  is determined as in our previous analysis of push and shift edges.
  If $e'_1$ is a support edge, then $\lcnext{\doteq} \psi \in \Psi^1_p$ iff
  $\lcnext{\doteq} \psi \in \Theta^1_p$ also by rule \ref{rule:lcnext-doteq-pop}.
  But note that if $\lcnext{\doteq} \psi \in \Theta^1_p$,
  then $\Theta^1_p \not\in \bar{F}_{\lcnext{\doteq} \psi}$.
  So, for the run starting with $\Phi^1$ to be final,
  at some point it must reach a shift edge where $\lcnext{\doteq} \psi$
  because $\psi$ holds (rule \ref{rule:lcnext-doteq-shift}).
  The same argument applies for whether $\lcnext{\doteq} \psi \in \Phi^2_p$,
  hence $\lcnext{\doteq} \psi \in \Phi^1_p$ iff $\lcnext{\doteq} \psi \in \Phi^2_p$.

\item $[\lcnext{\lessdot} \psi]$
  If $e_1$, $e_2$ are shift edges, then $\lcnext{\lessdot} \psi \not\in \Phi^1_p, \Phi^2_p$
  by rule \ref{rule:lcnext-lessdot-shift}.

  If $e_1$ is a push edge, $\lcnext{\lessdot} \psi \in \Phi^1_p$
  would imply $\lcnext{\lessdot} \psi \in \Psi^1_s$
  by rule \ref{rule:mc-stack-push}, which by Corollary~\ref{cor:no-in-stack} prevents $(\Psi, b', \ell)$,
  and thus $(\Phi, b, a)$ from reaching a final \ac{SCC}.
  The same argument applies to $e_2$, and thus
  $\lcnext{\lessdot} \psi \not\in \Phi^1_p, \Phi^2_p$.

  If $e_1$ is a support edge, by rule \ref{rule:lcnext-lessdot-pop} we have
  $\lcnext{\lessdot} \psi \in \Phi^1_p$ iff ($\zeta_L \in \Psi^1_p$ and (either
  $\lcnext{\lessdot} \psi \in \Psi^1_p$ or
  $\psi \in \Psi^1_c$)).
  In fact, let $\Theta^1$ be the state preceding $\Psi^1$ in the support represented by $e_1$:
  by rule \ref{rule:delta-pop} we have $\psi \in \Theta^1_c$ iff $\psi \in \Psi^1_c$.
  The same argument holds for $e_2$.
  Thus, if $b' \doteq \ell$, by rule \ref{rule:lcnext-lessdot-shift}
  $\zeta_L \not\in \Phi^1_p, \Phi^2_p$ and thus
  $\lcnext{\lessdot} \psi \not\in \Phi^1_p, \Phi^2_p$.
  If $b' \lessdot \ell$ and $\psi \in \Psi^1_c$, we know that
  $\psi \in \Psi^1_c$ iff $\psi \in \Psi^2_c$,
  and hence $\lcnext{\lessdot} \psi \in \Phi^1_p$ iff $\lcnext{\lessdot} \psi \in \Phi^2_p$.
  Otherwise, we can argue that $\lcnext{\lessdot} \psi \in \Psi^1_p$
  iff $\lcnext{\lessdot} \psi \in \Psi^2_p$ in a way similar as we did for
  $\lcnext{\doteq} \psi$ (i.e., we show that $\lcnext{\lessdot} \psi$ must disappear
  from the pending part of the states of a final run).

\item $[\lcnext{\gtrdot} \psi]$
  Suppose by contradiction that $\lcnext{\gtrdot} \psi \in \Phi^1_p$.
  By rule \ref{rule:lcnext-gtrdot-shift}, $e_1$ cannot be a shift edge.
  If $e_1$ is a push edge, then by rule \ref{rule:mc-stack-push} we have $\lcnext{\gtrdot} \psi \in \Phi^1_s$,
  which by Corollary~\ref{cor:no-in-stack} prevents $(\Psi^1, b', \ell)$,
  and thus $(\Phi^1, b, a)$ from reaching a final \ac{SCC}.

  If $e$ is a support edge, then by rule \ref{rule:lcnext-gtrdot-pop} we have
  $\lcnext{\gtrdot} \psi \in \Psi^1_p$,
  so $\Psi^1 \not\in \bar{F}_{\lcnext{\gtrdot} \psi}$ and, by repeating these same arguments,
  the same can be said for all its successors, whether they are linked by push or support edges.

  Hence, no node that can reach a final \ac{SCC} has $\lcnext{\gtrdot} \psi$
  in its pending part, and we can state $\lcnext{\gtrdot} \psi \not\in \Phi^1, \Phi^2$.
\end{itemize}
We can conclude that $\Phi^1_p = \Phi^2_p$.
\end{proof}

\subsubsection{Transformation from generalized to ordinary \acp{OPBA}}
\label{sec-potl-generalized-to-ordinary-Buechi-cond}
\begin{figure}[tb]
\centering
\begin{tikzpicture}
  [node distance=90pt, every state/.style={minimum size=0pt, inner sep=2pt}, >=latex, font=\scriptsize]
\node[state, initial by arrow, initial text=, accepting] (q1) {$1$};
\node[state]  (q2) [right of=q1, initial right, initial text=] {$2$};
\path[->] (q1) edge[loop, above] node{$ (Q \setminus F_2) \cup F_1$} (q1)
          (q1) edge[bend left=20] node[above]{$F_1$} (q2)
          (q2) edge[bend left=20] node[below]{$F_2 \setminus F_1$} (q1)
          (q2) edge[loop, above] node{$Q \setminus F_1$} (q2);
\end{tikzpicture}%
\caption{A backward-deterministic``counter'' \acs{NBA} to merge the acceptance sets $F_1$ and $F_2$ into a single acceptance set.
}
\label{fig:separated-counter}
\end{figure}

The transformation is based on the following Lemma.
\begin{lemma}
\label{lemma-potl-generalized-to-ordinary-Buechi-cond}
Let $ \mathcal{B} =
  ( \Sigma, \allowbreak
  M, \allowbreak
  Q, \allowbreak
  I, \allowbreak
  \mathbf{F}, \allowbreak
  \delta )$ 
be a generalized \ac{OPBA}, such that $|Q| = n$ and $|\mathbf{F}|= m \geq 2$.
There is an equivalent generalized \ac{OPBA} 
$ \mathcal{B}' =
  ( \Sigma, \allowbreak
  M, \allowbreak
  Q', \allowbreak
  I', \allowbreak
  \mathbf{F}', \allowbreak
  \delta )
$
such that $|Q| \leq 2n$ and $|\mathbf{F'}| = m-1$.
Furthermore, if $\mathcal{B}$ is separated (respectively complete),
then also $\mathcal{B}'$ is separated (complete).
\end{lemma}
\begin{proof}
We report the transformation, which has been originally
provided by \citet{CartonM03} for flat \acp{NBA}.
Let $\textbf{F} = \{ F_1, F_2, \dots F_m\}$.
The transformation is based on a cross-product of $\mathcal{B}$ with the \ac{NBA} $\mathcal{C}$ of Fig. \ref{fig:separated-counter}. 
$\mathcal{C}$ allows to merge the two acceptance sets $F_1$ and $F_2$ into a single acceptance set $F_2'$, 
while preserving the required properties. 
We define $\mathcal{B'}$ as follows:
\begin{itemize}
    \item $Q' = Q \times \{1,2\}$;
    \item $I' = I \times \{1,2 \}$;
    \item $\delta'$ is such that $\big((q,i), a, (q',j) \big) \in \delta'$, with $i,j \in \{1,2\}$, iff $(q,a,q') \in \delta$ 
    and $ \mathcal{C}$ has a transition $(i,C,j)$ with  $q' \in C$\footnote{Note that $q' \in \mathcal{C}$ is not a typo: each state uniquely determines the value of the counter at the previous step, mirroring the case of a deterministic counter.};
    \item $\mathbf{F}' = \{F'_2,F'_3, \dots, F'_m\}$, where $F'_2 = F_2 \times \{ 1 \}$, and, for $3 \leq k \leq m$, $F'_k = F_k \times \{1,2\}$. 
\end{itemize}
\end{proof}

We can transform a complete separated generalized \ac{OPBA} into a complete separated ordinary \ac{OPBA} 
by repeatedly applying Lemma \ref{lemma-potl-generalized-to-ordinary-Buechi-cond}, 
until $\mathcal{B}'$ has a single acceptance set. 
The algorithm yields an ordinary \ac{OPBA} with at most $n\cdot 2^{m-1}$ states.

\subsubsection{\acs{POTL} operators that yield non-separated \acp{OPBA}}
\label{sec:potl-construction-not-sep}
\begin{figure}
\begin{tikzpicture}
  [baseline=(current bounding box.center),
   every state/.style={minimum size=0pt, inner sep=2pt}, >=latex, font=\scriptsize,
   state/.style={draw, ellipse, inner sep=1pt}, >=latex, font=\scriptsize]
\node[state, initial by arrow, initial text=, initial below, accepting] (q0) {$\lcall$};
\node[state] (q1) [right=60pt of q0, initial below, initial text=, accepting] {$\lret$};
\node[state] (q2) [above=40pt of q0, accepting] {$\lcall, \ldback \lcall$};
\node[state] (q3) [above=40pt of q1, accepting] {$\lret, \ldback \lcall$};
\node[below=1pt of q0, xshift=-5pt] {$q_0$};
\node[below=1pt of q1, xshift=-5pt] {$q_1$};
\node[above=1pt of q2] {$q_2$};
\node[above=1pt of q3] {$q_3$};
\path[->] (q0) edge[double, loop left, left] node{$q_0, q_1, q_2$} (q0)
          (q0) edge[left] node{$\lcall$} (q2)
          (q0) edge[out=60, in=200, left] node{$\lcall$} (q3)
          (q2) edge[loop left, left] node{$\lcall$} (q22)
          (q2) edge[above] node{$\lcall$} (q3)
          (q3) edge[dashed, below] node{$\lret$} (q0)
          (q3) edge[dashed, right] node{$\lret$} (q1)
          (q1) edge[out=60, in=10, loop, right] node{$\lret$} (q1)
          (q1) edge[double, out=350, in=300, loop, right] node{$q_0, q_1, q_2$} (q1)
          (q1) edge[dashed, above] node{$\lret$} (q0)
          (q1) edge[out=200, in=340, below] node{$\lret$} (q0);
\end{tikzpicture}%
\begin{tikzpicture}
  [baseline=(current bounding box.center),
   state/.style={draw, ellipse, inner sep=1pt}, >=latex, font=\scriptsize,
   support/.style={decoration={snake, amplitude=1pt, segment length=3pt}, decorate}
  ]
\node[state, initial by arrow, initial text=, accepting] (n0) {$q_0, \#, \lcall$};
\node[state] (n1) [right=20pt of n0, accepting] {$q_2, \lcall, \lcall$};
\node[state, initial by arrow, initial text=, accepting] (n2) [below=20pt of n0] {$q_1, \#, \lret$};
\node[state] (n3) [right=20pt of n2, accepting] {$q_0, \lcall, \lcall$};
\path[->] (n0) edge[above] node{$\lcall$} (n1)
          (n0) edge[support, bend right] (n2)
          (n1) edge[loop above] node{$\lcall$} (n1)
          (n1) edge[support, bend right] (n3)
          (n2) edge[support, bend right] (n0)
          (n2) edge[support, loop below] (n2)
          (n3) edge[support, loop below] (n3)
          (n3) edge[bend right, right] node{$\lcall$} (n1);
\end{tikzpicture}
\caption{Left: \ac{OPBA} accepting the language denoted by formula $\neg \ldback \lcall$,
  built according to \cite{ChiariMPP23}.
  The alphabet includes only $\{\lcall, \lret\}$ and the corresponding subset of $M_\lcall$ (Fig.~\ref{fig:opm});
  states include only the current part, because the pending and in-stack parts are empty in this case.
  Right: trim support graph of the \ac{OPBA} on the left,
  in which push transitions are labeled with the symbol they read, to highlight the absence of backward determinism.}
\label{fig:non-sep-example}
\end{figure}

The construction of \cite{ChiariMPP23} yields \acp{OPBA} that are not separated
for all \acs{POTL} operators that are not included in \acs{POTLF}.
We illustrate this fact with the automaton generated for one of the simplest such formulas,
$\neg \ldback \lcall$.
The $\ldback$ operator is the past version of $\ldnext$:
a formula $\ldback \psi$ is true in positions such that $\psi$ holds in the \emph{previous} position,
and the label of the previous position is in the $\lessdot$ or $\doteq$ \ac{PR} with the current one%
---formally, $(w,i) \models \ldback \psi$ iff $(w,i-1) \models \psi$ and $(i-1) \lessdot i$ or $(i-1) \doteq i$.

Fig.~\ref{fig:non-sep-example} shows the resulting \ac{OPBA} and its support graph.
Clearly, the support graph is not backward deterministic:
node $(q_2, \lcall, \lcall)$ is reachable by a transition labeled with $\lcall$
from both $(q_0, \#, \lcall)$ and itself.
Moreover, node $(q_0, \#, \lcall)$ is reachable from both $(q_2, \lcall, \lcall)$ and itself
through the supports $q_2 \apush{\lcall} q_3 \ashift{\lret} q_0 \apop{q_2} q_0$ and
$q_0 \apush{\lcall} q_3 \ashift{\lret} q_0 \apop{q_0} q_0$,
which both read the string $\lcall \, \lret$.

The semantics of $\ldback$ is very similar to that of the \ac{LTL} back operator:
indeed, any construction for past \ac{LTL} based on the power set of the formula closure,
such as a past extension of the tableau by \citet{VardiW86},
yields automata that are not backward deterministic.
To see why this happens, recall that in the proof of Theorem~\ref{thm:potlf-separated-appendix}
we often used the fact that, given a transition between two states $\Phi \apush{a} \Psi$,
the contents of $\Psi$ and $a$ fully determine $\Phi$.
If we allow for past operators, however, this is not the case,
because whether or not a past operators appears in $\Phi$ depends on states reached by the automaton before $\Phi$.
For instance, if a computation reaches $q_3$ in the \ac{OPBA} of Fig.~\ref{fig:non-sep-example},
whether the previous state is $q_0$ or $q_2$ depends on whether $\lcall$ held in the state before.

For the same reason, all \acs{POTL} past operators,
including the past versions $\lcback{t}$ of the chain operators $\lcnext{t}$,
yield \ac{OPBA} whose support graph is not backward deterministic.

Furthermore, \acs{POTL} includes \emph{hierarchical} operators,
which hold in right or left chain contexts.
For instance, the \emph{upward hierarchical} next operator $\lhunext$ has the following semantics:
$(w,i) \models \lhunext \varphi$ iff
there exist a position $h < i$ s.t.\ $\chain(h,i)$ and $h \lessdot i$
and a position $j = \min\{ k \mid i < k \land \chain(h,k) \land h \lessdot k \}$
and $(w,j) \models \varphi$.
Even future operators of this kind yield automata that are not backward deterministic,
because they impose the existence of a past position $h$ in the word that serves as the other chain context
(we refer the reader to \cite{ChiariMPP23} for the semantics of other hierarchical operators).

\subsection{The Support Chain}
\label{sec:support-chain-proof}

We prove Theorem~\ref{thm:support-chain} by three separate lemmas.
First, we prove that $M_\mathcal{A}$ is a well-defined Markov chain.

In the following, given a \ac{pOPA} \emph{semi-configuration}
$c = (u, \alpha) \in Q \times \Gamma_\bot$, we define $\State{c} = u$.
\begin{lemma}
\label{lemma:support-is-mc}
$M_\mathcal{A}$ is a Markov chain.
\end{lemma}
\begin{proof}
We show that the probabilities assigned to outgoing edges of each state $Q_{M_\mathcal{A}}$ sum to 1.
Let $c = (u, \alpha) \in Q_{M_\mathcal{A}}$ and, w.l.o.g., $\alpha = [a, r]$
(the case with $\alpha = \bot$ is analogous).

If $a \doteq \Lambda(u)$, then only shift edges exit from $c$.
Thus,
\[
\sum_{d \in \mathcal{C}} \delta_{M_\mathcal{A}}(c, d) =
\smashoperator{\sum_{d \mid c \sshift d}} \delta_{\mathit{shift}} (u, \State{d}) \frac{\nex{d}}{\nex{c}} =
\frac{\sum_{d \mid c \sshift d} \delta_{\mathit{shift}} (u, \State{d}) \nex{d}}{\nex{c}}.
\]
Since all $c$'s outgoing edges are shift edges,
we have $\nex{c} = \sum_{d \mid c \sshift d} \delta_{\mathit{shift}} (u, \State{d}) \nex{d}$,
so $\sum_{d \in \mathcal{C}} \delta_{M_\mathcal{A}}(c, d) = 1$.

If $a \lessdot \Lambda(u)$, then $c$'s outgoing edges can be push or support edges.
The probability of $c$ being pending (i.e., part of an open support)
$\nex{c}$ is the sum of the probability of reaching a new semi-configuration $d$ by following a push edge,
and that $d$ is pending ($\nex{d}$):
\[
\mathcal{P}_\mathit{push}(c) = \smashoperator{\sum_{d \mid c \spush d}}
  \delta_\mathit{push}(u, \State{d}) \nex{d}
\]
and the probability of reaching a new pending semi-configuration $d$ through a support edge.
The latter is the probability of
\begin{enumerate}
\item reaching a new semi-configuration $d' = (v', [\Lambda(u), u])$ through a push edge;
\item that $d'$ is part of a closed support, i.e., that its stack symbol is eventually popped,
reaching the new semi-configuration $d$ ($\pvar{v'}{[\Lambda(u), u]}{\State{d}}$);
\item and that $d$ is pending ($\nex{d}$):
\end{enumerate}
\[
\mathcal{P}_\mathit{supp}(c) =
  \sum_{d \mid c \ssupp d}
  \sum_{v' \in V}
  \delta_\mathit{push}(u, v') \pvar{v'}{[\Lambda(u), u]}{\State{d}} \nex{d}
\]
where $V = \{v' \mid c \spush (v', [\Lambda(u), u])\}$.

Hence $\nex{c} = \mathcal{P}_\mathit{push}(c) + \mathcal{P}_\mathit{supp}(c)$:
if we divide both sides by $\nex{c}$, we obtain
\begin{align*}
1 &=
\smashoperator{\sum_{d \mid c \spush d}}
  \frac{\delta_\mathit{push}(u, \State{d}) \nex{d}}{\nex{c}}
  +
  \sum_{d \mid c \ssupp d}
  \sum_{v' \in V}
  \frac{\delta_\mathit{push}(u, v') \pvar{v'}{[\Lambda(u), u]}{\State{d}} \nex{d}}{\nex{c}} \\
&= \sum_{d \in \mathcal{C}} \delta_{M_\mathcal{A}}(c, d).
\end{align*}

We need not consider the case $a \gtrdot \Lambda(u)$,
because such a node would only lead to pop moves,
and hence would have $\nex{c} = 0$ and would not be part of $M_\mathcal{A}$.
\end{proof}

The rest of Theorem~\ref{thm:support-chain}'s claim is proved
in Lemmas~\ref{lemma:support-chain-zero} and \ref{lemma:support-chain-prob}.
Recall that $(\Omega, \mathcal{F}, P)$ is the probability space of $\Delta(\mathcal{A})$,
and $(\Omega', \mathcal{F}', P')$ is the one of $M_\mathcal{A}$.

\begin{lemma}
\label{lemma:support-chain-zero}
$P(\Omega \setminus \sigma^{-1}(\Omega')) = 0$
\end{lemma}
\begin{proof}
Let $\hat{\Omega} = \Omega \setminus \sigma^{-1}(\Omega')$.
Any $\rho \in \hat{\Omega}$ is such that $\sigma(\rho) \not\in \Omega'$.
We analyze the earliest semi-configuration in $\sigma(\rho)$ that is not in $M_\mathcal{A}$.
Let $\tau' = c_0 \dots c_i c_{i+1}$ be the shortest prefix of $\sigma(\rho)$
such that $c_i \in Q_{M_\mathcal{A}}$ but $c_{i+1} \not\in Q_{M_\mathcal{A}}$.
We call $T$ the set of such prefixes for all $\rho \in \hat{\Omega}$.
We define $\hat{\Omega}(\tau') = \{\rho \in \Omega \mid \sigma(\rho) = \tau'\tau'', \tau'' \in \mathcal{C}^{\omega}\}$
as the set of runs of $\mathcal{A}$ such that $\tau'$ is a prefix of their image through $\sigma$.

We prove that for any such $\tau'$ we have $P(\hat{\Omega}(\tau')) = 0$.
Let $c_i = (u_i, \alpha_i)$ and $c_{i+1} = (u_{i+1}, \alpha_{i+1})$.

If $\symb{\alpha_i} \doteq \Lambda(u_i)$, then $c_i \sshift c_{i+1}$.
In a run $\rho \in \hat{\Omega}(\tau')$, $u_i$ and $u_{i+1}$ are part of an open support,
i.e., $\rho$ reaches two consecutive configurations $(u_i, \alpha_iA) \allowbreak (u_{i+1}, \alpha_{i+1}A)$,
and $\alpha_{i+1}$ is never popped, but at most updated by shift moves.
Indeed, if $\alpha_{i+1}$ was popped, then $u_i$ and $u_{i+1}$ would be part of a closed support,
and a suffix of $\tau'$ including $c_i$ and $c_{i+1}$ would be replaced by a support edge in $\sigma(\rho)$,
contradicting $\rho \in \hat{\Omega}(\tau')$.
However, since $c_{i+1} \not\in Q_{M_\mathcal{A}}$, we have $\nex{c_{i+1}} = 0$,
and the probability measure of such trajectories is also 0, hence $P(\hat{\Omega}(\tau')) = 0$.

If $\symb{\alpha_i} \lessdot \Lambda(u_i)$, then $c_i \spush c_{i+1}$,
$c_i \ssupp c_{i+1}$, or both.
We split $\hat{\Omega}(\tau') = \hat{\Omega}_\mathit{push}(\tau') \cup \hat{\Omega}_\mathit{supp}(\tau')$ into two sets,
depending on whether $c_i$ and $c_{i+1}$ are linked by a push move or a support.
We have $P(\hat{\Omega}(\tau')) = P(\hat{\Omega}_\mathit{push}(\tau')) + P(\hat{\Omega}_\mathit{supp}(\tau'))$.

If $c_i \spush c_{i+1}$, the argument is the same as for the shift case,
except $c_{i+1}$ is part of a new open support, starting with the push move.
Again, the probability that this new support is open is 0,
hence $\hat{\Omega}_\mathit{push}(\tau') = 0$,

If $c_i \ssupp c_{i+1}$, then a run $\rho \in \hat{\Omega}(\tau')$ reaches a sequence of configurations
$(u_i, \alpha_iA) \allowbreak (u'_i, [\Lambda(u_i), u_i] \alpha_i A) \allowbreak \dots \allowbreak (u^{(k)}_i, [\Lambda(u^{(j)}_i), u_i], \alpha_i A) \allowbreak (u_{i+1}, \alpha_{i+1}A)$
with $1 \leq j < k$, in which $\alpha_i A$ always remains on the stack.
As discussed in the case for shift edges, $u_i$ and $u_{i+1}$ are part of an open support,
or they would not appear in $\sigma(\rho)$.
Since $\nex{c_{i+1}} = 0$, we have $P(\hat{\Omega}_\mathit{supp}(\tau')) = 0$, and $P(\hat{\Omega}(\tau')) = 0$.

Note that the support graph contains no pop edges (the supports they close are replaced by support edges),
so we need not consider the case $\symb{\alpha_i} \gtrdot \Lambda(u_i)$.

In conclusion, $\hat{\Omega} = \cup_{\tau' \in T} \hat{\Omega}(\tau')$ and, since $T \subseteq \mathcal{C}^*$,
$T$ is countable, and \[P(\hat{\Omega}) = \sum_{\tau' \in T} P(\hat{\Omega}(\tau')) = 0.\]
\end{proof}

\begin{lemma}
\label{lemma:support-chain-prob}
Given a set $F' \in \mathcal{F}'$, we have $F = \sigma^{-1}(F') \in \mathcal{F}$ and $P'(F') = P(F)$.
\end{lemma}
\begin{proof}
Let $\tau' = c_0 \dots c_n \in Q_{M_\mathcal{A}}^*$ be such that
$C'(\tau') = \{\tau \in \Omega' \mid \tau = \tau' \tau'', \tau'' \in Q_{M_\mathcal{A}}^\omega\} \in \mathcal{F}'$
is a cylinder set of $M_\mathcal{A}$.
We prove the claim for all such cylinder sets by induction on $n \geq 0$.

For the case $n = 0$ we have $\tau' = \varepsilon$ and $C'(\varepsilon) = \Omega'$.
Obviously, $\sigma^{-1}(\Omega') = \Omega \setminus (\Omega \setminus \sigma^{-1}(\Omega'))$,
and by Lemma~\ref{lemma:support-chain-zero} we have
$P(\sigma^{-1}(\Omega')) = P(\Omega) - P(\Omega \setminus \sigma^{-1}(\Omega')) = 1 - 0 = 1$.
Since $\Omega \setminus \sigma^{-1}(\Omega') \in \mathcal{F}$,
its complement $\sigma^{-1}(\Omega') \in \mathcal{F}$ is also in $\mathcal{F}$.

Let $C(\tau') = \sigma^{-1}(C'(\tau'))$.
We prove that $P(C(\tau' c_{n+1})) = P'(C'(\tau' c_{n+1}))$ starting from the inductive hypothesis
that the lemma claim holds for $\tau'$.

We define the event that the $i$-th semi-configuration of the $\sigma$-image of a run is $c$ as
$J_{i,c} = \{\rho \in \Omega \mid \sigma(\rho) = c_0 \dots c_i \dots \text{ with } c_i = c\}$.

Let $c_n = (u_n, \alpha_n)$ and $c_{n+1} = (u_{n+1}, \alpha_{n+1})$.
Note that, in general, if $\tau' c_{n+1}$ is a prefix of a trajectory in $M_\mathcal{A}$,
then any $\rho \in C(\tau')$ reaches two configurations
$(u_n, \alpha_nA), \allowbreak (u_{n+1}, \alpha_{n+1}A)$ either consecutive
(if $c_n$ and $c_{n+1}$ are linked by a push or a shift edge)
or separated by a chain support.
In any case, $\alpha_nA$ remains on the stack forever after $(u_n, \alpha_nA)$
(the label in $\alpha_n$ is at most updated by shift moves, but never popped).
Thus, $J_{n+1, c_{n+1}}$ only depends on $u_n$ and $\symb{\alpha_n}$,
and the probability of reaching $c_{n+1}$ through a push or a shift move, or a chain support.
Therefore, this process has the Markov property, and
$P(J_{n+1, c_{n+1}} \mid C(\tau')) = P(J_{1, c_{n+1}} \mid C(c_n))$.

Moreover, given $c \in Q_{M_\mathcal{A}}$ we define as $[c \uparrow]$
the event that $c$ is pending ($\nex{c}$ is the probability of this event).
Thanks to the Markov property of $\Delta(\mathcal{A})$,
this event is also independent of the stack contents and previous history of a run that
reaches a configuration corresponding to semi-configuration $c$.

In general, we can write the following:
\begin{align*}
&P(J_{n+1, c_{n+1}} \mid C(\tau')) & \\
&= P(J_{1, c_{n+1}} \mid C(c_n)) & \text{Markov property} \\
&= P(J_{1, c_{n+1}} \mid [c_n \uparrow]) & \text{we assume the run starts at $c_n$} \\
&= P(J_{1, c_{n+1}}) P([c_n \uparrow] \mid J_{1, c_{n+1}}) / P([c_n \uparrow]) & \text{Bayes' Theorem} \\
&= P(J_{1, c_{n+1}}) P([c_{n+1} \uparrow]) / P([c_n \uparrow]) & \\
&= P(J_{1, c_{n+1}}) \nex{c_{n+1}} / \nex{c_n} &
\end{align*}
where the last step occurs because,
once the transition from $c_n$ to $c_{n+1}$ is taken for granted,
$c_n$ is pending iff $c_{n+1}$ is.
To determine $P(J_{1, c_{n+1}})$, we need to study each type of edge
in the support graph separately.

\smallskip
\noindent $\symb{\alpha_n} \doteq \Lambda(u_n)$.
Then $c_n \sshift c_{n+1}$, and $P(J_{1, c_{n+1}}) = \delta_\mathit{shift}(u_n)(u_{n+1})$.
Thus,
\begin{align*}
P(C(\tau' c_{n+1}))
&= P(C(\tau')) \delta_\mathit{shift}(u_n)(u_{n+1}) \frac{\nex{c_{n+1}}}{\nex{c_n}} \\
&= P'(C'(\tau')) \delta_\mathit{shift}(u_n)(u_{n+1}) \frac{\nex{c_{n+1}}}{\nex{c_n}} \\
&= P'(C'(\tau' c_{n+1})),
\end{align*}
where the second equality is due to the inductive hypothesis,
and the third to the Markov property of $M_\mathcal{A}$
and the way $\delta_{M_\mathcal{A}}$ is defined.

\smallskip
\noindent $\symb{\alpha_n} \lessdot \Lambda(u_n)$.
Then $c_n \spush c_{n+1}$, $c_n \ssupp c_{n+1}$, or both.
The case $c_n \spush c_{n+1}$ is analogous to $c_n \sshift c_{n+1}$.

We analyze the case in which $c_n$ and $c_{n+1}$ are linked by a support edge.
Here $P(J_{1, c_{n+1}})$ is the probability that a closed chain support occurs
from $(u_n, \alpha_n\bot)$ to $(u_{n+1}, \alpha_{n+1}\bot)$.
Since the support edge may represent multiple supports,
this probability is a sum on all closed supports that can start from $u_n$,
depending on the push move that fires from $u_n$.
We denote as $S(u_n, v, u_{n+1})$ the event that a run contains a closed support
from $(u_n, \alpha_n\bot) \allowbreak (v, [\Lambda(u_n), u_n] \alpha_n\bot)$ to $(u_{n+1}, \alpha_{n+1}\bot)$.
$S(u_n, v, u_{n+1})$ is then the conjunction of the event that a push move
occurs from $u_n$ to $v$, and that the support is closed or,
equivalently, that $[\Lambda(u_n), u_n]$ is popped, denoted $[v, [\Lambda(u_n), u_n] \mid u_{n+1}]$.

Let $V = \{v \in Q \mid (u_n, \alpha_n) \spush (v, [\Lambda(u), u])\}$. We have:
\begin{align*}
P(J_{1, c_{n+1}})
&= P(\cup_{v \in V} S(u_n, v, u_{n+1})) \\
&= \sum_{v \in V} P([(u_n, \alpha_n\bot) \gedge (v, [\Lambda(u_n), u_n] \alpha_n\bot)] \cup [v, [\Lambda(u_n), u_n] \mid u_{n+1}]) \\
&= \sum_{v \in V} P([(u_n, \alpha_n\bot) \gedge (v, [\Lambda(u_n), u_n] \alpha_n\bot)]) P([v, [\Lambda(u_n), u_n] \mid u_{n+1}]) \\
&= \sum_{v \in V} \delta_\mathit{push}(u_n)(v) \pvar{v}{[\Lambda(u_n), u_n]}{u_{n+1}}.
\end{align*}
Thus,
$P(C(\tau' c_{n+1})) = P(C(\tau')) \sum_{v \in V} \delta_\mathit{push}(u_n)(v) \pvar{v}{[\Lambda(u_n), u_n]}{u_{n+1}}$,
and it is easy to see that, due to the definition of $\delta_{M_\mathcal{A}}$,
$P(C(\tau' c_{n+1})) = P'(C'(\tau' c_{n+1}))$.

If $c_n$ and $c_{n+1}$ are linked by both a push and a support edge,
$P(C(\tau' c_{n+1}))$ is the sum of the probabilities derived for each case.
\end{proof}

\subsection{Hardness of Model Checking}
\label{sec:hardness-proof}

\begin{proof}[Proof of Lemma~\ref{lemma:exptime-hardness}]
We reduce the acceptance problem of linear-space-bounded alternating Turing machines
to qualitative \acs{POTLF} model checking.
The construction uses ideas from the one given in the proof of \cite[Theorem 33]{EtessamiY12},
but differs due to the use of \acp{pOPA} instead of \acp{RMC},
and \ac{POTLF} instead of \acs{LTL}.

There exists a linear-space-bounded one-tape alternating Turing machine $\mathcal{M}$
for which the acceptance decision problem is \textsc{exptime}-complete w.r.t.\ a given input of length $n$.
Let $\mathcal{M} = (Q, \Gamma, \delta, q_0, g)$, where
$Q$ is a finite set of control states,
$\Gamma$ is a finite set of tape symbols, disjoint from $Q$,
$\delta : Q \times \Gamma \rightarrow \powset{Q \times \Gamma \times \{L, R\}}$ is the transition function
such that $|\delta(q, X)|$ is either 2 or 0 for each $(q, X) \in Q \times \Gamma$
(we assume w.l.o.g.\ that the machine has only two moves available in any state except halting states),
$q_0$ is the initial control state, and
$g : Q \rightarrow \{\exists, \forall, \textit{acc}, \textit{rej}\}$
is a function that classifies each state between existential, universal, accepting and rejecting
(the last two being always halting states).
Configurations are strings in $\Gamma^* (Q \times \Gamma) \Gamma^*$
such that the $i$-th symbol is $(q, X) \in Q \times \Gamma$ if
the machine's control is in state $q$ and its head is at position $i$ of the tape, which contains symbol $X$.
Remaining symbols only identify the tape content at the corresponding position.
A computation is a sequence of configurations such that the first one is
$(q_0, w_0) w_1 \dots w_n$, where $w = w_0 w_1 \dots w_n$ is the input word,
and the remaining configurations are obtained according to usual transition rules for Turing machines.
All computations end in a halting state.

For $\mathcal{M}$ to accept a word, the set of computations bifurcating from each universal configuration
must all end up in an accepting state.
We can see such a set as a \emph{computation tree} whose nodes are $\mathcal{M}$'s configurations,
and universal nodes have two children, existential nodes have one, and halting nodes are leaves.
If the input word is accepted, all leaves are accepting states.
For a given input $w$, we build a \ac{pOPA} $A_\mathcal{M}$ and a \acs{POTLF} formula $\varphi_\mathcal{M}$
s.t.\ if $\varphi_\mathcal{M}$ does not hold almost surely on $A_\mathcal{M}$,
then there exists an accepting computation of $\mathcal{M}$ on $w$.

\begin{figure}
\centering
\(
\begin{array}{r | c c c c c}
         & \lcall   & \lret   & \lstm \\
\hline
\lcall   & \lessdot & \doteq  & \doteq \\
\lret    & \gtrdot  & \gtrdot & \gtrdot \\
\lstm    & \lessdot  & \doteq & \doteq \\
\end{array}
\)
\caption{\ac{OPM} $M_A$.}
\label{fig:opm-ma}
\end{figure}

We define $A_\mathcal{M} = (\Sigma_A, \allowbreak M_A, \allowbreak U, \allowbreak u_0,
\allowbreak \delta^A, \allowbreak \Lambda)$ where
$\Sigma_A = \{\lcall, \lret, \lstm\}$, $M_A$ is reported in Fig.~\ref{fig:opm-ma},
and we describe states, transitions and labels in the following.
Note that, since we are interested in qualitative model checking,
the exact values of probabilities that we assign to transitions do not matter,
as long as they are positive and form well-defined distributions.
For each pair $(q, X) \in Q \times \Gamma$,
$U$ contains two states $u[q, X]$ and $v[q, X]$, both labeled with $\lstm$.
The distinguished initial state $u_0$ is labeled with $\lcall$,
and has a unique push transition targeting $u[q_0, w_0]$ with probability 1.
The rest of the \ac{pOPA} is made so that each computation will eventually reach
$v[q_0, w_0]$ with a topmost stack symbol containing $\lret$.
Then, a pop move links $v[q_0, w_0]$ to another distinguished state $v_0$ with probability 1,
and is a sink state $v_0$ labeled with $\lstm$ and has a shift self-loop.
The computation starting from $u[q_0, w_0]$ is intended to represent a depth-first traversal
of an accepting computation tree of $\mathcal{M}$.
We now describe the rest of $A_\mathcal{M}$ while describing such a computation.

Let $(q, X) \in Q \times \Gamma$ with $g(q) = \exists$.
$q$ represents the current state of $\mathcal{M}$ in the computation,
and $X$ the symbol in the current tape position, say $i$.
The computation reaches $u[q, X]$ from a $\lcall$ node $x$ with a push move.
The topmost stack symbol is $[\lcall, x]$ and $\Lambda(u[q, X]) = \lstm$,
so the next move will be a shift.
Let $(p_j, Y_j, D_j) \in \delta(q, X)$ for $j \in \{1,2\}$.
For each $j$ and $Z \in \Gamma$,
a shift move links $u[q, X]$ to a node $r[q, X, j, Z] \in U$ labeled with $\lstm$
that has a shift self-loop, and a shift move targeting another node $r'[q, X, j, Z] \in U$ labeled with $\lcall$.
By taking one of the shift moves starting from $u[q, X]$,
$A_\mathcal{M}$ will ``guess'' both the move $j$ chosen by $\mathcal{M}$ in the existential state $q$,
and the symbol $Z$ in the next tape position (hence, either $i+1$ if $D_i = R$ or $i-1$ if $D_i = L$). 
Consistency of such guesses will be enforced by $\varphi$,
including the fact that the tape initially contains the input word $w$.
Once it reaches $r[q, X, j, Z]$, $A_\mathcal{M}$ loops $i-1$ times around it,
and finally moves to $r'[q, X, j, Z]$, which is labeled with $\lcall$.
A push move links $r'[q, X, j, Z]$ to $u[p_j, Z]$.
The \ac{pOPA} is inductively built in such a way that any run eventually reaches
a new state $s[q, X, j, Z]$ with a move popping stack symbol $[\lcall, r'[q, X, j, Z]]$.
The topmost stack symbol is now again $[\lcall, x]$, and $s[q, X, j, Z]$ is labeled with $\lstm$,
so the next move is a shift.
State $s[q, X, j, Z]$ has indeed a shift self-loop,
and a shift move targeting another state $s'[q,X]$ labeled with $\lret$.
To simulate $\mathcal{M}$'s computation, the run in $A_\mathcal{M}$ keeps track, again,
of the current head position by looping $j-1$ times around $s[q, X, j, Z]$ (enforced by $\varphi$),
and then moves to $s'[q,X]$.
The latter is linked to $v[q, X]$ by a shift transition.
After this shift move, the topmost stack symbol is $[\lret, x]$,
ready to be popped (recall that $v[q, X]$ is labeled with $\lstm$ and $\lret \gtrdot \lstm$).
From $v[q, X]$ start many pop transitions, one for each possible state $x$ in the symbol $[\lcall, x]$
that was pushed onto the stack when reaching $u[q, X]$.
In particular, for each non-halting $q' \in Q$, $X' \in \Gamma$, $j' \in \{1,2\}$, and $Z' \in \Gamma$,
we have $(v[q, X], r'[q', X', j', Z'], s[q', X']) \in \delta^A_\textit{pop}$.

We now consider the case $g(q) = \forall$.
As in the $\exists$ case, the run reaches $u[q, X]$ from a $\lcall$ node $x$ with a push move,
so the topmost stack symbol is $[\lcall, x]$.
Let $(p_j, Y_j, D_j) \in \delta(q, X)$ for $j \in \{1,2\}$.
This time, the states of $A_\mathcal{M}$ representing the two moves of $\mathcal{M}$ are posed ``in series''.
This way, the \ac{pOPA} run will visit both of them, and an accepting run of $A_\mathcal{M}$
will encode both accepting computations spawning from a universal state of $\mathcal{M}$.
Hence, a shift move links $u[q, X]$ to a state $r[q, X, 1, Z]$ for each $Z \in \Gamma$
that has a shift self-loop and a shift move targeting $r'[q, X, 1, Z]$.
The run will loop around $r[q, X, 1, Z]$ for $j-1$ times (where $j$ is the current tape position)
before reaching $r'[q, X, 1, Z]$.
The latter state is labeled with $\lcall$ and a push move links it to $u[p_1, Z]$.
Once the run pops the pushed stack symbol $[\lcall, r'[q, X, 1, Z]]$,
it reaches a state $s[q, X, 1, Z]$ labeled with $\lstm$,
and it loops around it through a shift self-loop for $j-1$ times.
State $s[q, X, 1, Z]$ is linked by a shift move to a state $s'[q, X, 2, Z']$ for each $Z' \in \Gamma$,
so here the run guesses the tape symbol in the next tape position
and visits the computation tree spawned by the second universal move of $q$.
State $s'[q, X, 2, Z']$ is labeled with $\lcall$ and linked to $u[p_2, Z']$ by a push move.
When the run pops stack symbol $[\lcall, s'[q, X, 2, Z']]$,
it reaches a state $t[q, X, 2, Z'] \in U$ labeled with $\lstm$.
This state has a shift self-loop, that is meant to be visited for $j-1$ times,
and another shift move links it to another state $t[q, X]$, labeled with $\lret$.
A shift move then links $t[q, X]$ to $v[q,X]$ from which,
like in the existential case, starts a pop transition for each possible state $x$ in the stack symbol $[\lcall, x]$. 

Finally, let $g(q) \in \{\textit{acc}, \textit{rej}\}$.
Such a state would induce a leaf in $\mathcal{M}$'s computation tree.
Thus, the run in $A_\mathcal{M}$ reaches $u[q, X]$ from a $\lcall$ node $x$ with a push move,
and does not do any more push moves before popping $[\lcall, x]$.
A shift move links $u[q, X]$ to another state $r[q, X]$ labeled with $\lstm$,
that has a shift self loop, meant to be taken $j-1$ times,
and another shift move links it to node $s[q, x]$ labeled with $\lret$.
A shift move links $s[q, x]$ to $v[q, X]$, from which several pop moves start
(one for each possible $x$, as in the other cases).

Thus, at any given point of the run of $A_\mathcal{M}$,
the stack contains all moves taken by $\mathcal{M}$ in the path in its computation tree
going from the node encoded in the current configuration of $A_\mathcal{M}$ to the root.
Such moves, together with the number of times the run of $A_\mathcal{M}$ stays in self-loops,
can be used to reconstruct the whole computation tree of $\mathcal{M}$.

To make sure that $A_\mathcal{M}$'s run actually encodes such a computation tree, however,
We need to craft $\varphi$ so that a run of $A_\mathcal{M}$ that violates it with positive probability
encodes an accepting computation tree of $\mathcal{M}$.
$\varphi$ can be written as $\varphi = \neg \xi$, where $\xi$ describes a run of $A_\mathcal{M}$
that describes an accepting computation tree of $\mathcal{M}$.
The proof of \cite[Theorem 33]{EtessamiY12} builds $\xi$ as an \acs{LTL} formula
that has states in $A_\mathcal{M}$ as propositions,
and is a conjunction of several requirements.
So for instance, the fact that the run eventually reaches $v_0$ without ever touching any rejecting state
can be expressed as $\lluntil{(\land_{q,X \in Q \times \Gamma \mid g(q) = \textit{rej}} \neg u[q,X])}{v_0}$.
For any positive integer $i$, formula 
$\psi_i = \neg r \land (\land_{k=1}^i \lnext^k r) \land (\neg \lnext^{i+1} r)$
with $r = \lor_{q, X, j, Z \in Q \times \Gamma \times \{1,2\} \times \Gamma} r[q, X, j, Z]$
states that the run visits moves to a $r$ node in the next steps, and loops around it for $i$ times.
$\psi_i$ can be used to make sure that head positions are always correct according to the transitions rules
of the Turing machine, including the fact that its computation starts in position 1.
It can also be used to state that the first time each tape symbol is read,
it contains the appropriate character of the input word $w$.

Instead of rewriting each one of these formulas in \ac{POTLF}, we give a translation for \acs{LTL} operators,
that can be used to translate the formulas given in \cite[Theorem 33]{EtessamiY12}.
Such formulas then require minimal adaptations for the structure of the \ac{pOPA} that we build.
We have:
\begin{align*}
&\lnext \phi \equiv \ldnext \phi \lor \lunext \phi &
&\lluntil{\phi}{\psi} \equiv \lcuuntil{(\lglob{d} \phi)}{(\psi \lor (\lcduntil{(\phi \land \zeta)}{\psi})} \\
\shortintertext{where}
&\lglob{d} \phi \equiv \neg (\lcduntil{\top}{\neg \psi}) &
&\zeta \equiv \ldnext (\neg \lret \land \lglob{d} \neg \psi) \implies \lglob{d} \phi
\end{align*}
The translation for the $\lnext$ operator is trivial.
The one for the until operator works by concatenating an upward summary until which climbs up the $\chain$ relation
with a downward until, which climbs it down.
The left-hand-sides of both untils make sure that $\phi$ holds in all chain bodies skipped by the summary path.
Formula $\zeta$ states that is a chain body is not empty ($\ldnext \neg \lret$) and $\psi$ never holds in it,
then $\phi$ must hold in all of its positions.
We use it to distinguish chain bodies that are completely skipped by the downward summary path
from those that are entered (because $\psi$ holds in there).

Note that this translation only holds because \ac{OPM} $M_A$ from Fig.~\ref{fig:opm-ma}
induces a $\chain$ relation that is never one-to-many (but actually, always one-to-one).
Otherwise, the translation given in \cite[Section 3.4.1]{ChiariMP21b} is needed,
but it also contains past operators not included in \ac{POTLF}.
\end{proof}

\end{document}